%% file: main.tex
\documentclass[12pt]{article}

\usepackage[margin=1.25in]{geometry}

\usepackage{setspace}
\usepackage{amsmath,amssymb,amsfonts,amsthm,xcolor,url,array}
\usepackage{bm,natbib,algorithmic,accents}
\usepackage[ruled]{algorithm2e}
\usepackage{booktabs}
\usepackage{comment,soul}
\usepackage{multirow}
\usepackage{float}
\usepackage{graphicx}
\usepackage{subcaption}
\usepackage{enumitem}
\usepackage{siunitx}
\usepackage{todonotes}

\usepackage{hyperref}

\graphicspath{{figures/}}

\newcolumntype{L}[1]{>{\raggedright\let\newline\\\arraybackslash\hspace{0pt}}m{#1}}
\newcolumntype{C}[1]{>{\centering\let\newline\\\arraybackslash\hspace{0pt}}m{#1}}
\newcolumntype{R}[1]{>{\raggedleft\let\newline\\\arraybackslash\hspace{0pt}}m{#1}}

\newtheoremstyle{nopunct}
{\topsep}
{\topsep}
{\itshape}
{}
{\bfseries}
{}
{ }
{\thmname{#1}~\thmnumber{#2} \thmnote{(#3)}}

\theoremstyle{nopunct}
\newtheorem{lem}{Lemma}
\newtheorem{thm}{Theorem}
\newtheorem{prop}{Proposition}

\theoremstyle{definition}

\newtheorem{assumption}{Assumption}

\theoremstyle{remark}

\theoremstyle{plain}
\newtheorem*{assumptiona*}{\assumptionletter}
\makeatletter
\providecommand{\assumptionletter}{}

\makeatother

\makeatletter
\renewcommand{\@algocf@capt@plain}{above}
\makeatother

\definecolor{darkred}{RGB}{144,0,0}
\definecolor{darkblue}{RGB}{0,0,144}

\input{commands.tex}

\newcommand{\citeay}[1]{\citeauthor{{#1}}, \citeyear{{#1}}}

\SetKwInOut{KwFinOut}{Final output}

\begin{document}

\title{Empirical Bayes Estimation in Heterogeneous Coefficient Panel Models}
\author{\textsc{Myunghyun Song}\thanks{Department of Economics, Columbia University} \and \textsc{Sokbae Lee}\thanks{Department of Economics, Columbia University and Institute for Fiscal Studies} \and \textsc{Serena Ng}\thanks{Department of Economics, Columbia University and NBER\newline
We thank Jiaying Gu, Roger Koenker, Soonwoo Kwon, Bodhisattva Sen, Kaizheng Wang, and 
seminar participants at Indiana University, the University of Pennsylvania, the University of Toronto, and Yale University for helpful and encouraging comments. 
The third  author would like to thank the National Science Foundation for financial support  (SES: 2018369).}}

\date{November 11, 2026}

\maketitle

\thispagestyle{empty}
\setcounter{page}{0}

\begin{abstract}
We develop an empirical Bayes (EB) G-modeling framework for short-panel linear models with  nonparametric prior for the random intercepts, slopes, dynamics, and non-spherical error variances. We establish identification and consistency of the nonparametric maximum likelihood estimator (NPMLE) under general conditions, and provide low-level sufficient conditions for several models of empirical interest. Conditions for regret consistency of the EB estimators are also established. The NPMLE is computed using a Wasserstein-Fisher-Rao gradient flow algorithm adapted to panel regressions. Using data from the Panel Study of Income Dynamics, we find that the slope coefficient for potential experience is substantially heterogeneous and negatively correlated with the random intercept, and that error variances and autoregressive coefficients vary considerably across individuals. The EB estimates reduce mean squared prediction errors relative to individual maximum likelihood estimates.
\end{abstract}

\newpage

\section{Introduction}

Understanding differences in individual behavior is the primary goal of many economic analyses. Why do workers with similar experience exhibit different earnings profiles? Why do economic agents with similar observed characteristics respond differently to economic shocks and policy interventions? A challenge for empirical researchers in addressing these questions is the presence of individual-level differences that affect behavior but are not captured by observed variables, commonly referred to as unobserved heterogeneity. While it is not hard to allow for unit-specific intercepts in linear regressions, entertaining heterogeneous slope parameters and innovation variances is more challenging, and in a frequentist setting, this usually requires parametric assumptions to gain tractability at the cost of robustness. 

This paper considers a setting in which we have data for outcome, a small number of  covariates, and many individual characteristics. We analyze data  $\{ (Y_{it}, X_{2,it}): t=1,\ldots,T \}_{i=1}^N$ that are assumed to be independent and identically distributed (i.i.d.) across individuals, and the data are possibly residuals from  preliminary regressions that remove the effect of individual characteristics and time dummies. We propose a framework for short-$T$ linear panel models with four types of parameter heterogeneity. In particular, we consider the dynamic linear panel model with an exogenous covariate $X_{2,it}$ and AR(1) errors:
\begin{subequations}\label{def:model:hivxd}
\begin{align}
Y_{it} &= a_i + b_i X_{2,it} + u_{it}, \\
u_{it} &= \rho_i u_{it-1} + \sigma_i e_{it}, 
\quad e_{it} \overset{\mathrm{i.i.d.}}{\sim} \mathcal{N}(0,1), \qquad t \ge 1, 
\label{def:model:hivxd:2}
\end{align}
\end{subequations}
where 
$a_i \in \mathbb{R}$ is a random intercept, 
$b_i \in \mathbb{R}$ a random slope, 
$\rho_i \in (-1,1)$ a random autoregressive coefficient,  
$\sigma_i^2 > 0$ a random variance,  
and $u_{i0}$ follows the stationary distribution $\mathcal{N}(0,\sigma_i^2/(1-\rho_i^2))$ conditional on $(\sigma_i^2,\rho_i)$. We refer to the model in \eqref{def:model:hivxd} as HIVDX since it admits heterogeneity in four dimensions: intercepts ($I$), error variances ($V$), dynamics ($D$) through the autoregressive coefficient, and the coefficient on the exogenous covariate $X_2$. The HIVDX includes the widely studied Gaussian location-scale model $Y_{it} = a_i + \sigma_i e_{it}$, as well as models that impose $\rho_i = \rho$ and $b_i = b$ as special cases. The model can be written as an ADL(1,1) model 
\[
Y_{it} = a_i(1-\rho_i) + \rho_i Y_{it-1} + b_i X_{2,it} - \rho_i b_i X_{2,it-1} + \sigma_i e_{it}.
\]
Unlike an unrestricted ADL(1,1) model,  the unconditional mean of $Y_{it}$ in the HIVDX model does not depend on the dynamic parameter $\rho_i$.

 We analyze the HIVDX model and the special cases in Section~\ref{sec:model} using a general framework that partitions the random parameters into two sets: regression coefficients $\beta_i = (a_i, b_i)$ of dimension $d_\beta = 2$, and covariance parameters $\delta_i = (\sigma_i^2, \rho_i)$ of dimension $d_\delta = 2$. The parameters can be estimated separately for each $i$ by maximum likelihood. Instead, we assume that the random parameters $\theta_i \equiv (\beta_i, \delta_i)$ are drawn from a true but unknown common distribution $G_*$. Viewing $G_*$ as a prior, the optimal decision rule for each unit under compound squared-error loss is the posterior mean
$\thetapm_i := \mathbb{E}_{G_*}[\theta_i \mid Y_i, \X_{2,i}]$,
where $Y_i = (Y_{i1}, \ldots, Y_{iT})'$,  $X_{2,i} = (X_{2,i1}, \ldots, X_{2,iT})'$, and $\mathbb{E}_{G_*}[\cdot]$ denotes expectation with respect to $\theta_i \sim G_*$.

Specifying that $\theta_i \stackrel{\text{i.i.d.}}{\sim} G_*$ is grounded in the principle of empirical Bayes (EB) introduced by \citet{robbins:56}. This approach is intrinsically frequentist because it estimates the prior distribution from the data. EB estimation improves on individual-level maximum likelihood estimation through shrinkage, borrowing strength across units to learn both the individual parameters $\theta_i$ and their population distribution $G_*$.
Following the terminology of \citet{efron-14}, the procedure is referred to as G-modeling  when $G_*$ is estimated directly.
The resulting EB estimator is given by $\hat{\theta}^{\mathrm{EB}}_i := \mathbb{E}_{\hat{G}}[\theta_i \,|\, Y_i, \X_{2,i}]$, where $\hat{G}$ denotes an estimate of $G_*$. An EB estimator $\hat \Tau_i^{\EB}:=\mathbb E_{\hat G}[\Tau_i \,|\, Y_i,X_{2,i}] $ for general unit-specific parameters $\Tau_i$ defined as functions of $\theta_i$ can also be defined. While parametric G-modeling is convenient, nonparametric modeling guards against misspecification.
Recent reviews of EB and its applications in economics include \citet{KoenkerGu2024:JPEmicro,KoenkerGu2025book} and \citet{Walters:24}. To our knowledge, no nonparametric EB analysis has considered all four sources of heterogeneity present in the HIVDX model.

We study identification and estimation of the HIVDX model within an EB framework. The generality of our model with $d_\beta > 1$ and $d_\delta > 0$ introduces several challenges. First, while assuming the identification of $G_*$ in the Gaussian location model with i.i.d.\ errors  without formal proof may be defensible, this assumption cannot be taken for granted in more complex settings. 
Theorem~\ref{thm:identification} establishes identification by requiring (i) the design matrix to have full column rank, (ii) the covariance structure to remain identifiable after removing the mean component by appropriately differencing the data, and (iii) the unknown support of the distribution of $\delta_i$ to be finite.
These conditions imply that $d_\beta < T$, leaving $T - d_\beta$ degrees of freedom to identify the covariance structure; this imposes a constraint on the dimension of heterogeneous variance parameters ($d_\delta$). 

We apply Theorem~\ref{thm:identification} to location models with $d_\beta=1$ and AR(1) or ARMA(1,1) errors, and to the HIVDX model with $d_\beta=2$. Under the maintained assumption that the covariates are  independent of $\theta_i$, Propositions~\ref{prop:dyp-identification} to \ref{prop:HIVDX-id-T=5} show that identification holds in each case, but under different conditions on $T$. In general, identification must be worked out on a case by case basis. However, in each case,  we first eliminate the regression coefficients algebraically using “covariate-based differencing.” We then identify the distribution of $\delta_i$, followed by the conditional distribution of $\beta_i$ given $\delta_i$. 

The second challenge lies in establishing  consistent estimation of the HIVDX model.
The classic proof of \citet{kiefer1956consistency} relies on a compactification of the space of mixing distributions, which is non-compact when $\hc_i$ is allowed to be unbounded.
While the case of $\dbeta = 1$ can be resolved by embedding $\mathbb{R}$ into its compact completion $[-\infty, \infty]$, the multidimensional setting ($\dbeta \ge 2$) presents unique difficulties because mass can diverge along infinitely many distinct directions.
To address this, Section~\ref{sec:consistency} adopts a metric induced by the vague topology to effectively compactify the space of prior distributions.
Theorem~\ref{thm:consistency} establishes  almost sure consistency of the
nonparametric maximum likelihood estimator (\NPM) with respect to this metric.
Building on this result, Theorem~\ref{thm:EB-consistency} in Section~\ref{subsec:EB-consistency} provides conditions for regret consistency (equivalently, asymptotic optimality) of the EB estimator for $\Tau_i$, taking into account that $\Tau_i$  may involve functions of the random coefficients and the data. For example,  the optimal (but infeasible) one-step-ahead prediction under quadratic loss is $\taupm_i = \mathbb{E}_{G_*}[a_i  + b_i X_{2,i T+1} + \rho_i Y_{iT} \,|\, Y_{i}, X_{i},X_{2,i T+1}] 
 - \mathbb{E}_{G_*}[\rho_i a_i  +\rho_i b_i X_{2,iT} \,|\, Y_{i}, X_{i}]$. The difference $\hat \tau^{\EB}_i-\tau_i^*$ between the EB and oracle predictions thus involves interaction of the errors in  estimating the random coefficients and $(Y_i,X_i)$. Proposition~\ref{prop:EB-HIVDX} provides low-level sufficient conditions for the EB prediction to be regret-consistent.

A third challenge is to be able to compute the NPMLE of $G$ as the number of sources of heterogeneity increases. In location models, the NPMLE of $G$ is typically computed by optimizing over a fixed grid of $m$ points.
However, as $d_\theta$ increases, the grid  required to maintain accuracy needs to grow exponentially.
It would be desirable to have an algorithm that allows the support points to move adaptively.
This necessitates a shift from Euclidean optimization to optimization on the space of probability measures, where the updates are governed by gradient flows.
 Motivated by \citet{yan2024learning}, Section~\ref{sec:implementation} develops a Wasserstein-Fisher-Rao (WFR) gradient flow algorithm for panel regressions with multiple sources of heterogeneity. The Fisher-Rao step reweights the support points, while the Wasserstein step updates their locations.

We use the EB approach to analyze income data in the Panel Study of Income Dynamics (PSID) in Section~\ref{sec:empirical}.
The model allows earnings $Y_{it}$ to have unit-specific mean  $a_i$ and response $b_i$ to potential experience $X_{2,it}$, with  a heteroskedastic error covariance structure that depends on internal dynamics $\rho_i$ and exposure to shocks through $\sigma_i^2$.  We find that (i) the slope coefficient $b_i$ for potential experience exhibits substantial heterogeneity and is negatively correlated with the random intercept $a_i$, and (ii) there is pronounced cross-sectional heterogeneity in both the variance ($\sigma_i^2$) and the dynamics ($\rho_i$).  
In the PSID application, the EB predictions have lower prediction-error variance, and hence lower mean squared prediction error, than the individual maximum likelihood predictions.
Section~\ref{sec:Monte:Carlo} shows that   features of the EB estimates and forecast error reduction are replicated in   Monte Carlo experiments calibrated to the application.

\subsection{Related Work}

There is a large body of literature seeking to understand heterogeneity in earnings. 
See \citet{Altoniji:2023} and \citet{browning-ejrnaes:13} for recent reviews of this literature. 
Recent work such as \citet{GK:2017} has incorporated intercept and variance heterogeneity ($a_i$ and $\sigma_i^2$), while keeping homogeneity in the dynamic equation.
We study earnings dynamics using the HIVDX model allowing parameters of the dynamic process to be heterogeneous.   
\citet{Lee:2026} also allows for heterogeneous intercepts, slopes and persistence. He  leaves the distribution of the innovations unrestricted and imposes minimal assumptions to focus on  partial identification. As will be discussed below, we  assume that $\theta_i$ is independent of $X_{2i}$ and impose a Gaussian likelihood with a stationary initial condition; these restrictions allow us to identify and estimate the joint distribution of $\theta_i$ and use it for EB estimation and prediction.

Our identification analysis is also related to \citet{ArellanoBonhomme:2012}. They allow coefficients to depend on regressors, while their distributional identification results impose conditional independence between coefficients and errors. We instead assume independence of the latent heterogeneity and covariates but allow arbitrary dependence among the heterogeneous terms. Therefore, the two distributional frameworks are non-nested. Our HIVDX model also does not nest the  model of \citet{LMS:20}, which has homogeneous $\rho$ but allows for heterogeneous slopes in the form of correlated random coefficients and covariate-dependent heteroskedasticity.

 Heterogeneity is important in many problems of empirical interest, but such problems are often studied in a static setting. In a large-scale correspondence experiment, \citet{KlineRoseWalters:2022} measured contact rates as the probability of being contacted by an employer within 30 days.  Contact gaps are differences in these rates across applicant groups.
\citet{Giacomini-et-al:2025} use a similar model to study firm-level gender contact gaps across experimental waves. \citet{KlineRoseWalters:2024} provide EB estimates for the related decision problem of grading employers. Standard central limit theorem arguments suggest that aggregated firm-level contact-gap estimates have approximately normal sampling distributions. The HIVDX model could provide a more general framework for estimating the distribution of gender and racial gaps in employer contact rates.

There are also works  in the broader EB literature that have some connection to our analysis.
\citet{KevinChen:2024} analyzes conditionally Gaussian settings with known variances \(\sigma_i^2\). He first  studentizes the data by estimating the conditional mean and variance of outcomes given \(\sigma_i^2\), and then estimates the prior to plug into  EB decision rules. 
The paper emphasizes the importance of nonparametric estimation in both stages. 
In contrast, we allow for unobserved \(\sigma_i^2\) and heterogeneous dynamics across multiple time periods. 

\citet{Kwon:2025} considers a multivariate Gaussian model with a hierarchical prior, while \citet{Cheng-et-al:2025} propose a parametric EB estimator for two-way effects with assortative matching.
Both methods rely on linear shrinkage rules selected to minimize unbiased risk estimates.
In contrast, we adopt a fully nonparametric EB approach via \NPM. 
\citet{Gaillac2024} develops posterior mean estimators for multiple random coefficient models for cross-sectional regressions. 
In a panel extension, the paper constructs time averages for each unit and allows heterogeneity only in time-invariant regressors, while assuming homogeneous coefficients for time-varying regressors. 
In contrast, our framework focuses on random coefficients associated with regressors that vary over time.

\citet{AGT:EBAdaptive} investigate EB methods in compound adaptive experiments, where the outcome of each arm follows a normal distribution with an unknown mean. They demonstrate that the risk guarantees for G-modeling, originally derived under i.i.d.\ sampling, continue to hold for adaptively collected data without  knowledge of the sampling algorithm. In contrast, F-modeling is shown to yield biased estimates in this setting.

Recent work has extended nonparametric EB methods beyond the classical Gaussian location model to heteroskedastic models (e.g., \citeay{Bodhi:JRSSB}), 
high-dimensional or multivariate regression (e.g., \citeay{Bodhi:EB-HD}; \citeay{Wu:EB-HD}; \citeay{Jiang:2025}),  and variance estimation (e.g., \citeay{IgnatiadisSen2025AoS}). These works only consider a static model while our focus is the heterogeneous dynamic panel model presented above as HIVDX.

\section{The Econometric Setup and EB Modeling}\label{sec:model}

We consider a panel of observations on $N$ cross-sectional units over $T$ time periods, focusing on a short-panel framework in which $N$ grows while $T$ remains fixed. 
Our main setup is a general heterogeneous coefficient (HC) model that nests the HIVDX model:
\begin{equation}
\label{eq:model}
    Y_i = X_i \beta_i + P_i e_i ,
\end{equation}
where \(Y_i = (Y_{i1}, \ldots, Y_{iT})'\) is a \(T \times 1\) vector of outcomes, \(X_i = (X_{i1}, \ldots, X_{iT})'\) is a \(T \times d_{\beta}\) matrix of observed covariates, and $e_i = (e_{i1}, \ldots, e_{iT})'$ is a vector of independent Gaussian errors.  
We assume that the first column of $X_i$ is $\iota_T$, the $T$-dimensional vector of ones.
For each $i$, $P_i$ is lower triangular, and the corresponding conditional covariance matrix $\Sigma_i \equiv P_i P_i'$ is positive definite.  To capture its dependence on heterogeneous variance parameters $\delta_i \in \mathbb{R}^{\ddelta}$, we write $P_i = P(\delta_i)$, where $\delta \mapsto P(\delta)$ is a continuously differentiable function of $\delta$.

To see that the HIVDX model is  a special case of the HC model for particular $(\beta_i, \delta_i)$, note that since $e_{it}$ are i.i.d. $\mathcal{N}(0,1)$, $u_i = (u_{i1}, \ldots, u_{iT})'$ is normally distributed with conditional covariance matrix $\Sigma_i$, whose $(t,s)$ entry is given by
\begin{equation}\label{def:Sigma:model}
   \Sigma_{i,ts} = \sigma_i^2 \frac{\rho_i^{|t-s|}}{1-\rho_i^2}, \quad t, s = 1, \ldots, T.
\end{equation}
Hence, $P_i$ is defined via the Cholesky decomposition of $\Sigma_i = P_i P_i'$ with $\delta_i = (\sigma_i^2, \rho_i)$.  
 We develop our theoretical framework through the lens of the HC model since it provides a unified way to analyze a broader class of dynamics such as ARMA errors.
Let $I_T$ denote the $T \times T$ identity matrix. The following assumptions will be used throughout.

\begin{assumption}[DGP]\label{asm:DGP}
Let \((Y_i, X_i)_{i=1}^N\) be identically and independently distributed (i.i.d.) data, satisfying \eqref{eq:model}. In addition, the following conditions hold.
\begin{enumerate}

\item[(i)] 
The pairs \((\beta_i, \delta_i)_{i=1}^N\) are i.i.d. draws from an unknown distribution \(G_*\) and are independent of \((X_i)_{i=1}^N\).

\item[(ii)] 
The parameter space for $\delta_i$ is a compact subset $\mathcal{K}_{\delta} \subseteq \mathbb{R}^{\ddelta}$.
For each $\delta_i \in \mathcal{K}_{\delta}$, the covariance matrix $\Sigma_i \equiv P(\delta_i)P(\delta_i)'$ satisfies $\underbar{c}\, I_T \le \Sigma_i \le \bar{c}\, I_T$ for some $0 < \underbar{c} < \bar{c} < \infty$.

\item[(iii)] 
$e_i \sim \mathcal{N}(0, I_T)$ independent of $(\beta_i, \delta_i, X_i)$.

\end{enumerate}
\end{assumption}

Assumption~\ref{asm:DGP} requires that $(\beta_i,\delta_i)_{i=1}^N$ are iid and independent of $(X_i)_{i=1}^N$.  Independence is understood to be in a conditional sense since as noted above,  observables are assumed to have been partialled out from $X_i$. Early works in dynamic panel models typically do not allow for heterogeneous slope and variance parameters.  Understanding the issues in modeling these parameters as exogenous random coefficients  is a necessary step before  considering endogenous heterogeneous coefficients in a dynamic setting.

Conditional on \((\beta_i, \delta_i)\) and \(X_i\), the outcome vector \(Y_i\) follows the distribution \(\mathcal{N}(X_i \beta_i, \Sigma_i)\). Individuals in this model are distinguished by types characterized by their heterogeneous features. Given the information structure, individual types cannot be directly uncovered from the data, but their distributions can be identified.  

Let $\theta_i := (\hc_i, \delta_i)$ collect all heterogeneous parameters, lying in the space $\Theta := \mathbb{R}^{\dbeta}\times\mathcal{K}_{\delta} \subseteq \mathbb{R}^{\dwhole}$ with $\dwhole = \dbeta + \ddelta$.
If $\theta_i$ were observed, the conditional density of $Y_i$ given $\X_i$ is
\begin{align*}
   \ell(Y_{i} \,|\, \X_i, \theta_i) :=  \frac{1}{({2\pi})^{T/2}|P(\delta_i)|} \exp \left( - \frac{1}{2}\|P(\delta_i)^{-1} (Y_i - \X_i\hc_i )\|^2 \right).
\end{align*}
When $\theta_i$ is unobserved but its distribution is given as $G$, the marginal density of $Y_i$ given $\X_i$, obtained by integrating out $\theta_i$ from the likelihood, is
\begin{equation*}
    f_{G}(Y_i, \X_i) := \int_{\Theta} \ell(Y_i \,|\, \X_i, \theta) dG(\theta),
\end{equation*}
where the independence between $\theta_i$ and $X_i$ is used.\footnote{We use $G_*$ to denote the unknown true distribution and $G$ to denote a generic distribution in the class of possible distributions that includes $G_*$. The notation $f_{G}(Y_i, \X_i)$ emphasizes that the marginal density of $Y_i$ given $\X_i$ depends on $G$.} 

For the estimation of $(\theta_i)_{i=1}^N$, two conceptual frameworks are available.
The first is the \emph{separate} decision framework that treats each $\theta_i$ as fixed and minimizes the individual risk $\mathbb{E}[\|\hat \theta_i - \theta_i\|^ 2 \,|\, \theta_i]$.
This yields the individual MLE,
\begin{equation*}
\thetamle_i := \argmax_{\theta \in \Theta} \ell(Y_i \,|\, \X_i, \theta),
\end{equation*}
which provides an asymptotically efficient estimator for $\theta_i$ as $T \to \infty$.
However, \citet{stein:54} and \citet{james-stein:61} demonstrate  
for independent Gaussian observations
that the MLE is inadmissible under squared-error loss when the dimension of the parameter vector is at least three (here, $N \ge 3$). In this setting, the MLE is dominated by a shrinkage estimator that achieves lower total risk by pooling information across observations.

The second approach is the \emph{compound} decision framework. Unlike the separate framework, which minimizes risk for each unit individually, this approach seeks to minimize the aggregate risk over the ensemble of $N$ units.
While the strict compound decision formulation treats $(\theta_i)_{i=1}^N$ as fixed, \citet{robbins:56} demonstrates that this problem can be effectively solved by adopting an EB perspective.
By modeling the parameters as i.i.d.\ draws from a common latent distribution $G_*$, we can derive an optimal (oracle) decision rule, which takes the form of the posterior mean:
\begin{equation}
\label{eq:posterior-mean}
\thetapm_i := \mathbb{E}_{G_*}[\theta_i \,|\, Y_i, \X_i] = \frac{\int_{\Theta} \theta \ell(Y_i \,|\, \X_i, \theta) dG_*(\theta)}{\int_{\Theta} \ell(Y_i \,|\, \X_i, \theta) dG_*(\theta)}.
\end{equation}
This estimator improves efficiency by pooling information across units, thereby reducing the variance relative to the individual maximum likelihood estimators $\thetamle_i$.
These efficiency gains come at the cost of introducing individual bias (i.e., $\mathbb{E}[\thetapm_i \,|\, \theta_i] \ne \theta_i$), but the reduction in the overall compound mean squared error dominates this cost.
In essence, $\thetapm_i$ shrinks $\thetamle_i$ toward the prior.

In practice, the compound decision rule is infeasible due to the unknown $G_*$.
In Bayesian analysis, the prior $G_*$ is typically specified as a parametric distribution, 
either by fixing its hyperparameters or by assigning them their own distributions.
\citet{efron-morris:73} show that the James-Stein shrinkage estimator can be interpreted as a parametric EB estimator, with $G_*$ serving as a parametric prior whose hyperparameters are estimated from the data.

An Empirical Bayes analysis can approximate the infeasible oracle decision rule in one of two ways: F-modeling or G-modeling.
F-modeling directly computes the shrinkage term from the data without explicit knowledge of $G_*$.
For example, assuming $\Sigma_i \equiv \Sigma$ for some known homogeneous covariance matrix $\Sigma$ (that is, $\rho_i=\rho$ and $\sigma_i^2 = \sigma^2$ for some known $\rho$ and $\sigma^2$), the (infeasible) F-modeling EB estimator is
\begin{equation}\label{eq:Tweedie}
\betafeb_i := \betamle_i + (\X_i'\Sigma^{-1} \X_i)^{-1} \frac{\partial}{\partial \hat{\hc}}\log \hat f_{\hat{\hc}|X} (\betamle_i, \X_i),
\end{equation}
where $\betamle_i = (\X_i'\Sigma^{-1} \X_i)^{-1}\X_i'\Sigma^{-1} Y_i$ 
and $\hat f_{\hat{\hc}|X}(\hat{\hc}, x)$ is an estimator for $f_{\hat{\hc}|X}(\hat{\hc}, x)$ – the conditional density of $\betamle_i$ given $\X_i=x$ under $G_*$.
In practice, $\Sigma$ is typically unknown and must also be estimated. 
Equation~\eqref{eq:Tweedie}, commonly known as \textit{Tweedie's formula} in the literature, shows that the posterior mean of $\hc_i$ adjusts the corresponding MLE by a shrinkage term derived from the cross-sectional distribution of the MLE.\footnote{In the simplest case when $Y \,|\, \theta \overset{\text{ind.}}{\sim} \mathcal{N}(\theta, 1)$ with $T=1$, the Tweedie formula is $\mathbb{E}_{G_*}[\theta| y] = y + \tfrac{d}{dy} \log f_{G_*}(y)$, where $f_{G_*}(\cdot)$ denotes the marginal density of $Y$ under $G_*$.} 
The magnitude of the shrinkage term grows with the variance of $\betamle_i$, as measured by $(\X_i' \Sigma^{-1} \X_i)^{-1}$, and with the informativeness of the prior, as reflected in $f_{\hat{\hc}|X}$.\footnote{ It can be shown that the oracle-decision errors and the shrinkage term are uncorrelated, that is,
$\cov_{G_*} (\betapm_i - \hc_i, \betamle_i - \betapm_i ) = 0$, where $\betapm_i := \mathbb{E}_{G_*}[\beta_i \,|\, Y_i, \X_i]$,
implying that the variance reduction is attributable to the shrinkage term in \eqref{eq:Tweedie}.}

F-modeling typically relies on the assumption that the likelihood belongs to the linear exponential family, offering computational convenience when this condition holds.
\citet{Walters:24} provides a survey of F-modeling applications in labor economics.
\citet{LMS:20} estimate the posterior mean of \(a_i\)  by F-modeling using Tweedie's formula. However, this method does not readily extend to the HIVDX model. 
We adopt the alternative of G-modeling and explicitly estimate the unknown prior $G_*$. \citet{shen2022empirical} demonstrate that G-modeling has better statistical properties in  Poisson mixture models than F-modeling.

While parametric G-modeling imposes a specific functional form on the prior, nonparametric G-modeling characterizes the true distribution $G_*$ as the maximizer of the population log-likelihood:
\begin{equation}
\label{eq:npmle-population-program}
    G_* \in \argmax_{G \in \mathcal{G}} F(G),
\end{equation}
where $\mathcal{G}$ is the class of all probability distributions on $\Theta$ and $F(G) := \mathbb{E}[\log f_{G}(Y_i, \X_i)]$ denotes the expected marginal log-likelihood (i.e., the negative risk).
This formulation motivates estimating $G_*$ via its sample analog:
\begin{equation}
\label{eq:npmle-sample-program}
    \hat G \in \argmax_{G \in \mathcal{{G}}} F_N(G),
\end{equation}
where the empirical criterion is given by $F_N(G) := N^{-1} \sum_{i=1}^N \log f_{G}(Y_i, \X_i)$.
We refer to $\hat G$ as the \NPM\ of $G_*$.
Substituting $\hat G$ for $G_*$ in \eqref{eq:posterior-mean} yields the G-modeling EB estimator for $\theta_i$:
\begin{equation}\label{eq:posterior-mean-est}
\hat \theta_i^{\EB}:= \frac{\int_{\Theta} \theta \ell(Y_i \,|\, \X_i, \theta) d\hat G(\theta)}{\int_{\Theta} \ell(Y_i \,|\, \X_i, \theta) d\hat G(\theta)}.
\end{equation}
The \NPM\ is built on the same principle as parametric MLE, namely that $G_*$ is the maximizer of the population log-likelihood defined in \eqref{eq:npmle-population-program}.
If $G_*$ is the unique maximizer of $F$ over $\mathcal{G}$, it is said to be (point-)identified.
However, it is known that the \NPM\ in \eqref{eq:npmle-sample-program} may not be unique in multidimensional settings; in such cases, $\hat{G}$ refers to any maximizer.
Thus, before presenting an algorithm to compute $\hat G$ in Section \ref{sec:implementation}, we first establish conditions for the identification of $G_*$ in Section \ref{sec:identification}, and then show in Section \ref{sec:consistency} that consistent estimation is possible despite multiple sources of heterogeneity.

\section{Identification}\label{sec:identification}

This section establishes identification of the HC model defined in \eqref{eq:model}.
Let $\mathcal{G}$ denote the class of all distributions on $\Theta \equiv \mathbb{R}^{\dbeta}\times \mathcal{K}_{\delta}$, defined as
\begin{equation}\label{def:G-dist-class}
    \mathcal{G} := \mathcal{P}(\mathbb{R}^{\dbeta}\times \mathcal{K}_{\delta}) = \{ G \in \mathcal{P}(\mathbb{R}^\dwhole) : \supp_G(\delta_i)\subseteq \mathcal{K}_{\delta}\},
\end{equation}
where $\mathcal{K}_{\delta}\subseteq\mathbb{R}^{\ddelta}$ is a compact subset in which $\delta_i$ is supported, as in Assumption~\ref{asm:DGP}.
 Since any solution to the MLE in \eqref{eq:npmle-population-program} must lie in the set
\begin{equation}\label{eq:G-IDset-via-MLE}
    \left\{G \in \mathcal{G} : f_G(Y_i,\X_i) = f_{G_*}(Y_i, \X_i)\ \ \text{a.s.}\right\},
\end{equation}
the identified set of distributions consists of distributions consistent with the observed marginal likelihood.
We now state the conditions on the covariates, the parametric form of $P(\cdot)$, and the distribution of $\delta_i$ required for the identification of $G_*$.

\begin{assumption}
\label{asm:id-conditions}
There exists a pair $(\xmat,\Q)$ of fixed matrices with $\xmat \in \supp(X_i)$ and $\Q \in \mathbb{R}^{T\times (T-\dbeta)}$ such that the following conditions hold:
\begin{enumerate}

\item[(i)] $\rank(\xmat) = \dbeta \le T-1$, $\rank(\Q) = T - \dbeta \ge 1$, and $\Q'\xmat = 0$.

\item[(ii)] For any $\delta, \tilde{\delta} \in \mathcal{K}_{\delta}$,  
\[
\Q'P(\delta)P(\delta)'\Q = \Q'P(\tilde{\delta})P(\tilde{\delta})'\Q
\quad \text{if and only if} \quad \delta = \tilde{\delta}.
\]

\item[(iii)] The unknown support of the distribution of $\delta_i$ under $G_*$ is finite.
\end{enumerate}
\end{assumption}

Assumption~\ref{asm:id-conditions}(i) requires $\xmat\in \supp(\X_i)$ to have full column rank, paralleling the familiar rank condition in regression. 
This ensures the existence of an annihilator matrix $\Q$ with rank $T - \dbeta$.\footnote{The choice of $\Q$ is unique (up to a nonsingular linear transformation) given $\xmat$, since a different choice of $\Q$ can be written as $\tilde{\Q} = \Q A$ for some nonsingular $A$.}
We further require that $\dbeta \le T-1$ so that $\Q$ has non-null rank and can be used in part~(ii).
This restriction is needed because, in our setting, the covariances $\Sigma_i$ depend on unknown $\delta_i$.

Assumption~\ref{asm:id-conditions}(ii) imposes that $\delta\mapsto \Q'P(\delta)P(\delta)'\Q$ is a one-to-one function of $\delta \in \mathcal{K}_{\delta}$, meaning that the variance parameters remain identifiable after removing the effect of $\xmat \hc_i$.
The intuition is as follows.
Since the concatenated matrix $(\xmat, \Q)$ is of full rank, the information in $Y_i$, conditional on $\X_i = \xmat$, can be expressed as
\begin{equation*}
\left[ \begin{matrix}
    \Q' Y_i \\ 
    \xmat' Y_i     
\end{matrix} \right]
= \left[ \begin{matrix}
    0 \\ 
    \xmat' \xmat \hc_i     
\end{matrix} \right] 
+ 
\left[ \begin{matrix}
    \Q' P_i e_i \\ 
    \xmat' P_i e_i     
\end{matrix} \right],
\end{equation*}
where $\Q 'Y_i$ is a Gaussian mixture with  mean zero and variance $\Q'P_iP_i'\Q=\Q'P(\delta_i)P(\delta_i)'\Q$. The distribution of $\Q'P(\delta_i)P(\delta_i)'\Q$ can therefore be identified from the distribution of $\Q 'Y_i$  given Assumption~\ref{asm:id-conditions}(iii).
Then, the distribution of $\delta_i$ is identified in view of Assumption~\ref{asm:id-conditions}(ii).

To identify $G_*$, the remaining step is to recover the conditional distribution of $\hc_i$ given $\delta_i$.
Let $\mathcal{E} \in \mathbb{R}^{\dbeta\times(T-\dbeta)}$.
The key idea is to express
\begin{align*}
\Q'Y_i + \mathcal{E}' \xmat' Y_i & = (\Q' + \mathcal{E}' \xmat')P(\delta_i) e_i + \mathcal{E}'(\xmat'\xmat) \hc_i
\end{align*}
as a perturbation of $\Q'Y_i$ with noise $\mathcal{E}' \xmat' Y_i$.
We measure the shift in the mixing distribution of $\Q'Y_i + \mathcal{E}' \xmat' Y_i$ relative to the mixing distribution of $\Q'Y_i$ identified earlier.
This comparison reveals the conditional distribution of $\mathcal{E}'(\xmat'\xmat) \hc_i$ given $\delta_i$ as  $\mathcal{E}$ traverses $\mathbb{R}^{\dbeta\times(T-\dbeta)}$, which then identifies $\hc_i$ given $\delta_i$ via the Cram\'er--Wold device.

The following theorem formalizes this result.
The proof of Theorem~\ref{thm:identification} as well as the proofs of all other theorems and propositions are given in Supplemental Appendix~\ref{sec:proof}.

\begin{thm}
\label{thm:identification}
Under Assumptions~\ref{asm:DGP} and \ref{asm:id-conditions}, $G_*$ is identified in the class $\mathcal{G} = \mathcal{P}(\mathbb{R}^{\dbeta} \times \mathcal{K}_\delta)$.
\end{thm}

The primary challenge in identifying $G_*$ is to distinguish the mean and variance components that are convolved into the outcome distribution.
For instance, assume that $Y_i \sim \mathcal{N}(0,1)$ is a scalar observation.
The observed distribution of $Y_i$ can be rationalized by a continuum of models: $Y_i = a_i + \sigma_i e_i$, where $a_i \sim\mathcal{N}(0,v_a)$ and $\sigma_i\equiv \sqrt{1-v_a}$ for $v_a \in [0,1)$.
Thus, one cannot tell which model generates the observed distribution without imposing further restrictions on $a_i$ or $\sigma_i$.
\citet{bruni1985identifiability} establish identification results for Gaussian mixtures under a compact support assumption for both the mean and variance components.\footnote{The results of  \citet{bruni1985identifiability} are also used in identification of latent distributions in labor economic analyses (see, e.g., \cite{Pastorino:2024,bunting2024,dePaula2025}), but the context is different from identification of mixture models in econometrics such as discussed in  \cite{Compiani:Kitamura:2016}.}
Using a type of identification-at-infinity argument, they show that both components can be separated by leveraging variation in the tails of $Y_i$.
However, this argument breaks down as soon as the mean component places a slight nonzero mass in the tails, producing observationally indistinguishable yet fundamentally different models, as illustrated above.
In contrast, our identification strategy rests on the rank and one-to-one restrictions that preclude this type of identification failure.

Theorem~\ref{thm:identification} extends to models that include homogeneous slope parameters:
\begin{equation}
\label{eq:extended-model}
    Y_i = X_i \beta_i + W_i \gamma_* + P_i e_i,
\end{equation}
where \(W_i = (W_{i1}, \ldots, W_{iT})'\) is a \(T \times d_{\gamma}\) matrix of additional controls associated with common slopes $\gamma_* \in \mathbb{R}^{\dgamma}$. 
The common parameters $\gamma_*$ can often be identified without knowledge of $G_*$ by pooling cross-sectional information.

To see how $\gamma_*$ can be identified, rewrite \eqref{eq:extended-model} as a pooled-effects model:
\begin{equation*}
Y_{i} = X_{i} \beta_* + W_{i}\gamma_* + u_{i},
\end{equation*}
where $\beta_* := \mathbb{E}_{G_*}[\beta_i]$ and $u_i = X_i(\beta_i - \beta_*) + P_i e_i$ denotes the composite error.  
The parameter $\gamma_*$ is identified as the familiar regression coefficient
\begin{equation*}
    \gamma_* = \mathbb{E}[\tilde{W}_i' W_i]^{-1}\mathbb{E}[\tilde{W}_i' Y_i],
\end{equation*}
provided that $\mathbb{E}[\tilde{W}_i' u_i] = 0$ and $\mathbb{E}[\tilde{W}_i'\tilde{W}_i] > 0$, where $\tilde{W}_i := W_i - X_i \mathbb{E}[X_i'X_i]^{-1}\mathbb{E}[X_i' W_i]$.  
Once $\gamma_*$ is identified, the model can be analyzed using the residuals $Y_i - W_i \gamma_*$.
Theorem~\ref{thm:identification} serves as a general tool for analyzing identification in HC models. Before applying  Theorem~\ref{thm:identification} to the  HIVDX model, we first consider a HIVD model without the covariate $X_{2i}$. We verify Assumption \ref{asm:id-conditions} by demonstrating how to choose a pair $(\xmat,\Q)$.

\subsection{Identification of the HIVD Model}

We consider two variants of the HIVD model: one with AR(1) errors and the other with ARMA(1,1) errors.

\subsubsection{AR(1) Errors}
\label{subsubsec:HIVD-AR1}
Consider the following AR(1) panel model without covariates:
\begin{equation}
\label{eq:simple-dynamic-panel}
\begin{aligned}
    Y_{it} & = a_i + u_{it}, \\
    u_{it} & = \rho_i u_{it-1} + \sigma_i e_{it}, \quad t = 1,\ldots, T,
\end{aligned}
\end{equation}
where the initial condition $u_{i0}$ is assumed to be drawn from the stationary distribution $\mathcal{N}(0, \sigma_i^2/(1-\rho_i^2))$.
The parameters are collected in $\theta_i = (a_i, \sigma_i^2, \rho_i) \iid G_*$.

\begin{prop}
\label{prop:dyp-identification}
Assume that $T \ge 3$ and that Assumptions~\ref{asm:DGP} and \ref{asm:id-conditions}(iii) hold.
Consider the panel AR(1) model given in \eqref{eq:simple-dynamic-panel}.
 $G_*$ is identified in $\mathcal{G} = \mathcal{P}(\Theta)$, where $\Theta = \mathbb{R} \times \K_{(\sigma^2,\rho)}$
and $\K_{(\sigma^2,\rho)}$ is a compact subset of $(0, \infty) \times (-1,1)$.
\end{prop}

The restriction on $\K_{(\sigma^2,\rho)}$ implies that $\rho_i$ is bounded away from $\{-1,1\}$ and that $\sigma_i^2$ is bounded away from $0$ and infinity.
The condition $T \ge 3$ is not only sufficient but also necessary for the identification of the three-dimensional parameter $\theta_i$ when $\K_{(\sigma^2,\rho)}$ has a nonempty interior.
The variance parameters are identified from the covariance structure after purging the fixed effect $a_i$ by first-differencing $Y_i$ via a suitably chosen $\Q$.
Finally, the compact support for $\delta_i$ ensures that $Y_i$ has non-degenerate and bounded covariance matrices across all $\delta_i$; in particular, this condition rules out unit root processes.
Our choice of $\Q$ can be illustrated with $T=3$.
Let $\xmat$ be the column vector of ones whose rank is $1$. 
Set $\Q$ equal to the first-difference matrix:
\begin{equation*}
    \Q = 
    \begin{pmatrix}
    -1 & 0 \\
    1 & -1 \\
    0 & 1    
    \end{pmatrix}.
\end{equation*}
This choice of $\Q$ yields
\begin{equation*}
   \mathcal{V} := \Q' P(\delta)P(\delta)'\Q 
= \frac{\sigma^2}{1+\rho} 
\begin{pmatrix}
    2 & -1+\rho \\    
    -1+\rho & 2    
\end{pmatrix}.
\end{equation*}
Since $\rho = 1 + 2\mathcal{V}_{21}/\mathcal{V}_{11}$ and $\sigma^2 = \mathcal{V}_{11} + \mathcal{V}_{21}$, there exists a one-to-one relationship between $\mathcal{V}$ and $\delta$, which ensures that $\delta$ is identified from $\Q' P(\delta)P(\delta)'\Q$.

\subsubsection{ARMA(1,1) Errors}
\label{subsubsec:HIVD-ARMA}

Consider the ARMA(1,1) panel model:
\begin{equation*}
\begin{aligned}
    Y_{it} &= a_i + u_{it}, \\
    u_{it} &= \rho_i u_{it-1} + \sigma_i (e_{it} + \varphi_i e_{it-1}), \quad t = 1, \dots, T,
\end{aligned}
\end{equation*}
where $e_{it} \iid \mathcal{N}(0,1)$ and the initial condition $(u_{i0}, e_{i0})$ is drawn from its stationary joint distribution.
The additional random coefficient $\varphi_i \in [-1,1]$ captures heterogeneity in the MA component.
Under the given assumptions, $u_{i}$ is normally distributed with a conditional covariance matrix $\Sigma_i$, whose $(t,s)$ entry is given by
\begin{equation*}
\Sigma_{i,ts} = \frac{\sigma_i^2}{1-\rho_i^2} \times
\begin{cases}
    1+\varphi_i^2 + 2 \varphi_i \rho_i & \text{if } t=s, \\
    (\rho_i + \varphi_i) (1 + \rho_i \varphi_i) & \text{if } |t-s|=1, \\
    \rho_i^{|t-s|-1} (\rho_i + \varphi_i) (1 + \rho_i \varphi_i) & \text{if } |t-s| \ge 2.
\end{cases}
\end{equation*}

\begin{prop}
\label{prop:identification-ARMA}
Let $\theta_i = (a_i, \sigma_i^2, \rho_i, \varphi_i)$ collect all heterogeneous coefficients in the panel model with ARMA(1,1) errors, with $\theta_i \iid G_*$.
Assume that $T \ge 4$ and that Assumptions~\ref{asm:DGP} and \ref{asm:id-conditions}(iii) hold.
Then $G_*$ is identified in $\mathcal{G} = \mathcal{P}(\Theta)$, where $\Theta = \mathbb{R} \times \mathcal{K}_{(\sigma^2,\rho,\varphi)}$
and $\mathcal{K}_{(\sigma^2,\rho,\varphi)}$ is a compact subset of $\{(\sigma^2,\rho,\varphi) \in (0,\infty)\times (-1,1)\times [-1,1] : \rho + \varphi \neq 0\}$.
\end{prop}

The condition $T \ge 4$ is necessary to observe the process over a horizon sufficient to distinguish the two sources of persistence, namely $\rho_i$ and $\varphi_i$.
The restriction $\rho + \varphi \neq 0$ ensures that the two dynamic components do not perfectly offset one another.
We provide the specific choice of $(\xmat,\Q)$ for Assumption~\ref{asm:id-conditions} in the proof of Proposition~\ref{prop:identification-ARMA}.

\subsection{Identification of the HIVDX Model}

We now return to the HIVDX model given in \eqref{def:model:hivxd}.
In this model, the parameters are partitioned as $\hc_i = (a_i, b_i)$ and $\delta_i = (\sigma_i^2, \rho_i)$, yielding the full parameter vector $\theta_i = (a_i, b_i, \sigma_i^2, \rho_i)$ with dimension $\dwhole = 4$.
Define $\Delta x_t := x_t - x_{t-1}$ as the first difference operator applied to the sequence $(x_t)_{t=1}^T$.
Whether and how $G_*$ can be identified depends critically on the specific time-series variation of the covariates. 
We illustrate this using two processes for $X_{2it}$. In the first case, the parameters of the HIVDX model are identified when $T\ge 4$, whereas our sufficient conditions require $T\ge 5$ in the second case.

\subsubsection{Case I: Linear Trend}
We consider the general setup when $
\Delta X_{2,i2} = \pm \Delta X_{2,i3} \ne 0  \text{ or } \Delta X_{2,i2} \ne \Delta X_{2,i4}
$. In  the special case with $X_{2,it} = X_{2,i1} + (t - 1)$, we have  a linear trend with a random initial value. We start with the following proposition. 

\begin{prop}
\label{prop:HIVDX-id-T=4}
Assume that $T\ge 4$, that Assumptions~\ref{asm:DGP} and \ref{asm:id-conditions}(iii) hold, and either
$
\Delta X_{2,i2} = \pm \Delta X_{2,i3} \ne 0  \text{ or } \Delta X_{2,i2} \ne \Delta X_{2,i4}
$
holds with positive probability.
Then, $G_*$ is identified in $\mathcal{G} = \mathcal{P}(\Theta)$, where $\Theta = \mathbb{R}^2 \times \K_{(\sigma^2,\rho)}$
and $\K_{(\sigma^2,\rho)}$ is a compact subset of $(0, \infty) \times (-1,1)$.
\end{prop}

To fix ideas, consider again the linear trend case.  We have $\Delta X_{2,it} \equiv 1$ for all $t$.
We need to find  an appropriate  $(\xmat,\Q)$.
Pick $X_{2,1}$ in $\supp(X_{2,i1})$ and define
\begin{equation*}
\xmat := \left( \begin{matrix}
    1 & X_{2,1} \\
    1 & X_{2,1}+1 \\
    1 & X_{2,1}+2 \\
    1 & X_{2,1}+3 
\end{matrix} \right)  ,\quad 
\Q := \left( \begin{matrix}
    1 & 0 \\
    -2 & 1 \\
    1 & -2 \\
    0 & 1 
\end{matrix} \right),
\end{equation*}
where $\Q$ represents the second-difference operator in matrix form.
It is clear that $\xmat$ has full column rank.
Furthermore, because $\Q$ annihilates both the intercept (column of ones) and $X_{2, it}$, the transformation $\Q' Y_i$ eliminates the fixed effects $\hc_i$ entirely.
Let
\begin{equation*}
\mathcal{V} := \Q' P(\delta)P(\delta)' \Q \equiv \var(\Q' Y_i \mid \delta_i = \delta, \X_i = \xmat),
\end{equation*}
where $\Q' Y_i \in \mathbb{R}^2$ collects the second differences of the errors, $(\Delta^2 {u}_{i3}, \Delta^2 {u}_{i4})$, conditional on $\X_i=\xmat$.
We obtain the following equations:
\begin{align*}
\mathcal{V}_{11} &= \var(\Delta^2 {u}_{i3} \mid \delta_i = \delta)
= \var(u_{i3}-2 u_{i2} + u_{i1} \mid \delta_i = \delta)
= \frac{2(3-\rho)}{1+\rho}\sigma^2, \\
\mathcal{V}_{12} &= \cov(\Delta^2 {u}_{i3}, \Delta^2 {u}_{i4} \mid \delta_i = \delta)
= \cov(u_{i3}-2 u_{i2} + u_{i1}, u_{i4}-2 u_{i3} + u_{i2} \mid \delta_i = \delta) \\
    &= -\frac{\rho^2-3\rho+4}{1+\rho}\sigma^2.
\end{align*}
Hence, the system can be inverted:
\begin{align*}
\sigma^2 = \frac{\mathcal{V}_{11}}{2}\frac{1+\rho}{3-\rho},
\qquad 
\rho = 
\frac{3\mathcal{V}_{11}+2\mathcal{V}_{12}
-\sqrt{4\mathcal{V}_{12}^{2}-12\mathcal{V}_{11}\mathcal{V}_{12}-7\mathcal{V}_{11}^{2}}}
{2\mathcal{V}_{11}}.
\end{align*}
Here $\mathcal{V}_{12}/\mathcal{V}_{11}\in(-1,-1/2)$,
and the displayed root for $\rho$
is the unique root in $(-1,1)$; the other root exceeds 1.
This shows that the mapping $\delta \mapsto \mathcal{V}$ is one-to-one, thereby satisfying Assumption~\ref{asm:id-conditions}(ii). 
However, as shown in Appendix~\ref{app:T4-failure},
Proposition~\ref{prop:HIVDX-id-T=4} does not apply to the level-shift example in Case II that will now be discussed.

\subsubsection{Case II: Level Shift}
In this case, we suppose that $\Delta X_{2,it}$ is not identically zero for some $t \in \{2,\ldots, T\}$.  For example, in the case of a nationwide policy change that induces a    level shift: $X_{2,it} = \mathbf{1}\{t \ge 3\}$.
\begin{prop}
\label{prop:HIVDX-id-T=5}
Assume that $T \ge 5$, that Assumptions~\ref{asm:DGP} and \ref{asm:id-conditions}(iii) hold, and that $\Delta X_{2,it}$ is not identically zero for some $t \in \{2,\ldots, T\}$.
Then, $G_*$ is identified in $\mathcal{G} = \mathcal{P}(\Theta)$, where $\Theta = \mathbb{R}^2 \times \K_{(\sigma^2,\rho)}$
and $\K_{(\sigma^2,\rho)}$ is a compact subset of $(0, \infty) \times (-1,1)$.
\end{prop}
To fix ideas, consider again the level shift case.
Then, the sequence of first differences is
$
(\Delta X_{2,i2}, \Delta X_{2,i3}, \Delta X_{2,i4}, \Delta X_{2,i5}) = (0, 1, 0, 0).
$
Then, $\Q$ can now be constructed as
\begin{equation*}
\Q = \begin{pmatrix}
    -1 & 0 & 0 \\
    1 & 0 & 0 \\
    0 & -1 & 0 \\
    0 & 1 & -1 \\
    0 & 0 & 1
\end{pmatrix},
\end{equation*}
which yields $\Q'Y_i = (\Delta u_{i2}, \Delta u_{i4}, \Delta u_{i5})'$. Since the $(1,1)$ and $(2,3)$ elements of the covariance matrix $\mathcal{V} = \var(\Q'Y_i \mid \delta_i = \delta)$ are given by
\begin{align*}
\mathcal{V}_{11} & = \var(\Delta {u}_{i2} \mid \delta_i = \delta) = \frac{2\sigma^2}{1+\rho},\\
\mathcal{V}_{23} & = \cov(\Delta u_{i4}, \Delta u_{i5} \mid \delta_i = \delta) = \frac{(\rho-1)\sigma^2}{1+\rho},
\end{align*}
we can identify $\rho$ via the ratio $2\mathcal{V}_{23}/\mathcal{V}_{11} + 1$, and subsequently recover $\sigma^2$.
This establishes that the mapping $\delta \mapsto \mathcal{V}$ is one-to-one, as required by Assumption~\ref{asm:id-conditions}(ii).
The result in  Proposition~\ref{prop:HIVDX-id-T=5} also  extends to the analysis of heterogeneous treatment implemented at a known time $2 \le s \le 5$, giving
\begin{align*}
Y_{it} = a_i + b_i\,\mathbf{1}\{t \ge s\} + u_{it}, \quad t = 1,\ldots,5.
\end{align*}
Then the term $\mathbf{1}\{t \ge s\}$ represents a common treatment or policy, and $b_i$ captures the individual-specific treatment effect. 

As these examples illustrate, incorporating covariates is significantly more involved.
This is because the transformation matrix $\Q$ that yields differenced data often obscures the autocovariance structure required to identify $\delta_i$.
It is worth emphasizing that while the sufficient conditions in Proposition~\ref{prop:HIVDX-id-T=5} are simple to verify, the underlying proof is combinatorially complex.
Establishing this general result requires an exhaustive analysis of all distinct patterns of $\Delta X_{2,it}$ to verify the one-to-one mapping between $\delta$ and $\mathcal{V}$ in every possible scenario, as shown in the Appendix.

\section{Consistency of the \NPM}
\label{sec:consistency}

This section establishes that $G_*$ can be consistently estimated by the \NPM\ under Assumption~\ref{asm:id-conditions} and the additional regularity assumptions stated below.

\begin{assumption}\label{asm:consistency}
\begin{enumerate}
    \item [(i)] $\mathbb{E}[\|Y_i\|^2] < \infty$ and $\mathbb{E}[\|\X_i\|^2]<\infty$.
    \item [(ii)] $\mathbb{P}(\rank(\X_i) = \dbeta) = 1$.
\end{enumerate}
\end{assumption}

Assumption~\ref{asm:consistency}(i) imposes standard moment conditions.
Assumption~\ref{asm:consistency}(ii) requires full column rank of each $\X_i$. This ensures that each individual likelihood distinguishes between different values of $\hc_i$, so that the individual MLEs are well-defined.

Our proof of the consistency of the \NPM\ builds upon \citet[KW hereafter]{kiefer1956consistency}.
We verify their Assumptions~1--5, with their metric on the space of mixing distributions replaced by the vague metric $d$ defined below.
Distinct from the issue of identification, the multivariate setting presents a topological complexity not found in the scalar case: probability mass can diverge along infinitely many directions.
A primary technical challenge therefore lies in finding a metric that renders the space $\mathcal{G}$ totally bounded (thereby ensuring the compactness of its completion) while enabling the comparison of distributions with such diverging mass.

To this end, the vague topology is the natural choice. We adopt the canonical metrization of this topology, defined by
\begin{equation*}
d(G_0, G_1) = \sum_{r=1}^\infty \frac{1}{2^r} \left| \int_{\Theta} h_r(\theta)dG_0(\theta) - \int_{\Theta} h_r(\theta)dG_1(\theta) \right|,
\end{equation*}
where $\{h_r\}_{r=1}^\infty$ is a dense sequence in the unit ball of $C_c(\Theta)$ with respect to the supremum norm $\|h\|_{\infty}=\sup_{\theta \in \Theta} |h(\theta)|$.\footnote{Here, $C_c(\Theta)$ denotes the space of continuous real-valued functions on $\Theta$ with compact support.}
Since the test functions $h_r$ vanish at infinity, this metric treats any mass that drifts to infinity as if it were collapsed to a single ``point at infinity,'' irrespective of the direction of divergence.
This choice of metric is compatible with the model because the likelihood cannot distinguish between different parameter values at infinity (i.e., the likelihood vanishes as parameters diverge).
In effect, the metric compactifies the domain by collapsing all unbounded directions into a single point, mirroring the likelihood's asymptotic behavior.
Finally, we note that convergence in the vague topology coincides with standard weak convergence whenever the limiting distribution is itself a probability measure.

The next theorem presents the almost sure consistency of the \NPM\ with respect to the weak metric $d$.

\begin{thm}
\label{thm:consistency}
Let Assumptions~\ref{asm:DGP}, \ref{asm:id-conditions}, and \ref{asm:consistency} hold.
Then, $d(\hat{G}, G_*) \asto 0$ as $N \to \infty$.
\end{thm}

Alternative metrics are discussed in the \NPM\ literature, though they serve theoretical roles distinct from our metric $d$.
The first is the Hellinger distance between the marginal densities, $\frac{1}{2}\int (\sqrt{f_{\hat{G}}(y)} - \sqrt{f_{G_*}(y)})^2 dy$.
While Hellinger distance is frequently considered in mixture models (\citeay{jiangGeneralMaximumLikelihood2009}; \citeay{Bodhi:JRSSB}), it measures proximity in the data space rather than the parameter space. In our setting, the convergence of the fitted marginal density $f_{\hat{G}}$ to the truth $f_{G_*}$ is of secondary importance.
Another alternative is the Wasserstein distance (e.g., $W_2$ as in \citeay{Bodhi:JRSSB}).
However, $W_2$ requires the convergence of second moments and consequently becomes unbounded if probability mass escapes to infinity. This renders it unsuitable for our goal of establishing consistency through compactification.
The metric $d$ allows us to apply the Kiefer-Wolfowitz machinery. It strikes a necessary balance: it induces a topology sufficiently weak to ensure compactness (accommodating mass at infinity), yet strong enough, together with Assumption~\ref{asm:EB}, to imply the consistency of downstream EB estimation.

The consistency of the \NPM\ extends to the constrained \NPM, defined as
\begin{equation*}
\tilde{G} \in \argmax_{G \in \tilde{\mathcal{G}}} F_N(G),
\end{equation*}
where $\tilde{\mathcal{G}} \subseteq \mathcal{G}$ is a subset that contains $G_*$.
For instance, \citet{GK:2017} employ a profile likelihood approach to estimate a dynamic panel model similar to \eqref{eq:simple-dynamic-panel} under the restriction $\rho_i = \rho_*$. In our framework, imposing homogeneity ($\rho_i \equiv \rho$) corresponds to optimizing over a restricted subspace $\tilde{\mathcal{G}}$.

The estimator $\hat{G}$ allows for the recovery of various features of the data-generating process that can be expressed as continuous functionals. 
In particular, consistent estimators of prior and posterior moments can be obtained from the \NPM\ under suitable regularity conditions.
For example, the covariance matrix $\var_{G_*}(\theta_i)$ of the true prior can be estimated by $\var_{\hat{G}}(\theta_i)$, provided that $\|\theta\|^2$ is uniformly integrable with respect to the sequence of estimated priors.
See the discussion following Assumption~\ref{asm:EB} for sufficient conditions ensuring such uniform integrability.

\section{Regret Consistency and $\ell_p$ Convergence}
\label{subsec:EB-consistency}

Let $(\Tau_i)_{i=1}^N = (\Tau(\theta_i))_{i=1}^N$ denote the unit-level latent parameters of interest, each defined as a function of $\theta_i$.\footnote{We focus on the case of a scalar parameter here. However, the discussion and results in this section readily extend to the vector-valued case.}
The estimation problem consists of constructing a compound decision rule $(Y_i, \X_i)_{i=1}^N \mapsto (\hat{\Tau}_i)_{i=1}^N$. G-modeling yields the estimator:
\begin{equation*}
\taueb_i := \mathbb{E}_{\hat{G}}[\Tau_i \mid Y_i, \X_i] \equiv \frac{\int_{\Theta}\Tau(\theta) \ell(Y_i \mid \X_i,\theta) d\hat{G}(\theta)}{f_{\hat{G}}(Y_i,\X_i)}.
\end{equation*}
This section establishes asymptotic guarantees for $\taueb_i$ in terms of regret 
$R(\taueb, \Tau) - R(\taupm, \Tau)$,
where   $R(\hat{\Tau}, \Tau) := \mathbb{E}_{G_*}[L(\hat{\Tau},\Tau)]$ is compound risk. Specifically, we consider  quadratic loss $L(\hat{\Tau}, \Tau) :=\frac{1}{N}\sum_{i=1}^N (\hat\tau_i-\tau_i)^2$, and hence the \emph{oracle} posterior mean is defined by:
\begin{equation}
\label{eq:PM-tau}
\taupm_i := \mathbb{E}_{G_*}[\Tau_i \mid (Y_i,\X_i)_{i=1}^N] = \mathbb{E}_{G_*}[\Tau_i \mid Y_i, \X_i], \quad i=1,\ldots, N,
\end{equation}
where the second equality follows from the independence across units. Write
\begin{equation*}
R(\taueb, \Tau) - R(\taupm, \Tau)
= \mathbb{E}_{G_*}\left[\frac{1}{N} \sum_{i=1}^N (\taueb_i-\Tau_i)^2 - (\taupm_i-\Tau_i)^2\right] = \mathbb{E}\left[ \frac{1}{N} \sum_{i=1}^N (\taueb_i - \taupm_i)^2\right] \ge 0,
\end{equation*}
since $\mathbb{E}_{G_*}[(\taueb_i - \taupm_i)(\taupm_i - \Tau_i)] = 0$ by the properties of conditional expectation.
The condition $\lim_{N \to \infty} [R(\taueb, \Tau) - R(\taupm, \Tau)] = 0$ is referred to as \emph{regret consistency} (see, e.g., \citet{AGT:EBAdaptive}) or \emph{asymptotic optimality} (see, e.g., \citet{Zhang:1997}).

Consider the one-period-ahead prediction of $Y_{i T+1}$, which depends on $\theta_i$ and the data.
The oracle predictor $\taupm_i$ in the HIVDX model is given by
\begin{align}\label{eq:forecast-formula}
\taupm_i &= \mathbb{E}_{G_*}[a_i + b_i X_{2,i T+1} + \rho_i Y_{iT} \mid Y_{i}, \X_{i}, X_{2,i T+1}]  - \mathbb{E}_{G_*}[\rho_i a_i + \rho_i b_i X_{2,iT} \mid Y_{i}, \X_{i}].
\end{align}
The EB one-period-ahead prediction $\taueb_i$ is then obtained by substituting the \NPM\ for $G_*$ in \eqref{eq:forecast-formula}.  
Evaluating the EB method in this setting requires controlling the compound error arising from the interaction between the estimation discrepancy $\hat{\theta}^{\mathrm{EB}}_i - \thetapm_i$ and these observed variables.
Consequently, the analysis necessitates controlling the $\ell_p$ distance (where $p \ge 2$, with $p=4$ being a particularly important case) between the EB and oracle estimates.

\begin{assumption}
\label{asm:EB}
Assume that $\mathbb{E}_{G_*}[|\Tau_i|^p] < \infty$ for some $p \in [2, \infty)$. Moreover, the following conditions hold.
\begin{enumerate}
    \item [(i)] The function $\Tau(\theta)$ is continuous in $\theta \in \Theta$ and, for all $c > 0$, satisfies
    \begin{equation*}
        \lim_{\|\hc\|\to\infty} \exp(-c \|\hc\|^2) |\Tau(\theta)| = 0.
    \end{equation*}
    \item [(ii)] For all $\epsilon > 0$, there exists $M < \infty$ such that
    \begin{align*}
        \limsup_{N\to\infty} \mathbb{E} \left[ \int_{\Theta} |\Tau(\theta)|^p \mathbf{1}\{|\Tau(\theta)| \ge M\} d\hat{G}_N(\theta) \right] &\le \epsilon,
    \end{align*}
    where $\hat{G}_N$ is the \NPM\ defined in \eqref{eq:npmle-sample-program}, with $N$ denoting the sample size.
\end{enumerate}
\end{assumption}

Assumption~\ref{asm:EB}(i) requires that $\Tau(\theta)$ depends continuously on $\theta$ and grows slower than any quadratic exponential function of $\|\hc\|$.
This ensures that $|\taueb_i - \taupm_i| \to 0$ for each unit $i$ as $d(\hat{G}_N, G_*) \to 0$. The continuity requirement can be relaxed to allow $\Tau(\theta)$ to be discontinuous on a set where $G_*$ places zero mass.

Assumption~\ref{asm:EB}(ii) imposes a high-level condition on the \NPM\ sequence, requiring $|\Tau_i|^p$ to be uniformly integrable with respect to the sequence of estimated priors.
This condition is automatically satisfied when $\Tau(\theta)$ is bounded on $\Theta$ (e.g., $\Tau(\theta) = \delta$) but can be nontrivial for unbounded $\Tau_i$ such as slope parameters.
For such cases, Assumption~\ref{asm:EB}(ii) requires specific moment conditions on $(Y_i,X_i)$;
for example, the $p$-th moment of $\hc_i$ in the HIVDX model is bounded by the $p$-th moment of $Y_i$, provided $\X_i$ satisfies certain regularity conditions.

\begin{thm}
\label{thm:EB-consistency}
Let Assumptions~\ref{asm:DGP}, \ref{asm:id-conditions}, \ref{asm:consistency}, and \ref{asm:EB} hold for some $p \in [2, \infty)$.
Then $(\taueb_i)_{i=1}^N$ is regret-consistent, i.e.,
\begin{equation*}
    R(\taueb, \Tau) - R(\taupm, \Tau) \to 0 \quad \text{as } N \to \infty.
\end{equation*}
Moreover, the EB estimators converge in $\ell_p$-norm:
\begin{equation*}
    \mathbb{E}\left[\frac{1}{N}\sum_{i=1}^N |\taueb_i-\taupm_i|^p\right] \to 0 \quad \text{as } N \to \infty.
\end{equation*}
\end{thm}

Theorem~\ref{thm:EB-consistency} provides the theoretical justification for using $(\taueb_i)_{i=1}^N$ in large samples.
The primary challenge in establishing this result is that the weak convergence of the priors, $d(\hat{G}, G_*) \asto 0$, implies the convergence of integrals only for \emph{bounded} continuous functions.
However, $\Tau(\theta)$ and the squared-error loss are generally unbounded functions of the parameters.
Assumption~\ref{asm:EB}(ii) is critical because uniform integrability of the estimated prior moments ensures uniform integrability of the posterior moments across units.
This allows us to upgrade the weak convergence of $\hat{G}$ to the convergence of the posterior expectations (the EB estimators) in $\ell_p$-norm.
Notably, the second result of the theorem, that is, convergence in $\ell_p$-norm, is stronger than simple regret consistency ($p=2$) and is particularly useful when the latent parameters enter the decision problem nonlinearly or interact with unbounded covariates.

\subsection{Regret Consistency in the HIVDX Model}
\label{subsec:EB-estimation-HIVDX}

This subsection applies Theorem \ref{thm:EB-consistency} to the HIVDX model, focusing on several quantities of interest: the individual-level parameters $\theta_i$ and the one-period-ahead prediction. These examples provide the theoretical underpinnings for the empirical analysis presented in Section~\ref{sec:empirical}.

When $\Tau_i = \theta_i$ and $\theta_i=(\beta_i',\delta_i')'$ in the HIVDX model, the EB estimator of $\theta_i$ can be partitioned as $\hat{\theta}^{\mathrm{EB}}_i = (\hat{\hc}^{\mathrm{EB}}_i, \hat{\delta}^{\mathrm{EB}}_i)$.
Regret consistency of $\hat{\delta}^{\mathrm{EB}}_i$ follows immediately from Theorem~\ref{thm:EB-consistency}, since Assumption~\ref{asm:EB} holds automatically for $\Tau(\theta) = \delta$, which is continuous and bounded on the compact set $\mathcal{K}_\delta$.
For $\hat{\hc}^{\mathrm{EB}}_i$, however, Assumption~\ref{asm:EB}(ii) requires $\|\hc_i\|^p$ to be uniformly integrable with respect to the sequence of estimated priors.
This condition is guaranteed by a moment condition on $(Y_i,X_i)$.

Once $(\theta_i)_{i=1}^N$ have been estimated, we turn to the one-period-ahead EB prediction $\hat{Y}^{\mathrm{EB}}_{i T+1}$ of $Y_{i T+1}$. Given the oracle estimator in \eqref{eq:forecast-formula}, 
 the difference between the EB and oracle predictions involves interaction of the errors in  estimating the random coefficients and $(Y_i,X_i)$.

 Let $\bar{X}_{2,i} = T^{-1}\sum_{t=1}^T X_{2,it}$ denote the time average of $X_{2,it}$ for unit $i$.

\begin{prop}
\label{prop:EB-HIVDX}
Let Assumptions~\ref{asm:DGP}--\ref{asm:consistency} hold.
Assume that $|\bar{X}_{2,i}| \le M$ and $\sum_{t=1}^T (X_{2,it} - \bar{X}_{2,i})^2 \ge c$ for some constants $M, c > 0$ and for all $i \ge 1$. 
\begin{enumerate}[leftmargin=0.3in]

\item [(i)] $(\sigmaeb_i)_{i=1}^N$ and $(\rhoeb_i)_{i=1}^N$ are regret-consistent.

\item [(ii)] In addition, assume that $\mathbb{E}[\|Y_i\|^{2+\varepsilon}] < \infty$ for some $\varepsilon > 0$. Then $(\hat{a}^{\EB}_i)_{i=1}^N$ and $(\hat{b}^{\EB}_i)_{i=1}^N$ are regret-consistent.

\item [(iii)] Assume further that the HIVDX model and Assumption~\ref{asm:DGP} hold with $T$ replaced by $T+1$, $\mathbb{E}[\|Y_i\|^{4+\varepsilon}] < \infty$ for some $\varepsilon > 0$,
$\mathbb{E}[|X_{2,iT}|^{4}] < \infty$, and $\mathbb{E}[|X_{2,i T+1}|^{4}] < \infty$.
Then the predictions $(\hat{Y}^{\EB}_{i T+1})_{i=1}^N$ are regret-consistent.

\end{enumerate}
\end{prop}

Proposition~\ref{prop:EB-HIVDX}(iii) requires fourth (and slightly higher) moments of $(X_{2,iT}, X_{2,i T+1})$ and $Y_i$ to control the components of the prediction risk $\mathbb{E}[N^{-1} \sum_{i=1}^N (\hat{Y}^{\EB}_{i T+1} - Y^{\oracle}_{i T+1})^2]$.
For instance, applying the Cauchy-Schwarz inequality to the component involving the slope coefficient yields
\begin{equation*}
\mathbb{E}\left[ \frac{1}{N}\sum_{i=1}^N (\hat{b}^{\EB}_i-b^{\oracle}_i)^2 X_{2, i T+1}^2 \right] \le \sqrt{\mathbb{E}\left[ \frac{1}{N}\sum_{i=1}^N (\hat{b}^{\EB}_i-b^{\oracle}_i)^4\right]} \sqrt{\mathbb{E}\left[ \frac{1}{N}\sum_{i=1}^N X_{2, i T+1}^4\right]}.
\end{equation*}
The moment assumption on $Y_i$ ensures that the first term on the right-hand side vanishes asymptotically (by Theorem~\ref{thm:EB-consistency} with $p=4$), while the fourth moment of $X_{2, i T+1}$ ensures that the second term remains bounded.

\section{Computational Algorithms}
\label{sec:implementation}

A variety of computational methods have been proposed for implementing the \NPM. Although the optimization problem is convex, computation is challenging because the optimization variable $G$ is an infinite-dimensional object, rendering direct implementation intractable. A common remedy is to approximate $G$ by a discrete distribution supported on $m$ atoms (with $m \in \mathbb{N}$), thereby reducing the problem to fitting an $m$-component mixture model. This approach is justified by the well-known result that a discrete \NPM\ with at most $N_0$ support points exists (possibly among multiple maximizers), where $N_0$ is the number of distinct data points (\citeay{lindsay1983geometry}).\footnote{\citet{Polyanskiy-Wu:2020} show that the required number of components $m$ grows at the rate $O(\log N)$ in the one-dimensional case under subgaussian priors, and conjecture that in higher dimensions $m$ increases on the order of $O((\log N)^C)$, where $C$ depends on the dimension of $\theta_i$. For one-dimensional Gaussian location mixtures, \citet{Polyanskiy-Sellke:2025} establish computational guarantees, including an algorithm that computes the exact support size.}

Let $\pmb \theta = (\theta^j)_{j=1}^m$ denote a grid of support points for $G$, and let $\pmb \omega = (\omega^j)_{j=1}^m$ be the associated weights satisfying $\omega^j \ge 0$ for all $j$ and $\sum_{j=1}^m \omega^j = 1$, and write $G^{\pmb \theta, \pmb \omega}$ for the discrete distribution that places mass $\omega^j$ at $\theta^j$, $j=1,\ldots,m$.
Most EB estimators are solutions to the discretized problem of the form
\begin{equation}
\label{eq:numerical-objective}
\mathcal{L}(\pmb \theta, \pmb \omega) = 
\frac{1}{N} \sum_{i=1}^N \log \left( \sum_{j=1}^m \ell(Y_i \,|\, \X_i, \theta^j) \omega^j \right) 
- \lambda \left(\sum_{j=1}^m \omega^j - 1\right)
\end{equation}
For example, \citet{jiangGeneralMaximumLikelihood2009} introduce a fixed-point EM algorithm that attains a global optimum up to an explicit error bound.
\citet{GK:2017} estimate a dynamic panel model with two-dimensional heterogeneity using the interior-point algorithm of \citet{koenkerConvexOptimizationShape2014} to fit weights on a fixed grid.\footnote{
More recently, \citet{KCSA:2020}, \citet{ZCST:2024}, and \citet{WIM:2025} propose alternative methods based on general-purpose convex programming to improve scalability.}

A feature shared by all these methods is that the location of the support points is fixed, so optimization is performed solely with respect to the weights $\pmb \omega$.
Unlike in simple mixture models, determining where to position the grid points is not obvious when the dimension of $\theta$ exceeds two. It would be desirable to allow the support points to move adaptively. However,
this necessitates a shift from optimizing weights on a static grid to optimizing the probability distribution $G$ itself within the space of probability measures $\mathcal{P}(\Theta)$.
In this infinite-dimensional setting, the standard gradient descent used in Euclidean spaces needs to be replaced by \emph{gradient flow}, whose trajectory depends critically on the metric imposed on $\mathcal{P}(\Theta)$; different metrics can generate markedly different trajectories.

Consider the gradient flow induced by two commonly used metrics: the  Fisher-Rao (FR) and the  Wasserstein (W).   Intuitively, the FR flow increases probability mass at locations with high fitness but cannot create mass at new locations or move existing support points.
In contrast, the W gradient flow evolves according to the continuity (transport) equation
that effectively moves the support points across the parameter space $\Theta$.
  \citet{yan2024learning} show for the Gaussian location mixture case that if the FR flow initialized with full support converges, its limit is a global maximizer, while W may be stuck at a local maximum. To also exploit the gains from transporting mass across space, 
\citet{yan2024learning}
consider a discrete-time \WFR\ (WFR) gradient flow algorithm  that combines both reweighting of FR and transport dynamics of W when updating $G_t^{\wfr}$:
\[
{\frac{G_{t+ \Delta t}^{\wfr} - G_{t}^{\wfr}}{\Delta t} \approx - \mathrm{div}(G_{t}^{\wfr} V_t) + G_{t}^{\wfr}\big(\alpha_t - \mathbb{E}_{G_t^{\wfr}}[\alpha_t]\big)},
\]
where $\alpha_t(\theta) := N^{-1}\sum_{i=1}^N \ell(Y_i \,|\, \X_i, \theta)/f_{G_t^{\wfr}}(Y_i, \X_i)$ is the first variation of the log-likelihood and $V_t := \nabla_\theta \alpha_t$.
The algorithm  adaptively relocates support points so that  $m$ remains moderate and need not exceed the sample size $N$.  It is designed to combine the global-convergence property of the FR flow with the adaptivity of W through mass transport. Again in a Gaussian location mixture setting,  WFR is shown to be superior to the EM algorithm, and gradient-descent methods based solely on the FR or W flow. 

\subsection{WFR Algorithm for the HC Model}\label{subsec:alg:wfr}

Motivated by these encouraging results,  we extend  the WFR from a Gaussian location mixture to the HC model.  
The WFR gradient-ascent algorithm updates $(\pmb \theta_n, \pmb \omega_n)_{n \ge 0}$ by alternately reweighting and moving the mass along the steepest ascent path.  
Let $\eta \in (0,1]$ denote the step size (learning rate), and define the posterior weights
\begin{equation*}
\pi_{ij}(\pmb \theta, \pmb \omega) 
= \frac{\ell(Y_i \,|\, \X_i, \theta^{j}) \, \omega^{j}}{\sum_{k=1}^m \ell(Y_i \,|\, \X_i, \theta^{k}) \, \omega^{k}},
\qquad i=1,\ldots,N,\ j=1,\ldots,m,
\end{equation*}
which correspond to the E-step in the EM algorithm.  
Holding $\pmb \theta_n$ fixed, the algorithm first updates the weights $\pmb \omega_n$ via
\begin{align}
\label{eq:FR-step}
\begin{aligned}
\omega_{n+1}^j 
&= \omega_{n}^j \left(1 + \eta \left[ \frac{1}{N} \sum_{i=1}^N \frac{\ell(Y_i \,|\, \X_i, \theta_n^j)}{\sum_{k=1}^m \ell(Y_i \,|\, \X_i, \theta_n^{k}) \, \omega_n^{k}} - 1 \right]\right) \nonumber \\
&= (1-\eta) \omega_n^j +  \frac{\eta}{N} \sum_{i=1}^N \pi_{ij}(\pmb \theta_n, \pmb \omega_n),
\qquad j=1,\ldots,m,
\end{aligned} \tag{FR-step}
\end{align}
which we refer to as the FR-step.  
The reweighting function is chosen as the gradient of the objective function in \eqref{eq:numerical-objective} with respect to $\omega^j$, thereby ensuring ascent.\footnote{The first-order conditions with respect to $\pmb \omega$ imply that the optimal Lagrange multiplier is fixed at $\lambda = 1$: multiplying the condition for $\omega^j$ by $\omega^j$ and summing over $j$ yields $\lambda = 1$. The bracketed term in \eqref{eq:FR-step} is thus the gradient evaluated at $\lambda = 1$.}
Next, with $\pmb \omega_{n+1}$ fixed, the grid $\pmb \theta_n$ is updated according to
\begin{align}
\label{eq:W-step}
\begin{aligned}
\tilde{\theta}_{n+1}^j 
&= \theta_n^j + \frac{\eta}{N} \sum_{i=1}^N \pi_{ij}(\pmb \theta_n, \pmb \omega_{n+1}) 
\frac{\partial}{\partial \theta_n^j} \log \ell(Y_i \,|\, \X_i, \theta_n^j),  \\
\theta_{n+1}^j 
&= P_{\Theta}(\tilde{\theta}_{n+1}^j), 
\qquad j=1,\ldots,m,
\end{aligned}
\tag{W-step}
\end{align}
where $P_{\Theta}(\cdot)$ denotes the metric projection onto $\Theta$, ensuring that $\theta_{n+1}^j \in \Theta$.\footnote{Formally, for any $\tilde{\theta} \in \mathbb{R}^{d_\theta}$, the metric projection is
$
P_\Theta(\tilde{\theta}) := \argmin_{\theta \in \Theta} \|\tilde{\theta} - \theta\|.
$
It yields the closest point in $\Theta$ to $\tilde{\theta}$ under the Euclidean norm. For example, if $\Theta = [a,b]^{d_\theta}$ is a box constraint, then $P_\Theta(\tilde{\theta})$ simply clips each coordinate of $\tilde{\theta}$ to lie within $[a,b]$.
This ensures that each updated grid point $\theta_{n+1}^j$ remains feasible, even if the gradient step temporarily leaves the parameter space $\Theta$.}  
This W-step moves each grid point in the direction of the posterior-weighted average score, corresponding to a discrete version of the W gradient flow in which the velocity of each support point is scaled by its weight $\omega^j$.
Starting from \(n = 0\), the FR- and W-steps are applied alternately until \(n = \overline{n}\), the maximum number of iterations, or until the following convergence criterion is met: namely,  
\[
\max_{j} \left\{ \frac{1}{N} \sum_{i=1}^N \frac{\pi_{ij}(\pmb{\theta}_n, \pmb{\omega}_n)}{\omega_n^j} - 1 \right\} \le \tol
\]
for some tolerance parameter \(\tol>0\).\footnote{\label{fn:npmle-kkt}A necessary and sufficient condition for \(\hat G\) to be \NPM\ is
$
\frac{1}{N} \sum_{i=1}^N
\frac{\ell(Y_i \,|\, X_i, \theta)}{f_{\hat G}(Y_i, X_i)} \le 1
\text{ for all } \theta \in \Theta.
$
}
The left-hand side is the largest value of the bracketed term in \eqref{eq:FR-step} over the current support points, so the criterion checks the condition in footnote~\ref{fn:npmle-kkt} at those points.
Let \(\hat n\) denote the stopping index (either the first \(n\) satisfying the convergence criterion or \(\overline n\) if it is not met). The resulting estimate \(\hat G = G^{\pmb{\theta}_{\hat n}, \pmb{\omega}_{\hat n}}\) is our numerical approximation to the \NPM.
 Implementation of the HC model is summarized in Algorithm~\ref{alg:iteration} in the Appendix.
When $\eta$ in our algorithm is set to $1$, the FR-step coincides with the weight update in the EM algorithm.  
For smaller values of $\eta$, the weights are adjusted more conservatively, leading to more stable updates than in EM.  
Similarly, the W-step can be viewed as moving $\pmb \theta_n$ incrementally along the gradient of the $Q_N$-function of the EM algorithm defined in Section~\ref{subsec:alg:compare}.
Because it is formulated as a gradient-based procedure, we expect the WFR algorithm to have better convergence properties than EM, which generally requires more careful initialization to reach the global maximum (\citeay{balakrishnanStatisticalGuaranteesEM2017}). A formal analysis of the WFR algorithm is beyond the scope of this paper. However, for it to be useful in large-scale problems, it should at least have properties comparable to those of EM. This is explored in Section~\ref{sec:empirical}.

\paragraph*{Starting values}

Choosing an initial grid that is well spread over $\Theta$ is essential for the WFR gradient-ascent algorithm to converge to a global optimum (\citeay{yan2024learning}, Theorem~2).
If, instead, all atoms are initialized at the same location with equal weights, i.e., $\theta_0^j = \tilde{\theta}$ and $\omega_0^j = 1/m$ for all $j=1,\ldots,m$, then the subsequent updates satisfy $\theta_n^1 = \cdots = \theta_n^m$ for every $n \ge 0$.
In this case, the distribution $G^{\pmb \theta_n, \pmb \omega_n}$ degenerates to a single point mass at the common location $\theta_n^1$, failing to capture any heterogeneity in the $\theta_i$’s and hence unable to recover the \NPM.

In practice, we generate a dispersed initialization by computing MLEs on random subsamples of the data.
The procedure is as follows.
Fix the subsample size $B$. If the time dimension $T$ is sufficient to ensure that individual MLEs are well-defined, one may choose $B=1$ to maximize the diversity of the grid points.
For $j=1,\ldots,m$:
\begin{itemize}
\item[(i)] Sample a set of indices $\mathcal{I}_j = \{i(1), \ldots, i(B)\}$ uniformly at random from $\{1,\ldots,N\}$.
\item[(ii)] Compute the subsample MLE using $\mathcal{I}_j$:
\[
\theta_0^j = \argmax_{\theta \in \Theta} \sum_{b=1}^B \log \ell(Y_{i(b)} \,|\, \X_{i(b)}, \theta)
\]
and assign a uniform weight $\omega_0^j = 1/m$.
\end{itemize}
The resulting collection $(\theta_0^j, \omega_0^j)_{j=1}^m$ serves as the initial grid and weights.

\section{Empirical Analysis of Income Dynamics}
\label{sec:empirical}

We begin by highlighting several closely related methodological contributions. 
\citet{Chamberlain:Hirano:1999} consider a parametric Bayesian approach that allows for heterogeneity in $\sigma_i^2$. 
\citet{geweke-keane-joe:00} also adopt a Bayesian framework with a finite mixture model for the composite errors $\sigma_i e_{it}$, while \citet{hirano-ecma:02} extends this framework using an infinite mixture model. 
\citet{GK:2017} propose an EB approach that is closest in spirit to ours. 
These four papers did not allow for random $(b_i, \rho_i)$. 
\citet{Giacomini-et-al:2025} develop an individual-weight shrinkage estimator that optimizes unit-level (rather than aggregate) accuracy by exploiting each individual's own past history, with feasible weights motivated by minimax regret.	
\citet{moon-schorfheide-zhang} consider random components in all $(a_i, b_i, \rho_i, \sigma_i^2)$ within a parametric Bayesian framework. 
Specifically, they adopt a spike-and-slab prior that accommodates either dense or sparse structures but assume independence across random components. 
In contrast, we adopt a nonparametric EB approach for the prior distribution and allow for dependence among the random components.
As discussed in the introduction, \citet{Lee:2026} also allows for heterogeneous intercepts, slopes and persistence in PSID earnings. His dynamic regression can be derived by quasi-differencing a model like ours, and he reports confidence intervals for the mean, variance and distribution function of $\rho_i$ based on partial identification.

In this section, we re-examine the application in \citet[GK hereafter]{GK:2017}, which analyzes log real earnings data studied in  \citet{meghir2004income}. The extract consists of  938 individuals who have continuous earnings records from age 25 onward in the Panel Study of Income Dynamics (PSID) for the period 1968--1993.
The panel is unbalanced with varying numbers of time periods \(T_i\).\footnote{The only minor difference between the current model and the original HIVDX model in \eqref{def:model:hivxd} is that \(T_i\) may vary across \(i\). The likelihood function can be easily modified to accommodate unbalanced panels.}
As in \citet{GK:2017}, \(Y_{it}\) denotes residuals obtained from year-specific regressions of log real earnings on a vector of covariates. 
Though the HIVDX model motivates GK's analysis, they assume $\rho_i$ is homogeneous and $b_i=0$ in the application. Furthermore, their likelihood, expressed in terms of sufficient statistics, is conditioned on the initial observation \(Y_{i1}\), which implicitly assumes that the random component is independent of \(Y_{i1}\). Thus, their marginal likelihood depends on $Y_{i1}$ only through $Y_{i2}-\rho_* Y_{i1}$.

Using GK's PSID sample extract and their definition of $Y_{it}$ (residualized incomes), we estimate the HIVDX model  
with $\theta_i := (a_i, b_i, \sigma_i^2, \rho_i)$ and $X_{2,it} := \mathrm{Exp}_{it}/10$, so that $b_i$ is scaled up by a factor of $10$, and
$\mathrm{Exp}_{it}$ is potential experience constructed as $\mathrm{Exp}_{it} \;:=\; \text{age} \;-\; \max\{\text{years of schooling},\,12\} \;-\; 6$.\footnote{An equivalent form  used in the literature is $\mathrm{Exp}_{it} = \min\{ \text{age} - \text{years of schooling} - 6,  \text{age}  - 18 \}$.}
In this application, Assumption~\ref{asm:DGP}(i) requires $\theta_i$ to be independent of the entire experience path $X_{2,i}$. Since experience increases by one each year, this amounts to independence of $\theta_i$ from initial experience.
Furthermore,  our analysis differs from GK in three other respects.
First, we model $Y_{i1}$ under  a stationarity assumption without assuming independence between \(Y_{i1}\) and $\theta_i$.
Second,  the likelihood is directly evaluated from the observed data since sufficient statistics are not available in the presence of $X_{2,it}$ and $\rho_i$. Third,
we fit the model and estimate the prior distribution of $\theta_i$ using the WFR algorithm described in Section~\ref{subsec:alg:wfr}. In empirical analyses, we used $\overline{n} = 100,000$ as the maximum number of iterations and did not implement the stopping criterion.
We use the initialization procedure with $B = 1$ and $m=500$, and a step size of $\eta = 0.0005$.

For completeness, we also computed estimates using the EM algorithm discussed in Section~\ref{subsec:alg:compare}, initializing it with the same starting values used for WFR (as described in Section~\ref{subsec:alg:wfr}).
We find that the EM estimates are very similar to the WFR solution.
While the EM algorithm attains a slightly higher likelihood value, the WFR algorithm yields an approximate solution with fewer mass points.
Furthermore, although the EM algorithm generally requires fewer iterations to converge, the computational cost per iteration is higher.
These results are encouraging, as they suggest WFR is well-suited for higher-dimensional problems where fixed-grid algorithms become computationally challenging.

{
\begin{table}[h]
   \normalfont
   \centering
   \caption{First and Second Moments of the Estimated Prior $\hat{G}$}
   \label{tab:PSID:prior:moments}
   \medskip
   \input{Table_HIVDX.tex}
   \medskip
   {\begin{center}
   \parbox{5.5in}{\footnotesize{Notes: 
   The first four rows form a $4\times 4$ covariance matrix.
   The last row presents means.}}
   \end{center}
   }

\end{table}
}

Table~\ref{tab:PSID:prior:moments} shows the first and second moments from the estimated prior $\hat{G}$.
The marginal variances suggest that the coefficients are heterogeneous in all four dimensions. 
The variance of \(a_i\) is substantial, especially relative to the variance of \(Y_{it}\) across all units and time ($0.22$). Regarding \(b_i\), the variance of \(b_i X_{2,it}\) when \(X_{2,it} = x\) is approximately \(0.1 x^2\). The value of $x$ ranges over $[0.2, 3.2]$ within the sample, indicating notable heterogeneity in experience profiles. The variance of \(\sigma_i^2\) is 0.0229 and is non-negligible.
Finally, the distribution of $\rho_i$ is also spread out, with approximately $73\%$ of probability mass lying in $[0.2,0.8]$.
Because $\rho_i$ is unchanged by quasi-differencing, this estimate can be compared with \citet{Lee:2026}: although the samples differ, $\mathbb{E}_{\hat G}[\rho_i] = 0.422$ lies within his 95\% confidence interval $[0.253, 0.566]$ for the mean of $\rho_i$ in the heterogeneous-income-profile specification.

Turning to the covariances of the random coefficients,  the covariance between individual intercepts \((a_i)\) and variances \((\sigma_i^2)\) is negative, as in GK. However, allowing for additional sources of heterogeneity leads to some new findings.  We find a negative correlation between $\sigma_i^2$ and $b_i$, and  weak associations between $\rho_i$ with the other parameters ($a_i,b_i,\sigma_i^2$).
The most interesting result is the strong negative correlation between the individual intercepts 
($a_i$) and slopes ($b_i$). 

{
\begin{figure}[!htbp]
	\caption{Shrinkage Effects: Comparison of $(\hat{\theta}^{\mathrm{MLE}}_i)_{i=1}^N$ and $(\hat{\theta}^{\mathrm{EB}}_i)_{i=1}^N$}
	\label{fig:PSID:prior:dist}
    \begin{subfigure}[b]{0.5\textwidth}
        \centering
	\includegraphics[width=0.85\linewidth,keepaspectratio]{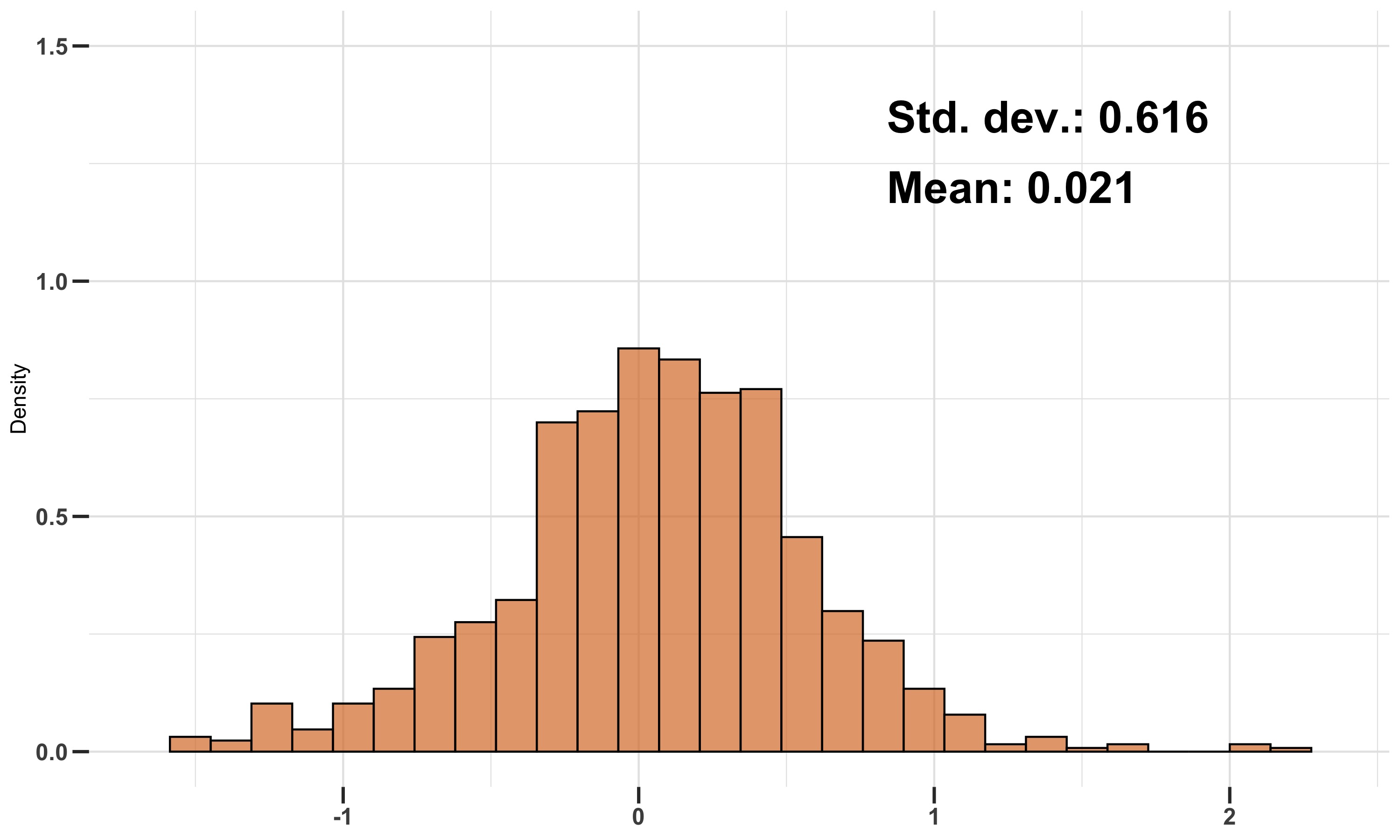} 
    \subcaption{MLEs of $a_i$.}
    \end{subfigure}
    \begin{subfigure}[b]{0.5\textwidth}
        \centering
        \includegraphics[width=0.85\linewidth,keepaspectratio]{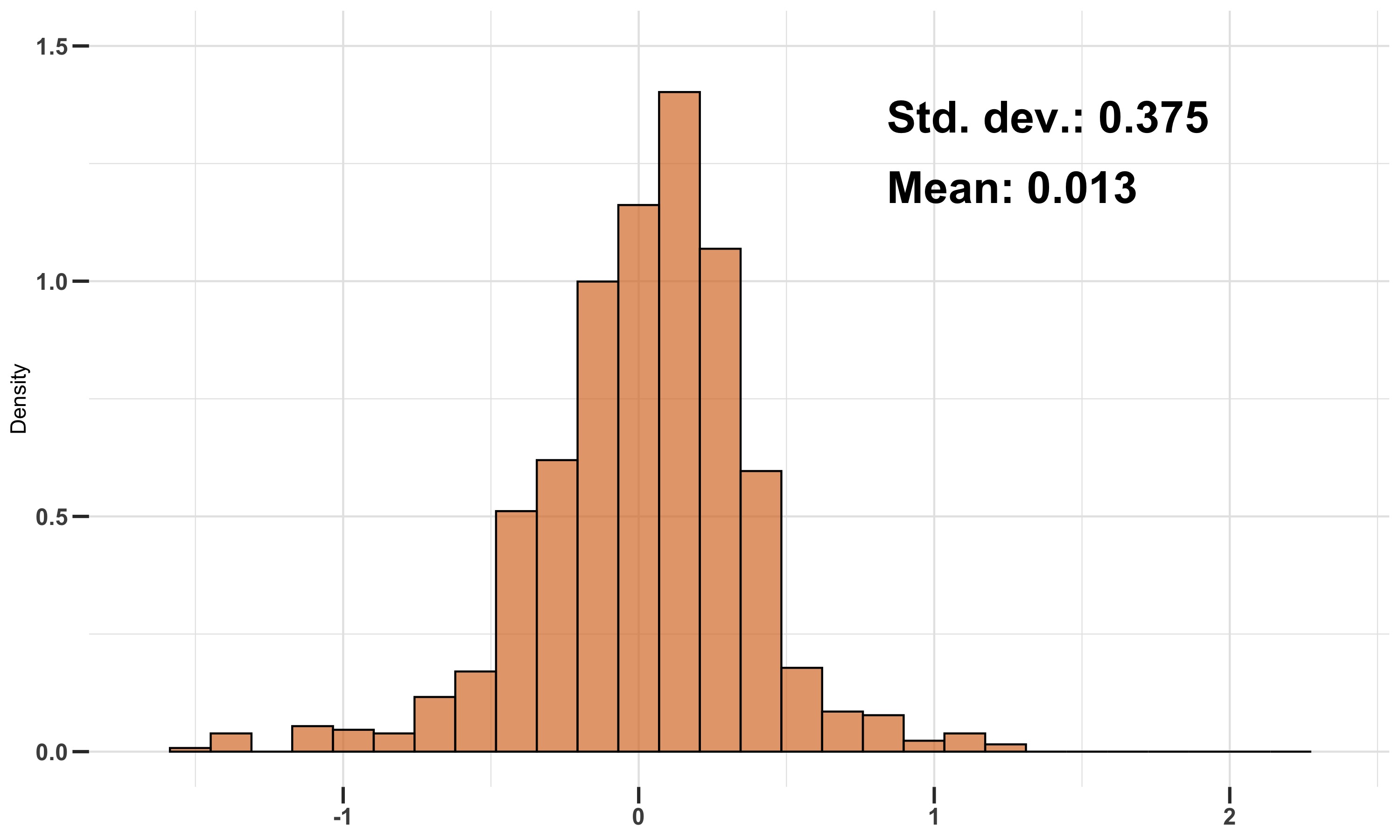}
    \subcaption{EB estimates of $a_i$.}
    \end{subfigure}

    \begin{subfigure}[b]{0.5\textwidth}
        \centering
	\includegraphics[width=0.85\linewidth,keepaspectratio]{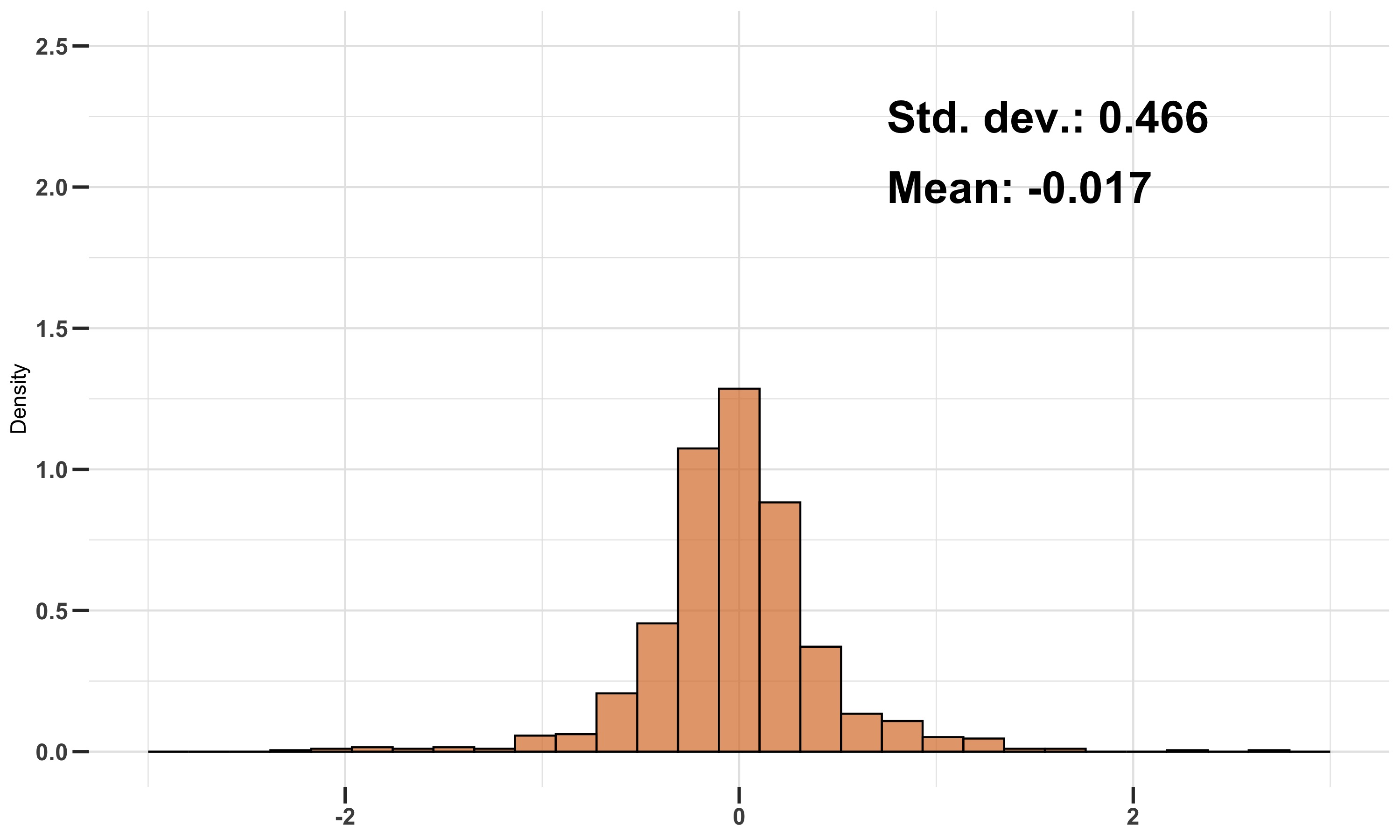} 
    \subcaption{MLEs of $b_i$.}
    \end{subfigure}
    \begin{subfigure}[b]{0.5\textwidth}
        \centering
        \includegraphics[width=0.85\linewidth,keepaspectratio]{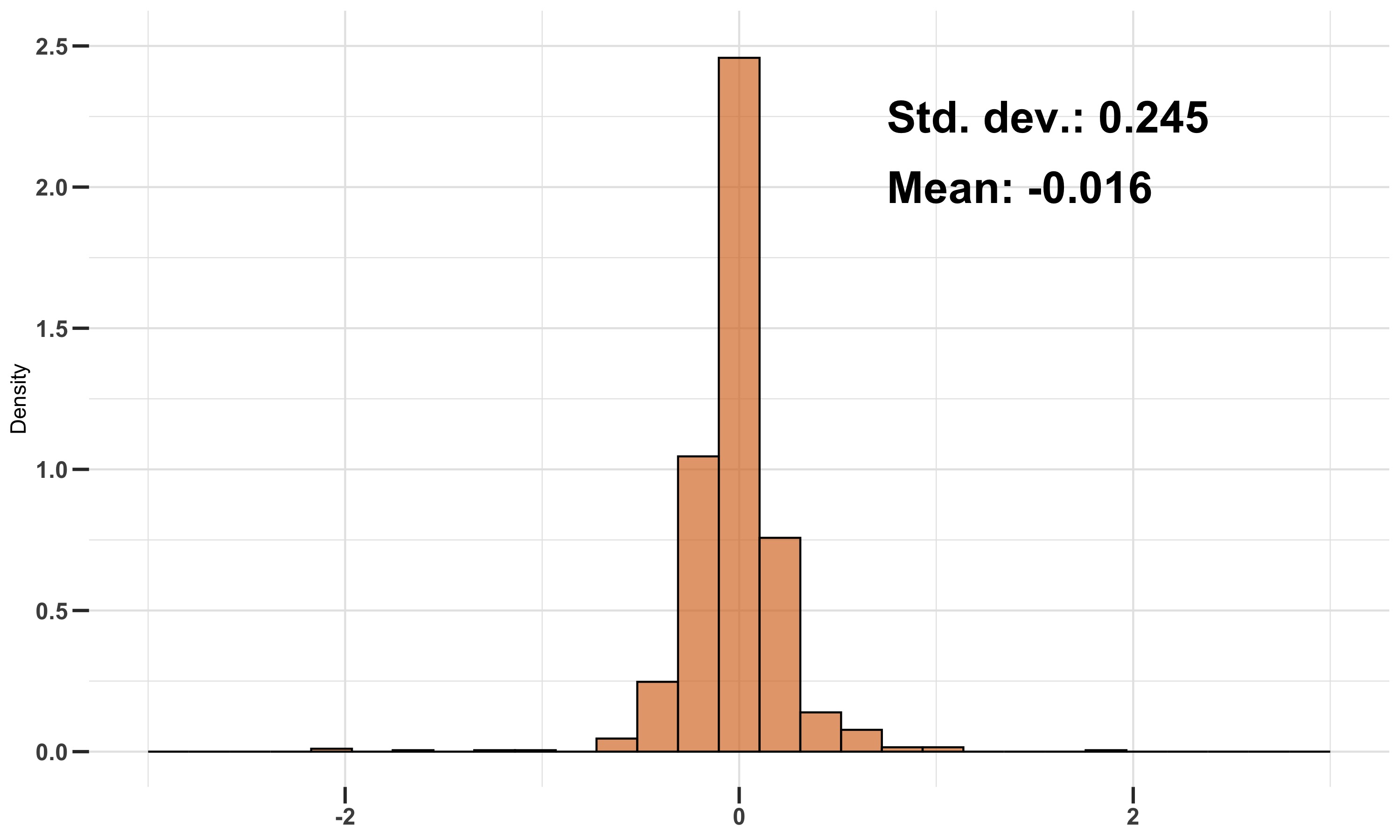}
    \subcaption{EB estimates of $b_i$.}
    \end{subfigure}

    \begin{subfigure}[b]{0.5\textwidth}
        \centering
	\includegraphics[width=0.85\linewidth,keepaspectratio]{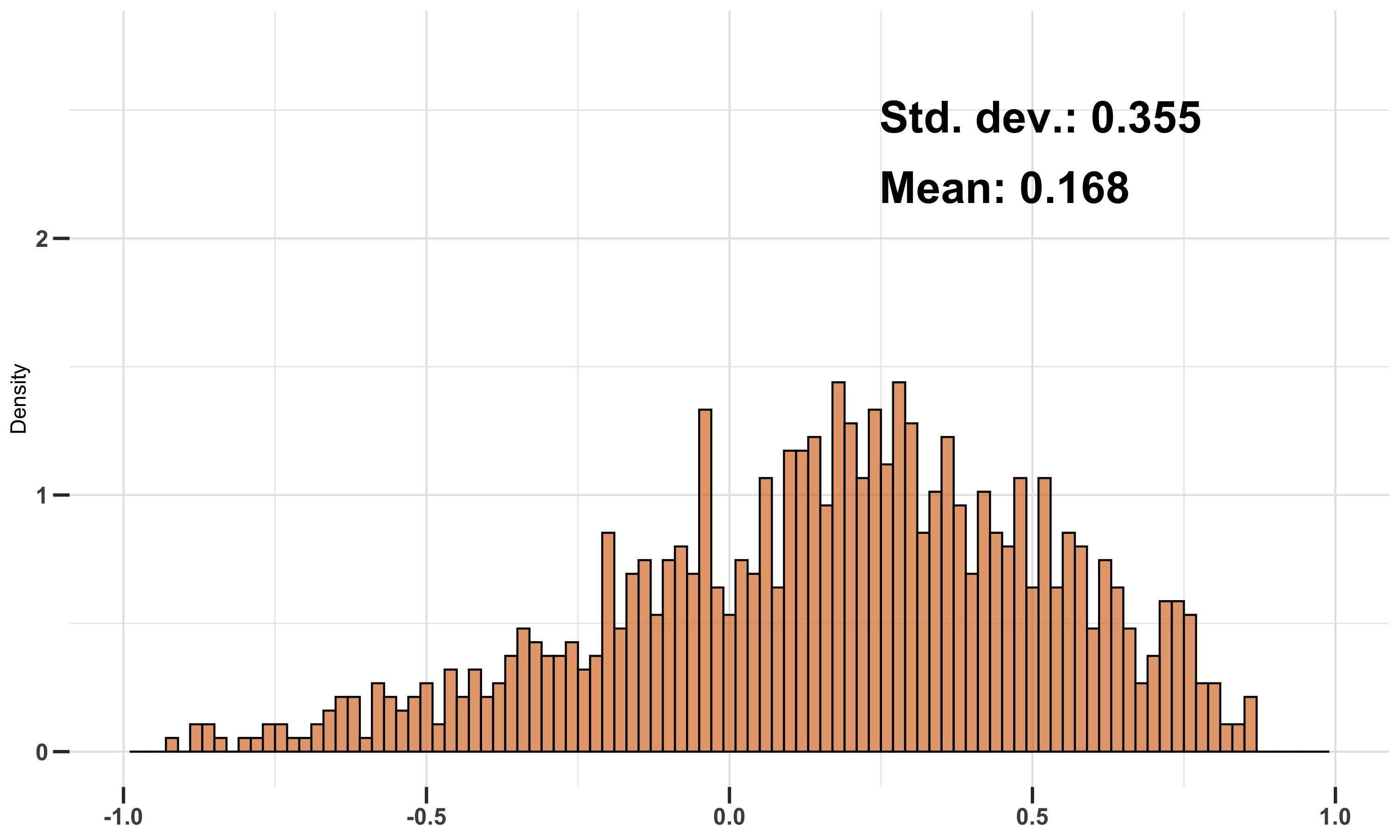} 
    \subcaption{MLEs of $\rho_i$.}
    \end{subfigure}
    \begin{subfigure}[b]{0.5\textwidth}
        \centering
        \includegraphics[width=0.85\linewidth,keepaspectratio]{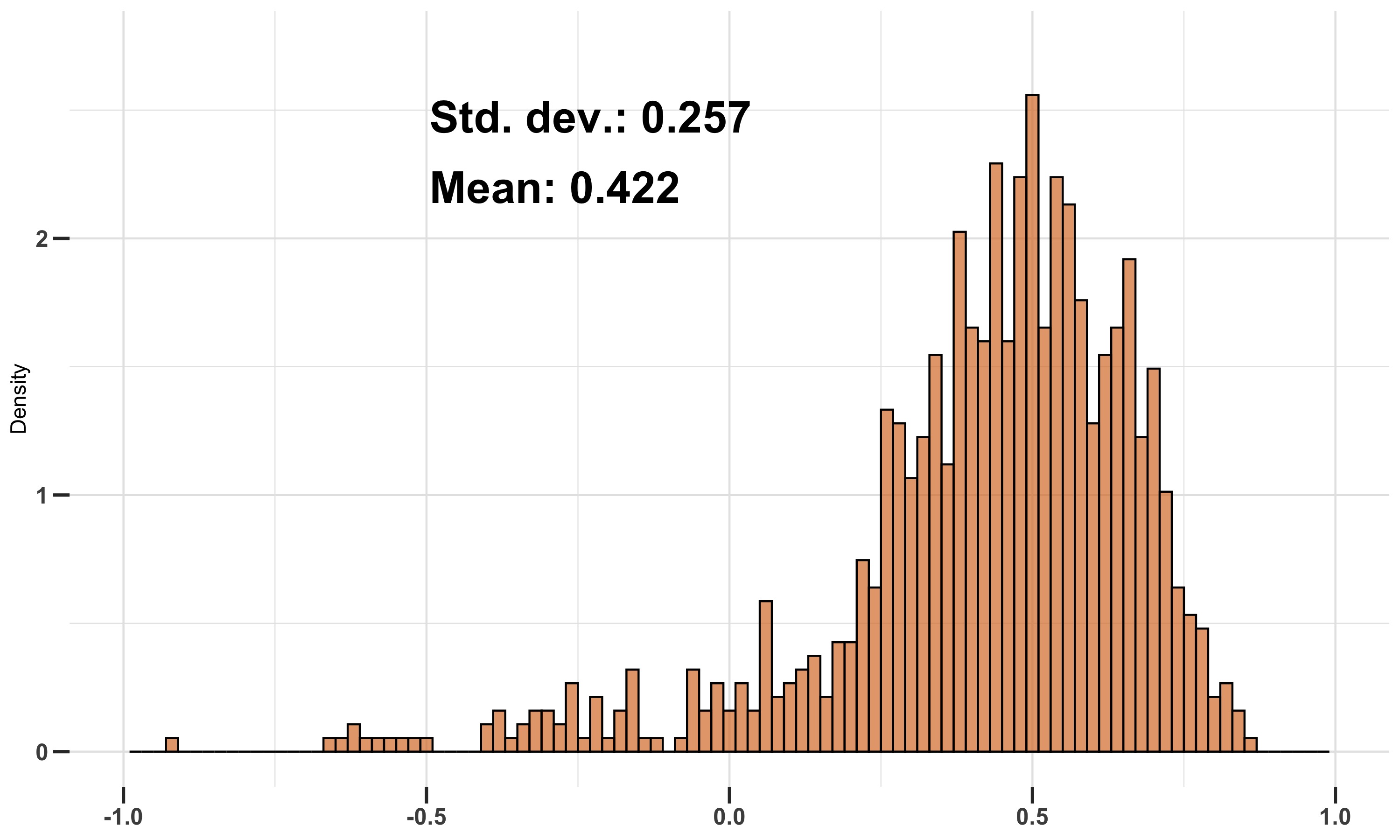}
    \subcaption{EB estimates of $\rho_i$.}
    \end{subfigure}
    
    \begin{subfigure}[b]{0.5\textwidth}
        \centering
        \includegraphics[width=0.85\linewidth,keepaspectratio]{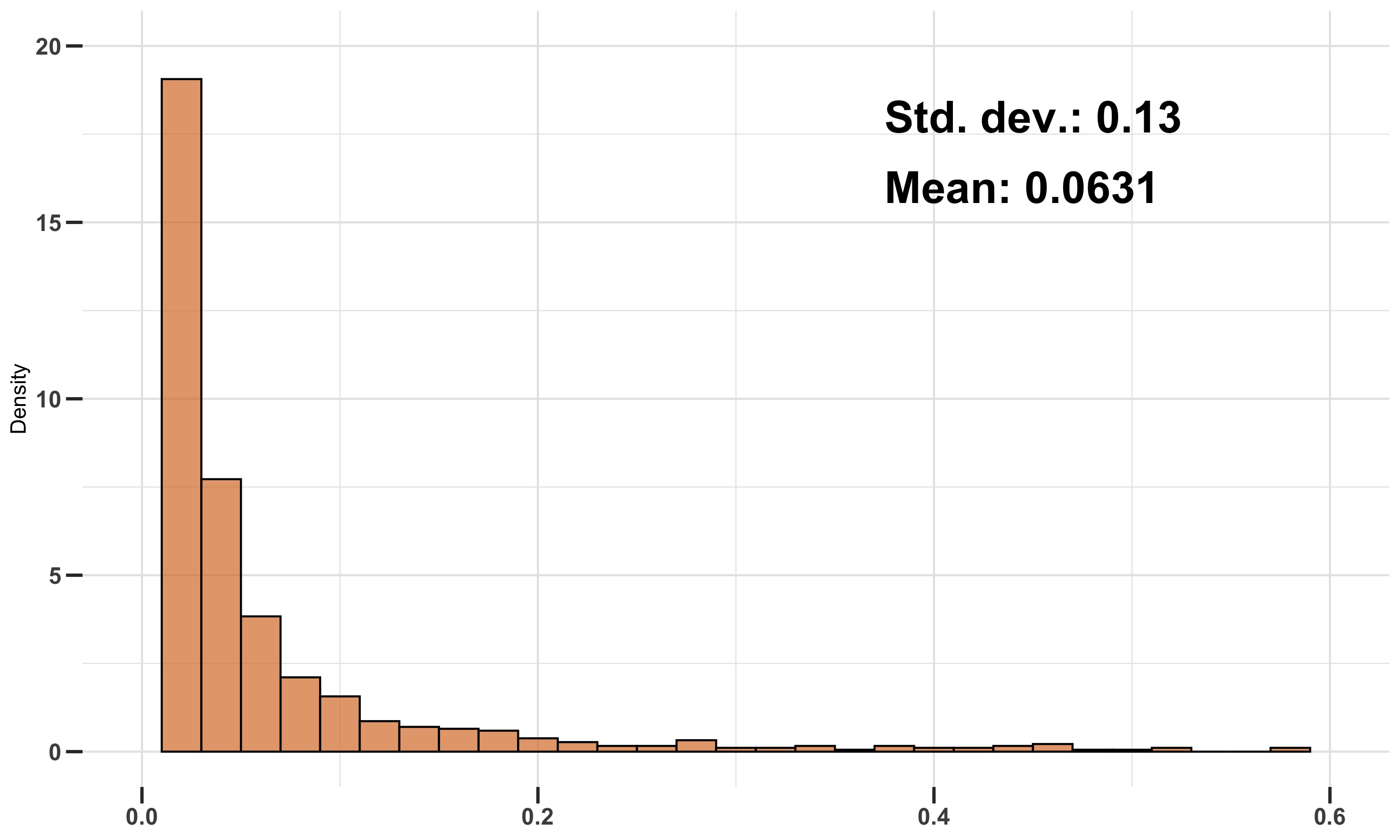}
    \subcaption{MLEs of $\sigma_i^2$.}
    \end{subfigure}
    \begin{subfigure}[b]{0.5\textwidth}
        \centering
        \includegraphics[width=0.85\linewidth,keepaspectratio]{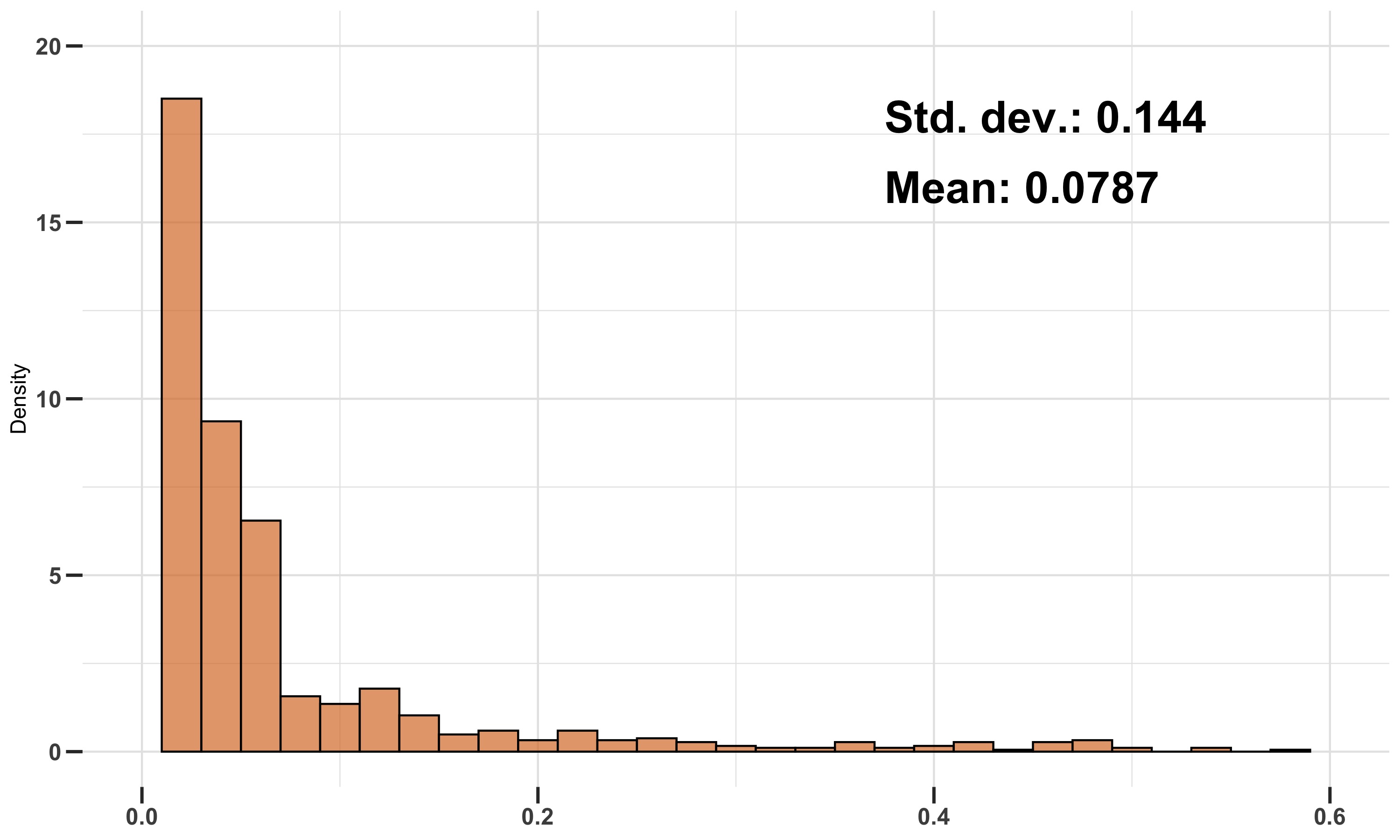}
    \subcaption{EB estimates of $\sigma_i^2$.}
    \end{subfigure}

    {\begin{center}
        \parbox{0.95\textwidth}{\footnotesize Note: The left and right panels display the histograms of the individual MLE and EB estimates for each component of $(\theta_i)_{i=1}^N$.}
    \end{center}
}
\end{figure}
}

A side-by-side comparison between the individual MLEs and EB estimates of $\theta_i$ is presented in Figure~\ref{fig:PSID:prior:dist}.
As we work with residual earnings, it is not surprising that the means of the individual MLEs and EB estimates of $a_i$ and $b_i$ are close to zero. A notable feature is that the individual MLEs of $\rho_i$ are highly dispersed, with a substantial fraction taking negative values, whereas the EB estimates are mostly positive and considerably more concentrated around their mean value of 0.422.
The cross-sectional variances of $\hat a_i^{\mathrm{EB}}$, $\hat b_i^{\mathrm{EB}}$, and $\hat \rho_i^{\mathrm{EB}}$ are reduced by approximately 63\%, 72\%, and 48\%, respectively, relative to the corresponding individual MLEs.
While no conspicuous shrinkage is observed for $\hat \sigma_i^{2,\mathrm{EB}}$, these reductions suggest substantial compound mean squared error gains from the EB approach.

Lastly, we consider one-period-ahead predictions of $(Y_{i,T_i})_{i=1}^N$ using a sample that holds out final-period observations.
Figure~\ref{fig:PSID:post:prec} compares the resulting predictions and prediction errors from the individual MLE and the EB methods.
The one-period-ahead prediction of $Y_{i,T_i}$ is computed as
\begin{align}\label{eq:forecast-formula:est}
\hat Y_{i,T_i} &= \hat a_i + X_{2,iT_i} \hat b_i + \hat \rho_i Y_{iT_i-1} - \widehat {\rho_i a_i} - X_{2,iT_i-1} \widehat {\rho_i b_i},    
\end{align}
where the hats indicate either EB or individual-level MLE estimates.
As shown in Figure~\ref{fig:PSID:post:prec}, the EB method exhibits substantial shrinkage, reducing the variance of the predicted values by approximately 25\% compared to the MLE. This translates into improved forecasting performance, leading to a variance reduction of approximately 20\% in the prediction errors. Because the prediction errors are centered at zero (Figure~\ref{fig:PSID:post:prec}), this corresponds to a similar reduction in the mean squared prediction error.

\begin{figure}[h]
	\caption{Pseudo-out-of-sample predictions and prediction errors}
	\label{fig:PSID:post:prec}
    \begin{subfigure}[b]{0.45\textwidth}
        \centering
    \subcaption{Predictions: $\hat{Y}^{\mathrm{MLE}}_{i,T_i}$}
	\includegraphics[width=0.875\linewidth,keepaspectratio]{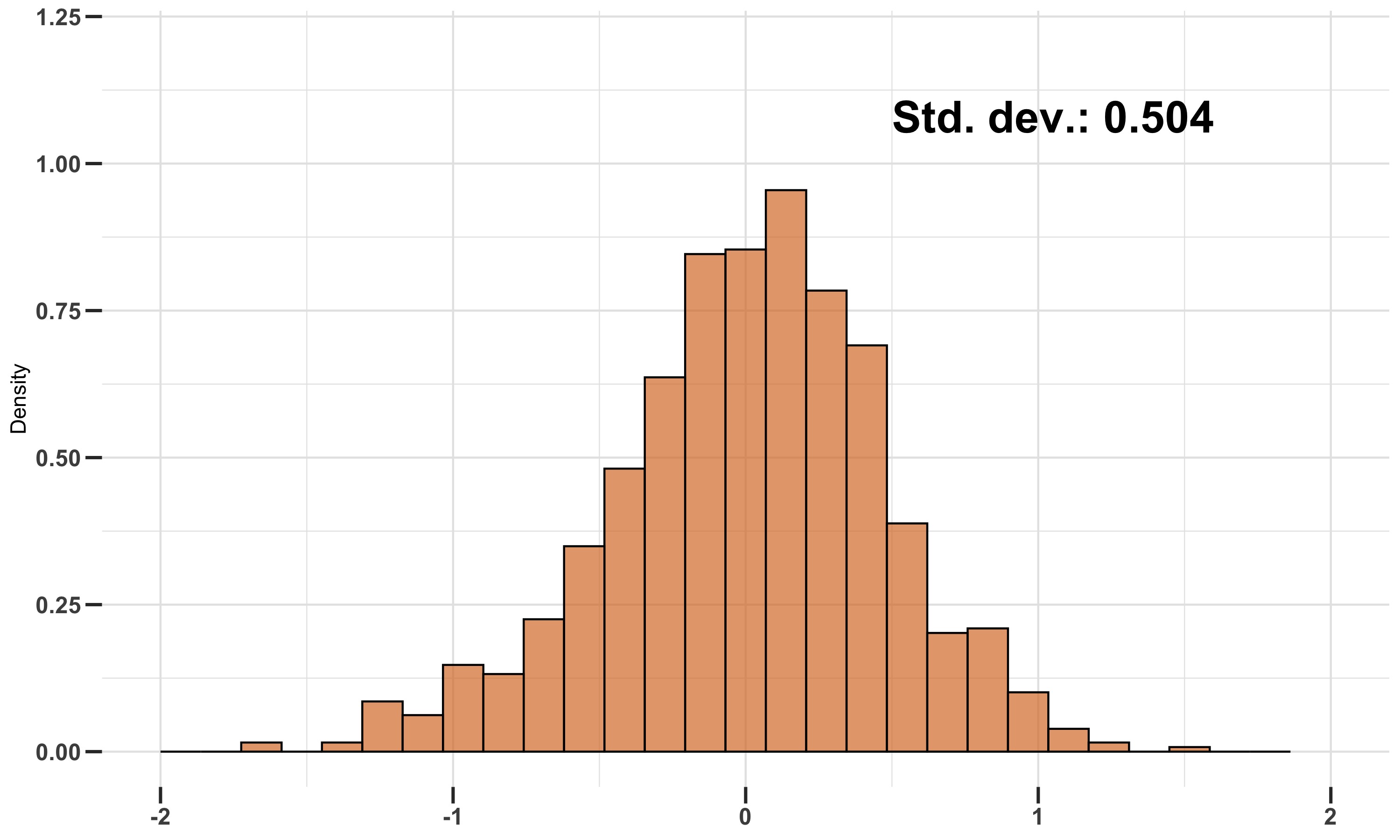} 
    \end{subfigure}
    \begin{subfigure}[b]{0.45\textwidth}
        \centering
    \subcaption{Predictions: $\hat{Y}^{\mathrm{EB}}_{i,T_i}$}
    \includegraphics[width=0.875\linewidth,keepaspectratio]{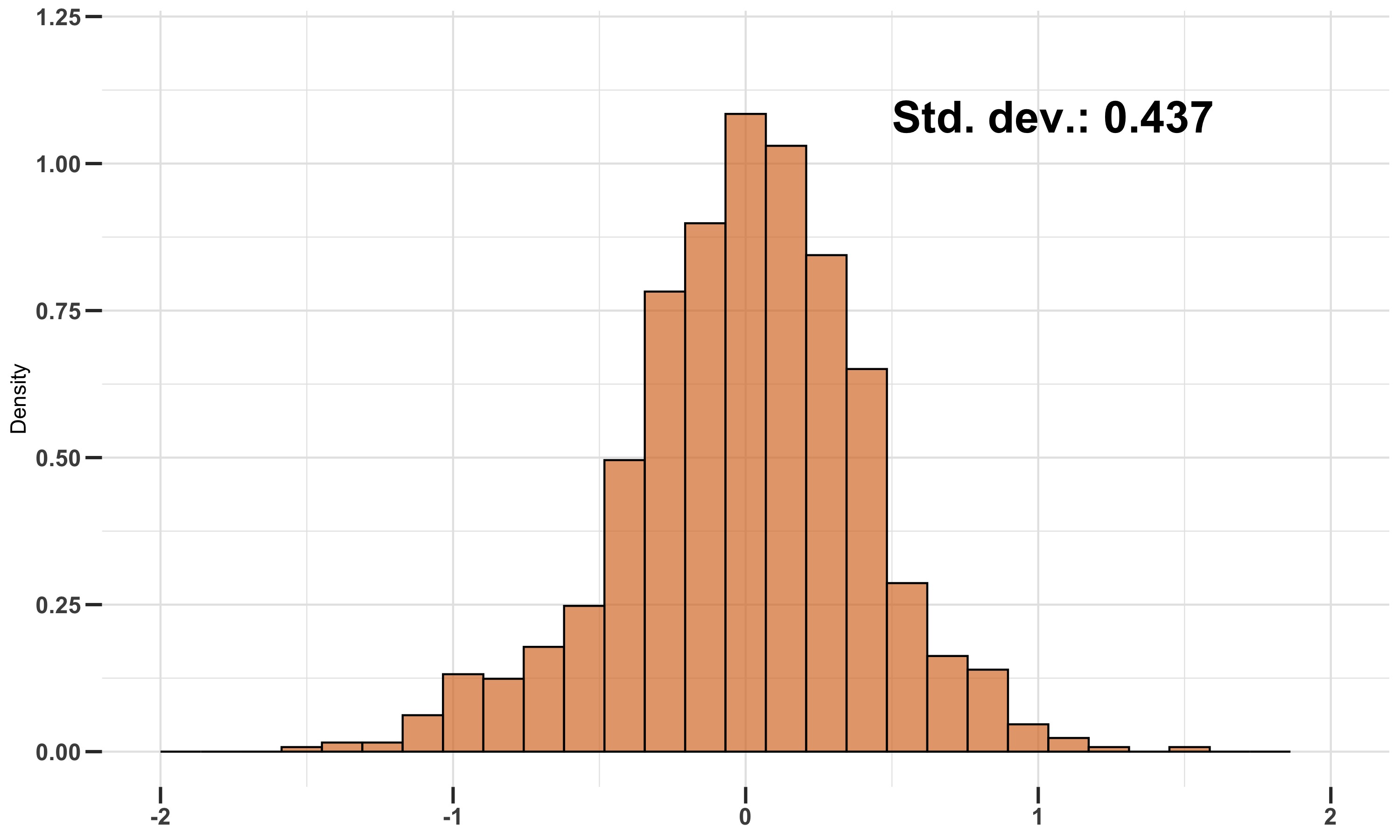}
    \end{subfigure}

    \begin{subfigure}[b]{0.45\textwidth}
        \centering
    \subcaption{Errors: $\hat{Y}^{\mathrm{MLE}}_{i,T_i}-Y_{i,T_i}$}
	\includegraphics[width=0.875\linewidth,keepaspectratio]{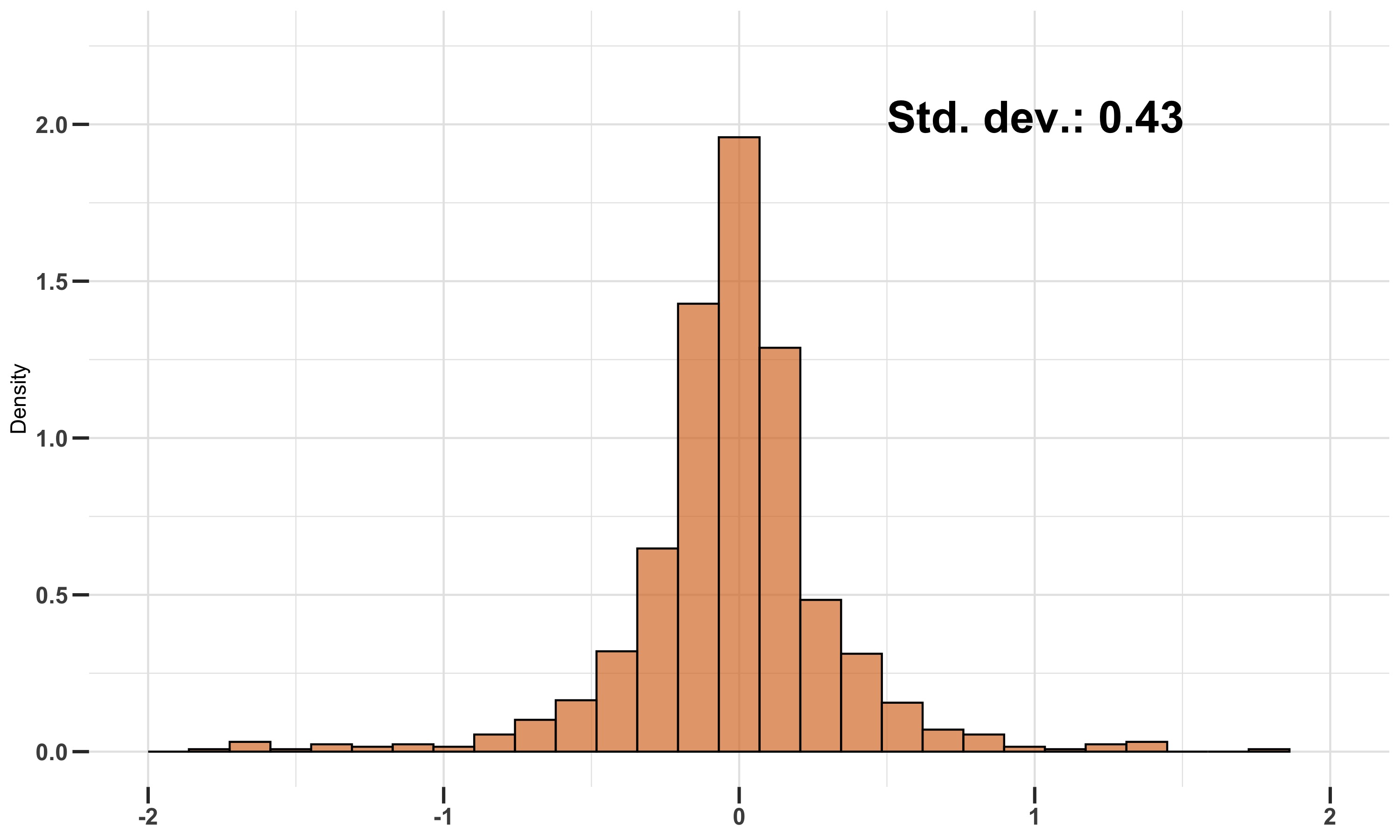} 
    \end{subfigure}
    \begin{subfigure}[b]{0.45\textwidth}
        \centering
    \subcaption{Errors: $\hat{Y}^{\mathrm{EB}}_{i,T_i} - Y_{i,T_i}$}
    \includegraphics[width=0.875\linewidth,keepaspectratio]{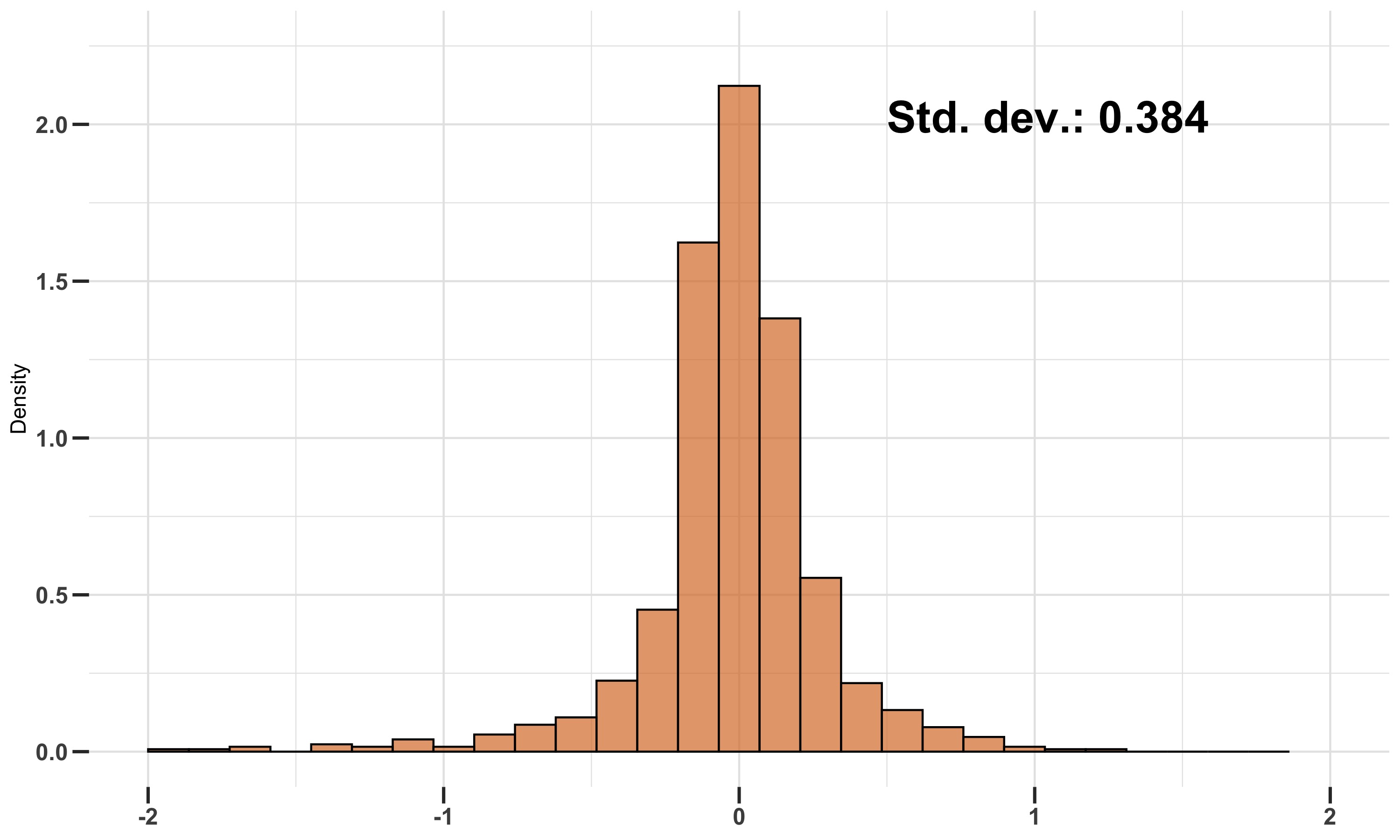}
    \end{subfigure}
    
    {\centering
\parbox{0.95\textwidth}{\footnotesize Notes: The reported standard deviations of the prediction errors are computed after excluding a single outlier observation whose prediction error falls below $-4.5$ under both methods; if this observation is included, the standard deviations for the MLE and EB errors are $0.456$ and $0.411$, respectively.}
}
\end{figure}

As shown in Table~\ref{tab:PSID:prior:moments}, a notable feature is the negative correlation between $a_i$ and $b_i$. 
To further examine this relationship, the left panel of Figure~\ref{fig:exp-var-profile} plots the EB estimate pairs $(\hat a_i^{\mathrm{EB}}, \hat b_i^{\mathrm{EB}})$ for $N = 938$ individuals, revealing a clear negative association between the individual intercepts $\hat a_i^{\mathrm{EB}}$ and slopes $\hat b_i^{\mathrm{EB}}$.\footnote{ Figure~\ref{fig:PSID:post:dist} in Appendix~\ref{appendix:pairwise:EB} presents all pairwise scatter plots of the EB estimates for \(\theta_i\).}
To explore the experience profile in the cross-section, we use the estimated moments in Table~\ref{tab:PSID:prior:moments} to compute the cross-sectional prior variance of \(a_i + b_i x\), which captures the deviation of an individual's income trajectory at \(X_{2,it} = x\) from the average. 
The mapping \(x \mapsto \mathbb{V}_{\hat G}(a_i + b_i x)\), shown in the right panel of Figure~\ref{fig:exp-var-profile}, indicates that the variance declines up to about \(x = 1\) (ten years of experience) and increases thereafter. 
This non-monotonic pattern arises from the negative covariance between \(a_i\) and \(b_i\). 
Interestingly, a similar U-shaped relationship between the log variance of residual earnings and experience is documented in \citet[Chapter~6, p.104]{Mincer:1974},  who attributes it to ``a weak correlation between post-school investment ratios and earning capacity.''

{
\begin{figure}[h]
    \caption{EB Estimates of $(a_i, b_i)$ and Cross-Sectional Variance of Earnings Profiles}
    \label{fig:exp-var-profile}
    \begin{center}

    \begin{subfigure}[b]{0.45\textwidth}
        \centering
        \includegraphics[width=0.99\linewidth,keepaspectratio]{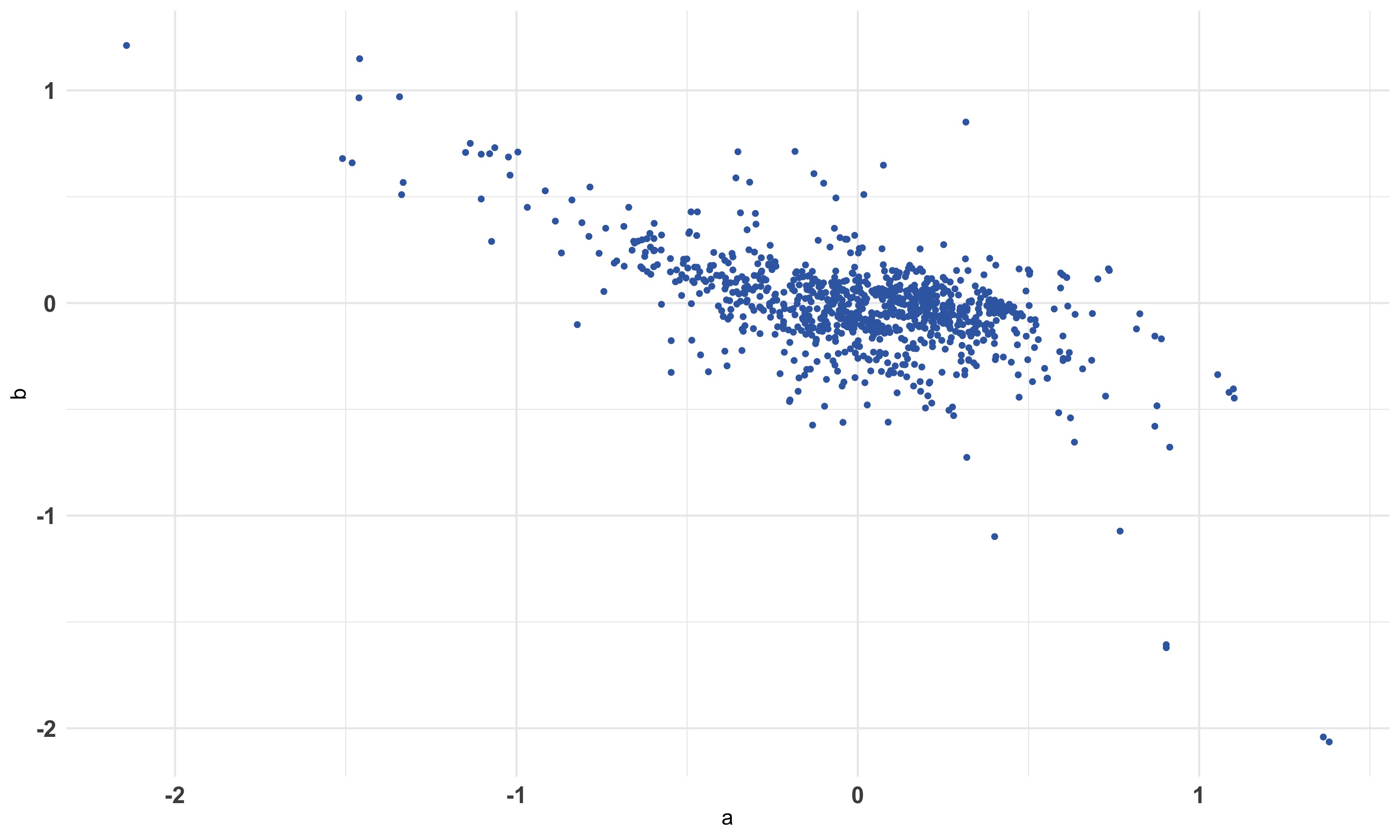}
        \subcaption{Pairs $(\hat a_i^{\mathrm{EB}}, \hat b_i^{\mathrm{EB}})$.}
    \end{subfigure}
    \begin{subfigure}[b]{0.45\textwidth}
        \centering
        \includegraphics[width=0.99\linewidth,keepaspectratio]{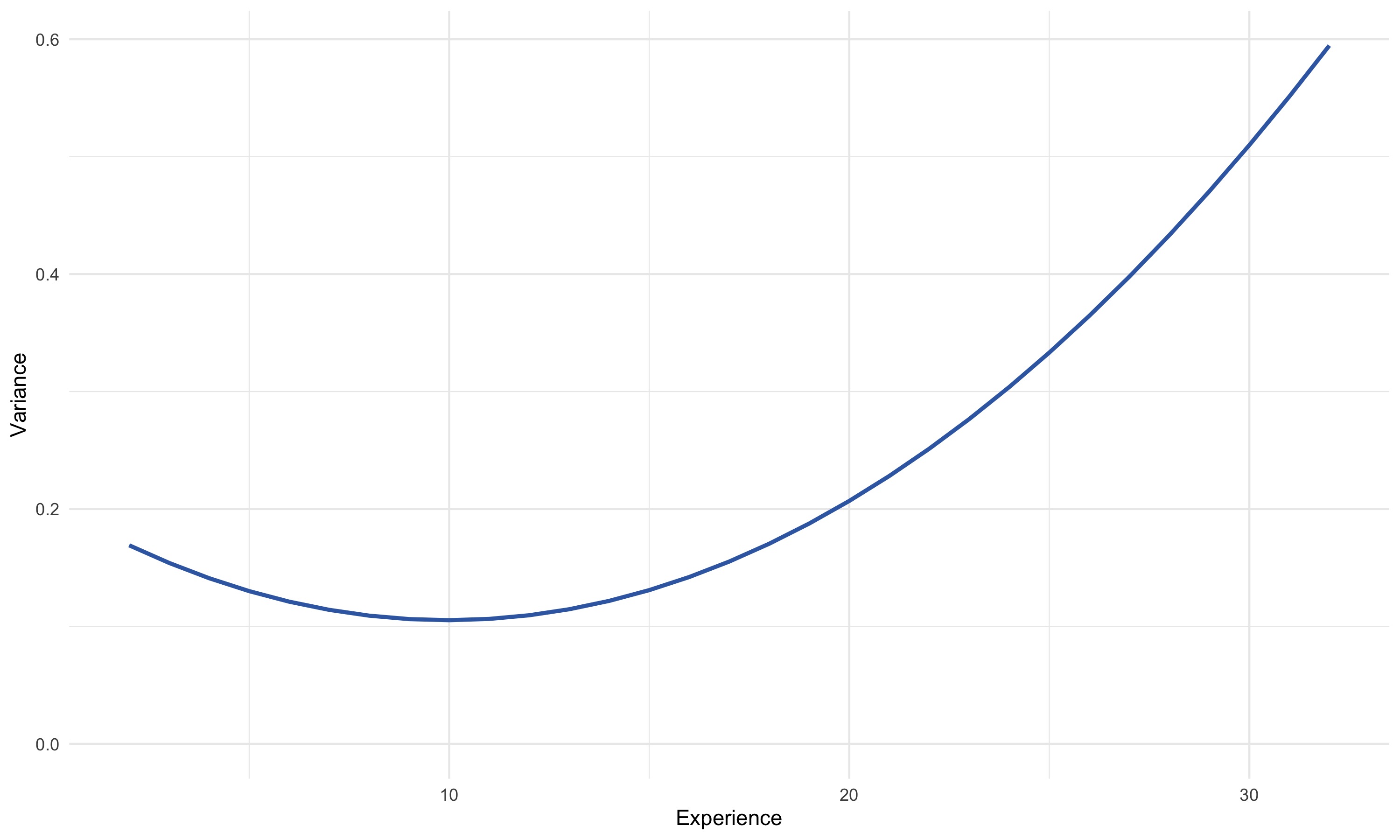}
        \subcaption{Variance profile by experience.}
    \end{subfigure}

    \end{center}

    \begin{center}
        \parbox{0.95\textwidth}{\footnotesize 
        Note: The left panel plots the EB estimates $(\hat a_i^{\mathrm{EB}}, \hat b_i^{\mathrm{EB}})_{i=1}^N$. 
        The right panel displays the cross-sectional variance of earnings profiles as a function of years of experience, 
        $\text{Exp} \mapsto \mathbb{V}_{\hat G}(a_i + b_i \text{Exp}/10)$.
        }
    \end{center}
\end{figure}
}

\section{Monte Carlo Experiments}\label{sec:Monte:Carlo}

We conduct Monte Carlo experiments to evaluate the performance of the EB estimator in dynamic panel models relative to (i) the oracle decision rule and (ii) the individual MLE. 

The following DGPs are considered. 

\begin{enumerate}

\item The HIVD model 
with $\theta_i = (a_i, \sigma_i^2, \rho_i)$.
The individual intercepts $(a_i)_{i=1}^N$ are drawn from the $\mathrm{Gamma}(1/2,\sqrt{2})$ distribution (shape $1/2$, scale $\sqrt{2}$) independently of $(\sigma_i^2, \rho_i)$.
The  joint distribution of $(\sigma_i^2, \rho_i)$  assigns probability mass \(1/3\) to \((\sigma_i^2, \rho_i) = (0.1, 0.8)\) and \((0.3, 0.2)\), and probability mass \(1/6\) to \((\sigma_i^2, \rho_i) = (0.1, 0.2)\) and \((0.3, 0.8)\). 
This design features a negative correlation of $-0.33$ between $\rho_i$ and $\sigma_i^2$.
The initial $u_{i0}$ is drawn from the stationary distribution $\mathcal{N}(0,\sigma_i^2/(1-\rho_i^2))$.
The sample size for estimation is $(N,T)=(1000,5)$. An additional observation $Y_{i T+1}$ is generated for each unit  to compute prediction errors.

\item The Restricted HIVD model follows the HIVD model above but  imposes $\rho_i \equiv 0.5$ and assumes $\sigma_i^2$ has a two-point distribution satisfying $G_*(\sigma_i^2 = 0.1) = G_*(\sigma_i^2 = 0.3) = 0.5$.
The initial distribution of $u_{i0}$ is $\mathcal{N}(0, \sigma_i^2/(1-0.5^2))$. 
The distribution of $a_i$ and
the sample size is the same as in HIVD.
This DGP mimics a scenario in which a researcher does not know that $\rho_i$ is a constant and continues to estimate the HIVD model.

\item The HIVDX model generates   $\theta_i = (a_i, b_i, \sigma_i^2, \rho_i)$  from the prior distribution estimated from the PSID sample. The covariate (interpreted as experience divided by 10) is   assumed  to evolve  according to $X_{2,it} = X_{2,i,t-1} + 0.1$, with  $(N,T) = (1000, 13)$, approximately as in  the PSID sample.
The initial value $X_{2,i1}$ is drawn from a discrete distribution with probabilities $\mathbb{P}(X_{2,i1}=0.3) = 0.25$, $\mathbb{P}(X_{2,i1}=0.4) = 0.05$, $\mathbb{P}(X_{2,i1}=0.5) = \mathbb{P}(X_{2,i1}=0.6) = 0.1$, and $\mathbb{P}(X_{2,i1}=0.7) = 0.5$.
\end{enumerate}

For the first two models,   the WFR parameters are $\eta = 0.1$,  $\overline{n} = 2000$, and  $\tol = 10^{-4}$. 
For the more complex HIVDX,  $\eta = 0.005$,  $\overline{n} = 10000$ without applying the early stopping criterion to ensure stable and reliable convergence of the algorithm.

Let $\hat{\varepsilon}_i = \hat{\Tau}_i - \Tau_i$ denote the estimation error for unit $i$, where $\tau_i$ is one of the model parameters, or the one-step-ahead prediction.
For each method and replication, we compute 
\begin{equation*}
\mathrm{Bias} = \frac{1}{N}\sum_{i=1}^N \hat{\varepsilon}_i,\quad
\mathrm{SD} = \left( \frac{1}{N}\sum_{i=1}^N \left(\hat{\varepsilon}_i - \frac{1}{N}\sum_{i=1}^N \hat{\varepsilon}_i\right)^2 \right)^{1/2},\quad
\mathrm{RMSE} = \left( \frac{1}{N}\sum_{i=1}^N \hat{\varepsilon}_i^{2} \right)^{1/2}.
\end{equation*}
Additionally, a measure of predictability, denoted by R$^2$, is computed as one minus the ratio of the mean squared prediction error to the sample variance of $Y_{i T+1}$. 

Table~\ref{tab:MC-HIVD-HIVDR} reports results  for  the HIVD model and its  restricted variant.
In both designs, the oracle estimator achieves the lowest estimation error, but the EB estimator improves upon the individual MLE with uniformly smaller RMSE across all parameters.
The gains are particularly pronounced for $\rho_i$. 
The gains from EB are even larger in the restricted HIVD design, where the prior for $\rho_i$ is more informative. Except for $\rho_i$, the RMSEs of the EB estimator are within about 12\% of those of the oracle, and within 2\% for prediction.

Table~\ref{tab:MC-HIVDX} presents the results for the HIVDX model.
Similarly to Table~\ref{tab:MC-HIVD-HIVDR}, the EB estimator outperforms the individual MLE across all parameters.
For both $a_i$ and $b_i$, the EB method achieves MSEs that are approximately 55\% of those of the MLE.
The EB estimator substantially reduces the MSE for $\rho_i$ as well, primarily attributable to bias reduction.
It is interesting to note that the Monte Carlo results for the bias of the MLE of $\rho_i$ echo the pattern observed in the empirical application in Section~\ref{sec:empirical}. In terms of out-of-sample prediction accuracy, R$^2$ increases from 0.591 for the MLE to 0.662 for EB.

Finally, Table~\ref{tab:prior-moment} shows the prior moment estimation results  in the HIVDX design. The first moments of $\theta_i$ are generally well estimated, with the exception of a modest downward bias in the estimate of $\mathbb{E}_{G_*}[\rho_i]$. The variances of $\theta_i$ tend to be somewhat overestimated, except for $\sigma_i^2$.
The negative covariance between $a_i$ and $b_i$ is estimated with a downward bias but the sign is reliably recovered.

\begin{table}[htbp!]
\caption{Errors from 500 Monte Carlo Replications}
\label{tab:MC-HIVD-HIVDR}
\centering

\begin{tabular}{
L{2.8cm}
C{1.6cm}
L{1.6cm}
S[table-format=1.3, table-column-width=2cm]
S[table-format=1.3, table-column-width=2cm]
S[table-format=1.3, table-column-width=2cm]}
\toprule
 & & & \multicolumn{3}{c}{Estimator} \\
\cmidrule(lr){4-6}
Errors & DGP & Metric & \multicolumn{1}{c}{\text{Oracle}} & \multicolumn{1}{c}{\text{MLE}} & \multicolumn{1}{c}{\text{EB}} \\
\midrule

\multirow{6}{*}{$\hat a_i - a_i$}
 & \multirow{3}{*}{HIVD} & Bias & -0.004 & 0.000 & -0.002 \\
 &  & SD   & 0.378 & 0.449 & 0.414 \\
 &  & RMSE & 0.378 & 0.449 & 0.414 \\
\cmidrule(lr){2-6}
 & \multirow{3}{*}{HIVDR} & Bias & 0.001 & 0.000 & -0.003 \\
 &  & SD   & 0.303 & 0.346 & 0.339 \\
 &  & RMSE & 0.303 & 0.346 & 0.339 \\
\midrule

\multirow{6}{*}{$\hat\sigma_i^{2} - \sigma_i^{2}$}
 & \multirow{3}{*}{HIVD} & Bias & 0.000 & -0.077 & -0.006 \\
 &  & SD   & 0.078 & 0.117 & 0.085 \\
 &  & RMSE & 0.078 & 0.141 & 0.086 \\
\cmidrule(lr){2-6}
 & \multirow{3}{*}{HIVDR} & Bias & 0.000 & -0.077 & -0.007 \\
 &  & SD   & 0.079 & 0.118 & 0.085 \\
 &  & RMSE & 0.079 & 0.141 & 0.086 \\
\midrule

\multirow{6}{*}{$\hat\rho_i - \rho_i$}
 & \multirow{3}{*}{HIVD} & Bias & -0.003 & -0.552 & -0.062 \\
 &  & SD   & 0.273 & 0.482 & 0.323 \\
 &  & RMSE & 0.273 & 0.733 & 0.330 \\
\cmidrule(lr){2-6}
 & \multirow{3}{*}{HIVDR} & Bias & 0.000 & -0.546 & -0.048 \\
 &  & SD   & 0.000 & 0.447 & 0.174 \\
 &  & RMSE & 0.000 & 0.706 & 0.182 \\
\midrule

\multirow{8}{*}{\footnotesize $\hat Y_{i T+1} - Y_{i T+1}$}
 & \multirow{4}{*}{HIVD} & Bias & -0.001 & 0.000 & 0.000 \\
 &  & SD   & 0.485 & 0.560 & 0.490 \\
 &  & RMSE & 0.485 & 0.560 & 0.490 \\
 &  & $\mathrm{R}^2$ & 0.825 & 0.766 & 0.821 \\
\cmidrule(lr){2-6}
 & \multirow{4}{*}{HIVDR} & Bias & -0.001 & 0.000 & -0.001 \\
 &  & SD   & 0.471 & 0.562 & 0.477 \\
 &  & RMSE & 0.471 & 0.562 & 0.477 \\
 &  & $\mathrm{R}^2$ & 0.824 & 0.749 & 0.819 \\
\bottomrule
\end{tabular}
\begin{center}
\parbox{0.9\textwidth}{\footnotesize{Notes: Each row reports averages computed from $500$ replications. 
Oracle refers to the true posterior mean, and MLE corresponds to the individual MLE.}}
\end{center}
\end{table}

\begin{table}[!htbp]
\caption{Errors from 500 Monte Carlo Replications (HIVDX)}
\label{tab:MC-HIVDX}
\centering
\begin{tabular}{L{2.5cm}L{2cm}
S[table-format=1.3, table-column-width=2.5cm]
S[table-format=1.3, table-column-width=2.5cm]
S[table-format=1.3, table-column-width=2.5cm]}
\toprule
 &  & \multicolumn{3}{c}{Estimator} \\ 
\cmidrule(lr){3-5}
Errors & Metric & \text{Oracle} & \text{MLE} & \text{EB} \\
\midrule
\multirow{3}{*}{$\hat a_i - a_i$} 
 & Bias & 0.001 & -0.001 & -0.001 \\
 & SD & 0.244 & 0.426 & 0.318 \\
 & RMSE & 0.244 & 0.426 & 0.318 \\
\midrule
\multirow{3}{*}{$\hat b_i - b_i$} 
 & Bias & -0.001 & 0.001 & 0.001 \\
 & SD & 0.188 & 0.338 & 0.251 \\
 & RMSE & 0.188 & 0.338 & 0.251 \\
\midrule
\multirow{3}{*}{$\hat\sigma_i^2 - \sigma_i^2$}
 & Bias & 0.000 & -0.020 & -0.006 \\
 & SD & 0.047 & 0.066 & 0.064 \\
 & RMSE & 0.047 & 0.069 & 0.064 \\
\midrule
\multirow{3}{*}{$\hat\rho_i - \rho_i$}
 & Bias & -0.002 & -0.326 & -0.082 \\
 & SD & 0.212 & 0.308 & 0.275 \\
 & RMSE & 0.212 & 0.448 & 0.288 \\
\midrule
\multirow{4}{*}{\footnotesize $\hat Y_{i T+1} - Y_{i T+1}$}
 & Bias & 0.000 & 0.001 & 0.001 \\
 & SD & 0.299 & 0.341 & 0.314 \\
 & RMSE & 0.299 & 0.341 & 0.314 \\ 
 & $\mathrm{R}^2$ & 0.692 & 0.591 & 0.662 \\ 
\bottomrule
\end{tabular}
\begin{center}
\parbox{0.9\textwidth}{\footnotesize{Notes: Each row reports averages computed from $500$ replications. 
Oracle refers to the true posterior mean, and MLE corresponds to the individual MLE.}}
\end{center}
\end{table}

\begin{table}[!htbp]
\caption{Monte Carlo Results for Moment Estimates of $G_*$ (HIVDX)}
\label{tab:prior-moment}
\centering
\renewcommand{\arraystretch}{1.2}

\begin{tabular}{L{1.6cm} *{6}{S[table-format=-1.3]}} 
\toprule\midrule

\multicolumn{7}{l}{\text{Panel A. Means}} \\
& & {$\mathbb{E}_{\hat{G}}[a_i]$} & {$\mathbb{E}_{\hat{G}}[b_i]$} &  {$\mathbb{E}_{\hat{G}}[\sigma_i^2]$} & {$\mathbb{E}_{\hat{G}}[\rho_i]$}  \\
\midrule
Truth & &  0.013 & -0.017 &  0.079 &  0.422  \\
Bias  & &-0.001 &  0.001 & -0.006 & -0.082  \\
SD    & & 0.022 &  0.017 &  0.005 &  0.016  \\
RMSE  & & 0.023 &  0.018 &  0.011 &  0.080  \\
\midrule\midrule

\multicolumn{7}{l}{\text{Panel B. Variances}} \\
& &{$\mathbb{V}_{\hat{G}}(a_i)$} & {$\mathbb{V}_{\hat{G}}(b_i)$} & {$\mathbb{V}_{\hat{G}}(\sigma_i^2)$} & {$\mathbb{V}_{\hat{G}}(\rho_i)$}  \\
\midrule
Truth & & 0.205 & 0.101 &  0.023 & 0.105  \\
Bias  & & 0.054 & 0.030 & -0.007 & 0.017  \\
SD    & & 0.037 & 0.025 &  0.006 & 0.009  \\
RMSE  & & 0.065 & 0.039 &  0.009 & 0.020  \\
\midrule\midrule

\multicolumn{7}{l}{\text{Panel C. Pairwise Covariances}} \\
& {$(a_i,b_i)$} & {$(a_i,\sigma_i^2)$} & {$(a_i,\rho_i)$} & {$(b_i,\sigma_i^2)$} & {$(b_i,\rho_i)$} & {$(\rho_i,\sigma_i^2)$} \\
\midrule
Truth & -0.100 & -0.011 & -0.004 & -0.010 &  0.007 & -0.001 \\
Bias  & -0.038 & -0.003 &  0.005 &  0.007 & -0.004 &  0.003 \\
SD    &  0.028 &  0.011 &  0.009 &  0.010 &  0.006 &  0.002 \\
RMSE  &  0.047 &  0.010 &  0.011 &  0.012 &  0.008 &  0.003 \\
\bottomrule
\end{tabular}

\begin{center}
\parbox{0.95\textwidth}{\footnotesize Notes: The row labeled `Truth' reports the true values of the corresponding moments.
Results are based on 500 replications.}
\end{center}
\end{table}

\section{Conclusions}\label{sec:conclusions}

We have proposed a nonparametric empirical Bayes framework for analyzing short-panel data with rich forms of unobserved heterogeneity. The analysis generalizes classical G-modeling to accommodate heterogeneous slopes and non-spherical error structures. Identification, consistency, and optimality results are established under general conditions, with primitive sufficient conditions derived for specific cases of interest. 
However, we have not developed any accompanying inference methods. 
Recent contributions such as \citet{AKPM:2022} and \citet{ignatiadis2022confidence} can be useful for future research.
Our identification result also requires the distribution of the variance parameters $\delta_i$ to have finite, though unknown, support; relaxing this restriction is another direction for future research.

\clearpage
\newpage

\bibliographystyle{econometrica}
\bibliography{EB.bib}

\clearpage
\newpage

\appendix

\renewcommand{\thepage}{S-\arabic{page}}
        \setcounter{page}{1}
        
        \renewcommand{\thesection}{S-\arabic{section}}
        \setcounter{section}{0}

        \begin{center}
            \Large{Supplement to ``Empirical Bayes Estimation in Heterogeneous Coefficient Panel Models''}
        \end{center}

\section{Proofs}
\label{sec:proof}

\subsection{Proofs of Main Results}

\begin{proof}[Proof of Theorem~\ref{thm:identification}]
We divide the proof into two parts.
Part (i) establishes identification of the marginal distribution of $\delta_i$; part (ii) builds on this to show that the full mixing distribution $G_*$ is identified.

\medskip

\noindent \textit{Part~(i)}.
Take any $G_\oo \in \mathcal{G}$ with $f_{G_\oo}(Y_i, \X_i) = f_{G_*}(Y_i, \X_i)$ almost surely.
Since $\ell(y \,|\, x, \theta)$ is continuous in $(y,x)$ and bounded uniformly in $\theta$ under Assumption~\ref{asm:DGP}(ii), $f_G(y,x)$ is continuous in $(y,x)$ for every $G \in \mathcal{G}$ by dominated convergence. Hence, the equality holds for all $y \in \mathbb{R}^T$ and all $x \in \supp(\X_i)$, which justifies conditioning on $\X_i = \xmat$ below even if this event has probability zero.
Let $H_*$ and $H_\oo$ denote the marginal distributions of $\delta_i$ induced by $G_*$ and $G_\oo$, respectively.
In what follows, we show that $H_* = H_\oo$.

Let $\xmat$ and $\Q $ be given as in Assumption~\ref{asm:id-conditions}.
Conditioning on $\X_i= \xmat$ and premultiplying both sides of \eqref{eq:model} by $\Q'$, we obtain $\Q' Y_i =  \Q' P_i e_i$.
Taking conditional characteristic functions of $\Q' Y_i$ according to $f_{G_\oo}(\cdot \,|\, \xmat)$ and $f_{G_*}(\cdot \,|\, \xmat)$, respectively, we have
\begin{align*}
\mathbb{E}_{G_\oo}[e^{\ii t' \Q' Y_i} | \X_i = \xmat] &= \int  e^{-t' \Q'P(\delta)P(\delta)'\Q t/2} dH_\oo(\delta),  \\
\mathbb{E}_{G_*}[e^{\ii t' \Q' Y_i} | \X_i = \xmat] &= \int  e^{-t' \Q'P(\delta)P(\delta)'\Q t/2} dH_*(\delta).
\end{align*}
Equating $\mathbb{E}_{G_\oo}[e^{\ii t' \Q' Y_i} | \X_i = \xmat]$ with $\mathbb{E}_{G_*}[e^{\ii t' \Q' Y_i} | \X_i = \xmat]$, we have
\begin{align}
\label{eq:Ho-Hstar}
& \int_{\mathcal{K}_{\delta}} \frac{|\Q'P(\delta)P(\delta)'\Q|^{1/2}}{(2\pi)^{(T-\dbeta)/2}} e^{-t' \Q'P(\delta)P(\delta)'\Q t/2}\cdot |\Q'P(\delta)P(\delta)'\Q|^{-1/2} dH_\oo(\delta) \nonumber \\
= &\ \int_{\mathcal{K}_{\delta}} \frac{|\Q'P(\delta)P(\delta)'\Q|^{1/2}}{(2\pi)^{(T-\dbeta)/2}} e^{-t' \Q'P(\delta)P(\delta)'\Q t/2}\cdot |\Q'P(\delta)P(\delta)'\Q|^{-1/2} dH_*(\delta),
\end{align}
which can be rewritten as
\begin{equation*}
    \mathcal{T}[0, (\Q'P(\delta)P(\delta)'\Q)^{-1},\tilde{H}_\oo](t) = \mathcal{T}[0, (\Q'P(\delta)P(\delta)'\Q)^{-1},\tilde{H}_*](t), \quad \forall t\in \mathbb{R}^{T-\dbeta},
\end{equation*}
with $\mathcal{T}$ defined before Lemma~\ref{lem:finite-mixture}. The probability measures $\tilde{H}_\oo$ and $\tilde{H}_*$ are defined via
\begin{equation*}
d\tilde{H}_\oo(\delta) := \frac{|\Q'P(\delta)P(\delta)'\Q|^{-1/2}}{C} dH_\oo(\delta),\quad d\tilde{H}_*(\delta) := \frac{|\Q'P(\delta)P(\delta)'\Q|^{-1/2}}{C} dH_*(\delta),
\end{equation*}
where the norming constant $C>0$ is
\begin{equation*}
C = \int_{\mathcal{K}_\delta} |\Q'P(\delta)P(\delta)'\Q|^{-1/2} dH_\oo(\delta) = \int_{\mathcal{K}_\delta} |\Q'P(\delta)P(\delta)'\Q|^{-1/2} dH_*(\delta).
\end{equation*}
This equality follows from integrating both sides of \eqref{eq:Ho-Hstar} with respect to $t \in \mathbb{R}^{T-\dbeta}$, ensuring that both $\tilde{H}_*$ and $\tilde{H}_\oo$ are probability measures.
We are prepared to apply Lemma~\ref{lem:finite-mixture}.
Since $\delta \mapsto (\Q'P(\delta)P(\delta)'\Q)^{-1}$ is continuous and, by Assumption~\ref{asm:id-conditions}(ii), injective, Lemma~\ref{lem:finite-mixture} applies with $(\lambda_1(\delta),\lambda_2(\delta)) := (0, (\Q'P(\delta)P(\delta)'\Q)^{-1})$.
By Lemma~\ref{lem:finite-mixture}(b), which applies because $\tilde{H}_*$ has finite support by Assumption~\ref{asm:id-conditions}(iii), it follows that $\tilde{H}_* = \tilde{H}_\oo$, and hence $H_* = H_\oo$.

\bigskip

\noindent \textit{Part~(ii)}.
Next, we show that the conditional distribution of $\xmat \hc_i$ given $\delta_i = \delta$ is identified for each $\delta$ in the support of $H_*$.

Let $\xmat$ and $\Q $ be again given as in Assumption~\ref{asm:id-conditions}.
Computing the conditional characteristic functions of $Y_i$ given $\X_i = \xmat$ according to the densities $f_{G_\oo}(\cdot\,|\,\xmat)$ and $f_{G_*}(\cdot\,|\,\xmat)$, respectively, we obtain:
\begin{align*}
\varphi_\oo(t) &:= \mathbb{E}_{G_\oo}[e^{\ii t' Y_i} | \X_i = \xmat]  =   \int \mathbb{E}_{G_\oo}[e^{\ii t'\xmat \hc_i} |\delta_i = \delta] e^{-t'P(\delta)P(\delta)'t/2} d H_\oo(\delta),   \\
\varphi_*(t) &:= \mathbb{E}_{G_*}[e^{\ii t' Y_i} | \X_i = \xmat] = \int \mathbb{E}_{G_*}[e^{\ii t'\xmat \hc_i} |\delta_i = \delta] e^{-t'P(\delta)P(\delta)'t/2} d H_*(\delta).
\end{align*}
These functions coincide with each other since $f_{G_\oo} = f_{G_*}$.
By Part~(i), we also have $H_\oo = H_*$.
Fix $t \in \mathbb{R}^T$, and let
\begin{equation*}
\psi(\delta) := \mathbb{E}_{G_*}[e^{\ii t'\xmat \hc_i} |\delta_i = \delta]-\mathbb{E}_{G_\oo}[e^{\ii t'\xmat \hc_i} |\delta_i = \delta].
\end{equation*}
Equating $\varphi_\oo(t+ \Q s)$ with $\varphi_*(t + \Q s)$, we then obtain
\begin{align}
\label{eq:ch-ftn}
\int \psi(\delta) \exp\left( -\frac{1}{2} \|P(\delta)'(\Q s + t)\|^2 \right) d H_*(\delta) = 0
\end{align}
for all $s \in \mathbb{R}^{T-\dbeta}$.

For each $\delta \in \mathcal{K}_{\delta}$, we write 
\begin{equation*}
t = \Q  \alpha_\delta + e_\delta,
\end{equation*}
where $\alpha_\delta = (\Q' P(\delta)P(\delta)' \Q)^{-1} \Q' P(\delta)P(\delta)' t$ is chosen so that $\Q' P(\delta)P(\delta)' e_\delta = 0$.
This allows us to rewrite \eqref{eq:ch-ftn} as
\begin{align}
\label{eq:ch-ftn-2}
    0 &= \int \psi(\delta)  \exp\left( -\frac{1}{2} \|P(\delta)'e_\delta\|^2  -\frac{1}{2} (s + \alpha_\delta)'\Q 'P(\delta)P(\delta)'\Q (s + \alpha_\delta) \right)  d H_*(\delta) \nonumber\\
    &= \mathlarger{\int}
    \begin{aligned}
        & (2\pi)^{(T-\dbeta)/2}|\Q 'P(\delta)P(\delta)'\Q|^{-1/2} \psi(\delta)  \exp\left( -\frac{1}{2} \|P(\delta)'e_\delta\|^2\right) \\
    &\quad \times \frac{|\Q 'P(\delta)P(\delta)'\Q|^{1/2}}{(2\pi)^{(T-\dbeta)/2}}\exp\left(  -\frac{1}{2} (s + \alpha_\delta)'\Q 'P(\delta)P(\delta)'\Q (s + \alpha_\delta) \right)   d H_*(\delta)
    \end{aligned}
\end{align}
for all $s \in \mathbb{R}^{T-\dbeta}$.
Define a complex Borel measure
\begin{equation*}
 d\mu(\delta) := (2\pi)^{(T-\dbeta)/2}|\Q 'P(\delta)P(\delta)'\Q|^{-1/2} \psi(\delta)  \exp\left( -\frac{1}{2} \|P(\delta)'e_\delta\|^2\right) d H_*(\delta).
\end{equation*}
Our next goal is to establish that the real and imaginary parts of $\mu$ are both zero.
By construction of $\mu$, \eqref{eq:ch-ftn-2} implies that
\begin{align}
\label{eq:ch-ftn-3}
       & \mathcal{T}[-\alpha_\delta, (\Q'P(\delta)P(\delta)'\Q)^{-1}, \mu](s) \nonumber \\
       =&\
       \int \frac{|\Q 'P(\delta)P(\delta)'\Q|^{1/2}}{(2\pi)^{(T-\dbeta)/2}}\exp\left(  -\frac{1}{2} (s + \alpha_\delta)'\Q 'P(\delta)P(\delta)'\Q (s + \alpha_\delta) \right)   d \mu(\delta)  = 0
\end{align}
for all $s \in \mathbb{R}^{T-\dbeta}$.
Using the Jordan decomposition theorem, we decompose $\mu$ into
\begin{equation*}
\mu = \mu_+^R  - \mu_-^R + \ii (\mu_+^I - \mu_-^I),
\end{equation*}
where $\mu_\pm^R$ are mutually singular non-negative Borel measures, representing the positive and negative real parts of $\mu$, respectively.
The same holds for the positive and negative imaginary parts $\mu_\pm^I$.
By linearity of $\mathcal{T}$ with respect to $\mu$, it follows from \eqref{eq:ch-ftn-3} that
\begin{align}
\label{eq:real-part}
 \mathcal{T}[-\alpha_\delta, (\Q'P(\delta)P(\delta)'\Q)^{-1}, \mu^R_+](s) &= \mathcal{T}[-\alpha_\delta, (\Q'P(\delta)P(\delta)'\Q)^{-1}, \mu^R_-](s), \\
 \label{eq:imag-part}
 \mathcal{T}[-\alpha_\delta, (\Q'P(\delta)P(\delta)'\Q)^{-1}, \mu_I^+](s) &= \mathcal{T}[-\alpha_\delta, (\Q'P(\delta)P(\delta)'\Q)^{-1}, \mu_I^-](s).
\end{align}
Integrating both sides of \eqref{eq:real-part} with respect to $s \in \mathbb{R}^{T-\dbeta}$, by Fubini's theorem, it follows
\begin{align*}
\mu^R_+(\mathcal{K}_\delta) &= \int_{\mathbb{R}^{T-\dbeta}} \mathcal{T}[-\alpha_\delta, (\Q'P(\delta)P(\delta)'\Q)^{-1}, \mu^R_+](s) ds \\
& = \int_{\mathbb{R}^{T-\dbeta}} \mathcal{T}[-\alpha_\delta, (\Q'P(\delta)P(\delta)'\Q)^{-1}, \mu^R_-](s) ds = \mu^R_-(\mathcal{K}_\delta).
\end{align*}
We examine two possible cases.
If $\mu_-^R(\mathcal{K}_\delta) = \mu_+^R(\mathcal{K}_\delta) = 0$, the real part of $\mu$ vanishes trivially.
Otherwise, we may without loss of generality assume that both $\mu_-^R$ and $\mu_+^R$ are probability measures.
Since $\mu^R_+$ and $\mu^R_-$ are supported on the finite support of $H_*$, Lemma~\ref{lem:finite-mixture}(a), applied with the continuous function $\lambda_1(\delta) := -\alpha_\delta$ and the same $\lambda_2$ (the pair is injective because $\lambda_2$ is), then implies that $\mu^R_+ = \mu^R_-$.
However, since $\mu^R_+$ and $\mu^R_-$ are mutually singular, this is only possible when both $\mu^R_+$ and $\mu^R_-$ have total mass of zero, which leads to a contradiction.
Thus, we conclude $\mu_-^R = \mu_+^R \equiv 0$, and hence the real part of $\mu$ is identically zero.
The same applies to $\mu_I^+$ and $\mu_I^-$. It follows $\mu \equiv 0$, thus implying that $\psi(\delta) = 0$ $H_*$-a.e. on $\mathcal{K}_\delta$.

Thus far, we have established that 
\begin{equation*}
\mathbb{E}_{G_*}[e^{\ii t'\xmat \hc_i} |\delta_i = \delta] - \mathbb{E}_{G_\oo}[e^{\ii t'\xmat \hc_i} |\delta_i = \delta]   = 0
\end{equation*}
for $H_*$-a.e. $\delta$ for each fixed $t\in\mathbb{R}^T$.
Define the function 
\begin{equation*}
\omega(t,\delta) := \mathbb{E}_{G_*}[e^{\ii t'\xmat \hc_i}|\delta_i = \delta] - \mathbb{E}_{G_\oo}[e^{\ii t'\xmat \hc_i}|\delta_i = \delta]
\end{equation*}
which admits a jointly measurable version, and hence can be regarded as Borel-measurable in $(t,\delta)$.
By Fubini's theorem, we have
\begin{align*}
     0 = \int_{\mathbb{R}^T} H_*\left( \{ \delta \in \K_\delta : \omega(t,\delta) \ne 0 \}\right)dt &= (\operatorname{Leb}\times H_*) \left( \{ (t,\delta) \in \mathbb{R}^T\times \K_\delta : \omega(t,\delta) \ne 0 \}\right) \\
     &=\int_{\mathcal{K}_\delta} \operatorname{Leb}\left( \{ t \in \mathbb{R}^{T} : \omega(t,\delta) \ne 0 \}\right)dH_*(\delta),
\end{align*}
where $\operatorname{Leb}\times H_*$ denotes the product measure of the Lebesgue measure and $H_*$.
It follows that $\operatorname{Leb}(\{t:\omega(t,\delta) \ne 0\}) = 0$ for $H_*$-a.e. $\delta \in \K_\delta$.
Observe that $\omega(\cdot, \delta)$ is a uniformly continuous function of $t \in \mathbb{R}^T$ for $\delta \in \mathcal{K}_{\delta}$.
Since the set of all zeros of $\omega(\cdot, \delta)$ is dense in $\mathbb{R}^T$, $\omega(\cdot, \delta)$ must vanish everywhere.
It follows that for $H_*$-a.e. $\delta$,
\begin{equation*}
\mathbb{E}_{G_\oo}[e^{\ii t'\xmat \hc_i}|\delta_i = \delta] = \mathbb{E}_{G_*}[e^{\ii t'\xmat \hc_i}|\delta_i = \delta],\quad \forall t \in \mathbb{R}^T.
\end{equation*}
This establishes that the conditional distribution of $\xmat \hc_i$ given $\delta_i$ is identified.
Finally, since $\xmat$ has full column rank, this is equivalent to identifying the conditional distribution of $\hc_i$ given $\delta_i$, implying that $G_*$ is identified from the conditional distribution $\hc_i \,|\, \delta_i$ and the marginal distribution of $\delta_i$.
\end{proof}

\begin{proof}[Proof of Theorem~\ref{thm:consistency}]

We verify Assumptions~1--5 of KW, with their metric replaced by $d$.

KW~Assumption~1 is trivially met, since $Y \mapsto f_{G}(Y,\X)$ is a legitimate density with respect to the Lebesgue measure on $\mathbb{R}^T$.
KW~Assumption~2 requires 
\begin{equation*}
f_{\tilde{\mu}}(Y,\X) \to f_{\mu}(Y,\X)\quad \text{as}\quad {d}(\tilde{\mu}, \mu) \to 0
\end{equation*}
for any $\mu$ and $\mathbb{P}$-a.e. $(Y,\X)$, which is established in Lemma~\ref{lem:likelihood-continuity}.
KW~Assumptions~3 and 4 are verified in Lemma~\ref{lem:key-to-consistency}(ii) and (iii), respectively.
Finally, KW~Assumption~5 follows from the fact that $\mathbb{E}[|\log f_{G_*}(Y_i,\X_i)|] < \infty$ and the random variable $\log \bar{f}_{\mu, \varepsilon}(Y_i,\X_i)$ is uniformly bounded above for all $\mu\in\bar{\mathcal{G}}$ and $\varepsilon>0$, as established in Lemma~\ref{lem:key-to-consistency}(i) and (ii).
\end{proof}

\begin{proof}[Proof of Theorem~\ref{thm:EB-consistency}]
We first show that $ \mathbb{E} [N^{-1}\sum_{i=1}^N |\taueb_i - \taupm_i|^p] \to 0$ as $N \to \infty$.
Let $\epsilon > 0$ be fixed, and let $M>0$ be given as in Assumption~\ref{asm:EB}(ii).
Consider the following decomposition
\begin{equation*}
\Tau_i = \Tau_{M,i} + r_{M,i},
\end{equation*}
where
\begin{equation*}
\Tau_{M,i} \equiv \Tau_M (\theta_i) := \frac{M}{\max\{|\Tau(\theta_i)|, M\}}\Tau(\theta_i) \equiv \begin{cases}
\Tau(\theta_i) & \text{if }|\Tau(\theta_i)| \le M \\
\frac{M}{|\Tau(\theta_i)|}\Tau(\theta_i) & \text{if }|\Tau(\theta_i)|> M
\end{cases} 
\end{equation*}
denotes the truncation of $\Tau_i$ at level $M$, and $r_{M,i} = r_M(\theta_i) := \Tau(\theta_i) - \Tau_M(\theta_i)$ is the remainder term.
By construction, part~(i) of Assumption~\ref{asm:EB} implies that $\Tau_M(\theta)$ is a bounded continuous function of $\theta$, and part~(ii) implies
\begin{equation*}
\limsup_{N\to\infty} \mathbb{E} \int_{\Theta} \|r_M(\theta)\|^p d \hat{G}_N(\theta) \le \epsilon.
\end{equation*}
Write $\Tau_{M,i}^{*}$ and $\taueb_{M,i}$ for the oracle and the EB decision rules for $\Tau_{M,i}$, and similarly for $r_{M,i}^{*}$ and $\hat{r}^{\EB}_{M,i}$.
Using the fact that $\taupm_i=\Tau_{M,i}^{*}+r_{M,i}^{*}$ and $\taueb_i = \taueb_{M,i} + \hat{r}^{\EB}_{M,i}$, we have
\begin{align*}
\frac{1}{N} \sum_{i=1}^N |\taueb_i - \taupm_i|^p &\le C_p \left( \frac{1}{N} \sum_{i=1}^N |\taueb_{M,i} - \Tau_{M,i}^{*}|^p + \frac{1}{N} \sum_{i=1}^N |\hat{r}^{\EB}_{M,i} - r_{M,i}^{*}|^p\right) \\
&=: C_p(I_1 + I_2)
\end{align*}
for some $C_p<\infty$ that only depends on $p$.
In the following steps, we show that both $\mathbb{E}[I_1]$ and $\mathbb{E}[I_2]$ converge to zero.

Let $M_0 > 0$ be a sufficiently large number such that $\mathbb{P}(\|Y_i\| \le M_0, \, M_0^{-1} \I_{\dbeta} \le \X_i'\X_i \le M_0 \I_{\dbeta}) \ge 1 - \epsilon/(2M)^p$.
Let $\mathcal{R} = \{(Y,\X) : \|Y\| \le M_0, \, M_0^{-1} \I_{\dbeta} \le \X'\X \le M_0 \I_{\dbeta}\}$, and let $\mathcal{R}_i = 1\{\|Y_i\| \le M_0, \, M_0^{-1} \I_{\dbeta} \le \X_i'\X_i \le M_0 \I_{\dbeta}\}$ be the corresponding indicator for unit $i$.
We decompose $I_1$ into two terms:
\begin{align*}
I_{11} = \frac{1}{N} \sum_{i=1}^N  |\taueb_{M,i} - \Tau_{M,i}^{*}|^p (1-\mathcal{R}_i), \quad
I_{12} = \frac{1}{N} \sum_{i=1}^N  |\taueb_{M,i} - \Tau_{M,i}^{*}|^p \mathcal{R}_i.
\end{align*}
It is straightforward to see that
\begin{equation*}
\mathbb{E}[I_{11}] \le (2M)^p \mathbb{E}\left[ \frac{1}{N} \sum_{i=1}^N (1-\mathcal{R}_i)\right]  \le \epsilon
\end{equation*}
by construction of $\mathcal{R}_i$.
On the other hand,
\begin{align*}
\mathbb{E}[I_{12}] \le \mathbb{E}\sup_{i : (Y_i,\X_i) \in \mathcal{R}} |\taueb_{M,i} - \Tau_{M,i}^{*}|^p \le \mathbb{E} \sup_{(Y,\X) \in \mathcal{R}} |\taueb_{M}(Y,\X) - \Tau_{M}^{*}(Y,\X)|^p \to 0
\end{align*}
by Lemma~\ref{lem:EB-unif-conv} and the bounded convergence theorem (\citeay{RoydenFitzpatrick2010}).
These results imply that $\limsup_{N\to\infty} \mathbb{E}[I_1] \le \epsilon$.

To show $\limsup_{N\to\infty}\mathbb{E}[I_2] \le \epsilon$, first observe that
\begin{align*}
I_2 &\le C_p \left( \frac{1}{N} \sum_{i=1}^N |\hat{r}^{\EB}_{M,i}|^p + \frac{1}{N} \sum_{i=1}^N |r_{M,i}^{*}|^p \right)    =: C_p (I_{21} + I_{22}).
\end{align*}
For the first term, we have, by the Jensen's inequality,
\begin{equation*}
|\hat{r}^{\EB}_{M,i}|^p \le \mathbb{E}_{\hat{G}}[|r_M(\theta_i)|^p \,|\, Y_i, \X_i] = \frac{\int_{\Theta} |r_M(\theta)|^p \ell(Y_i\,|\, \X_i, \theta) d\hat{G}_N(\theta)}{f_{\hat{G}}(Y_i, \X_i)},  \quad i=1,\ldots, N.
\end{equation*}
This implies
\begin{align*}
\frac{1}{N} \sum_{i=1}^N |\hat{r}^{\EB}_{M,i}|^p &\le  \frac{1}{N} \sum_{i=1}^N \frac{\int_{\Theta} |r_M(\theta)|^p \ell(Y_i\,|\, \X_i, \theta) d\hat{G}_N(\theta)}{f_{\hat{G}}(Y_i, \X_i)}\\
&= \int_{\Theta} \frac{1}{N} \sum_{i=1}^N \frac{\ell(Y_i\,|\, \X_i, \theta) }{f_{\hat{G}}(Y_i, \X_i)} |r_M(\theta)|^p d\hat{G}_N(\theta).
\end{align*}
The first-order conditions for the \NPM\ state that $N^{-1}\sum_{i=1}^N {\ell(Y_i\,|\, \X_i, \theta) }/{f_{\hat{G}}(Y_i, \X_i)} = 1$ for $\hat{G}_N$-a.e. $\theta$, which imply
\begin{equation*}
\int_{\Theta} \frac{1}{N} \sum_{i=1}^N \frac{ \ell(Y_i\,|\, \X_i, \theta) }{f_{\hat{G}}(Y_i, \X_i)} |r_M(\theta)|^p d\hat{G}_N(\theta) = \int_{\Theta}|r_M(\theta)|^p d\hat{G}_N(\theta).
\end{equation*}
It follows that $\limsup_{N\to\infty} \mathbb{E}[I_{21}] \le \epsilon$ by Assumption~\ref{asm:EB}(ii).

Turning to the term $I_{22}$, we have, by Jensen's inequality,
\begin{align*}
I_{22} & \le \frac{1}{N} \sum_{i=1}^N \mathbb{E}_{G_*}[|r_M(\theta_i)|^p \,|\, Y_i, \X_i] ,
\end{align*}
which implies that $\limsup_{N\to\infty}\mathbb{E}[I_{22}] \le \mathbb{E}_{G_*}[|r_M(\theta_i)|^p]$.
Since $\theta \mapsto |r_M(\theta)|^p$ is a nonnegative, continuous function of $\theta$ and $\hat{G}_N$ converges weakly to $G_*$ almost surely, Portmanteau theorem implies that
\begin{equation*}
    \int_{\Theta} |r_M(\theta)|^p dG_*(\theta) \le \liminf_{N\to\infty} \int_{\Theta}|r_M(\theta)|^p d\hat{G}_N(\theta)\quad \text{a.s.}
\end{equation*}
By Fatou's lemma, $\mathbb{E}[\liminf_{N\to\infty} \int_{\Theta}|r_M(\theta)|^p d\hat{G}_N(\theta)] \le \liminf_{N\to\infty} \mathbb{E}[\int_{\Theta}|r_M(\theta)|^p d\hat{G}_N(\theta)]$.
By Assumption~\ref{asm:EB}(ii), this implies that $\limsup_{N\to\infty} \mathbb{E}[I_{22}]  \le \epsilon$.

Putting these pieces together, we conclude that
\begin{equation*}
\limsup_{N\to\infty} \mathbb{E}\left[ \frac{1}{N} \sum_{i=1}^N |\taueb_i - \taupm_i|^p\right]\le C_p(1+2C_p)\epsilon.
\end{equation*}
Since $\epsilon>0$ was arbitrarily chosen, we have $\mathbb{E}[ N^{-1}\sum_{i=1}^N |\taueb_i - \taupm_i|^p] \to 0$.

To prove the regret consistency of $\taueb_i$, observe that Assumption~\ref{asm:EB} holding for some $p \in [2,\infty)$ automatically implies that it also holds for $p=2$.
By the same argument, we have $\mathbb{E}[ N^{-1}\sum_{i=1}^N |\taueb_i - \taupm_i|^2] \to 0$, from which regret consistency follows.
\end{proof}

\begin{proof}
[Proof of Proposition~\ref{prop:dyp-identification}]
It suffices to verify Assumption~\ref{asm:id-conditions}(i)--(ii) and it is verified in the main text. See Section~\ref{subsubsec:HIVD-AR1} for details.
\end{proof}

\begin{proof}
[Proof of Proposition~\ref{prop:identification-ARMA}]
Again, it suffices to verify Assumption~\ref{asm:id-conditions}(i)--(ii).
We specify the choice of $(\xmat,\Q)$ for Assumption~\ref{asm:id-conditions} in the context of the HIVD model with ARMA(1,1) errors.
For simplicity, we assume $T = 4$.
To illustrate the construction, let $\xmat$ be a column of ones (corresponding to the random intercept) and define the differencing matrix
\begin{equation*}
\Q = \begin{pmatrix}
    -1 & 0 & 0 \\
    1 & -1 & 0 \\
    0 & 1 & -1 \\
    0 & 0 & 1
\end{pmatrix}.
\end{equation*}
This yields the differenced vector $\Q'Y_i = (\Delta u_{i2}, \Delta u_{i3}, \Delta u_{i4})'$.
Consequently, the transformed covariance matrix $\mathcal{V} := \Q' P(\delta) P(\delta)' \Q = \var(\Q' Y_i \mid \delta_i = \delta)$ has elements:
\begin{align*}
\mathcal{V}_{11} &= 2\,\frac{1+\varphi^{2} - \varphi (1-\rho)}{1+\rho}\,\sigma^{2},\\[6pt]
\mathcal{V}_{12} &= \frac{\varphi^{2}(\rho - 1) + \varphi(\rho^{2} - \rho + 2) + (\rho - 1)}{1+\rho}\,\sigma^{2},\\[6pt]
\mathcal{V}_{13} &= \frac{(\rho - 1)(\rho + \varphi)(1+\rho\varphi)}{1+\rho}\,\sigma^{2}.
\end{align*}
Using these moments, $\rho$ is explicitly identified as
\begin{equation*}
    \rho = 1 + \frac{\mathcal{V}_{13}}{\mathcal{V}_{11}/2 + \mathcal{V}_{12}},
\end{equation*}
where the condition $\rho+\varphi \ne 0$ guarantees that the denominator is non-zero.
Once $\rho$ is determined, $\varphi$ and $\sigma^2$ are recovered from the following system:
\begin{align}
    \frac{1+\rho}{2} \mathcal{V}_{11} + (1-\rho) (\rho \mathcal{V}_{12} - \mathcal{V}_{13}) &= (1+\varphi^2) \sigma^2, \label{eq:sys1} \\
    \rho \mathcal{V}_{12} - \mathcal{V}_{13} &= \varphi \sigma^2. \label{eq:sys2}
\end{align}
We distinguish two cases based on the left-hand side of \eqref{eq:sys2}:
\begin{enumerate}
    \item If $\rho \mathcal{V}_{12} - \mathcal{V}_{13} = 0$, then $\varphi = 0$ (since $\sigma^2 > 0$). In this scenario, $\sigma^2$ is directly identified by \eqref{eq:sys1}.
    \item If $\rho \mathcal{V}_{12} - \mathcal{V}_{13} \ne 0$, then $\varphi \ne 0$. We take the ratio of \eqref{eq:sys1} to \eqref{eq:sys2} to eliminate $\sigma^2$:
    \begin{equation*}
        \frac{1+\varphi^2}{\varphi} = \frac{\frac{1+\rho}{2} \mathcal{V}_{11} + (1-\rho) (\rho \mathcal{V}_{12} - \mathcal{V}_{13})}{\rho \mathcal{V}_{12} - \mathcal{V}_{13}} \equiv \mathcal{R}.
    \end{equation*}
    Rearranging this expression leads to the quadratic equation $\varphi^2 - \mathcal{R}\varphi + 1 = 0$.
    From the properties of quadratic equations, $\mathcal{R}$ represents the sum of the two roots, while the constant $1$ represents their product. This implies that the roots are reciprocals, say $z$ and $1/z$.
    The condition $|\mathcal{R}| \ge 2$ (which holds by construction for any real $\varphi$) ensures that the roots are real.
    Finally, the assumption $|\varphi| \le 1$ allows us to uniquely identify $\varphi$ as the root that is less than or equal to 1 in absolute value.
\end{enumerate}
Having identified $\varphi$, if $\varphi \ne 0$, $\sigma^2$ is determined by $\sigma^2 = (\rho \mathcal{V}_{12} - \mathcal{V}_{13}) / \varphi$.
Therefore, we have verified the required one-to-one mapping between $\mathcal{V}$ and $\delta$.
For $T > 4$, let $\Q$ be the $T \times (T-1)$ first-difference matrix; by stationarity, the upper-left $3 \times 3$ block of $\mathcal{V}$ is the matrix above, so the mapping $\delta \mapsto \mathcal{V}$ remains one-to-one.
\end{proof}

\begin{proof}[Proof of Proposition~\ref{prop:HIVDX-id-T=4}]

By assumption, there exists $\pmb x = (\pmb x_1,\pmb x_2,\pmb x_3,\pmb x_4) \in \supp(X_{2,i})$ such that either $\Delta \pmb x_2 = \pm \Delta \pmb x_3 \ne 0$ or $\Delta \pmb x_2 \ne \Delta \pmb x_4$ where $\Delta \pmb x_t = \pmb x_t - \pmb x_{t-1}$.
Upon inspecting the proof of Proposition~\ref{prop:HIVDX-id-T=5}, we find that the assumption $T \ge 5$ is explicitly invoked only in Cases~1(c), 2(c), and 3(c); in the other cases, $T = 4$ suffices for identification.
Further inspection reveals that Case~1(c) requires $(\Delta \pmb x_2, \Delta \pmb x_3, \Delta \pmb x_4) = (0, z, 0)$ for some $z\ne 0$.
Similarly, Cases~2(c) and 3(c) require that $(\Delta \pmb x_2, \Delta \pmb x_3, \Delta \pmb x_4)$ take the forms $(z,0,z)$ for some $z\ne 0$, and $(z,z',z)$ for some $z \ne \pm z'$, $z,z' \ne 0$, respectively.
Since the given $\pmb x$ belongs to none of these cases, the identification follows from the proof of Proposition~\ref{prop:HIVDX-id-T=5}.
\end{proof}

\begin{proof}[Proof of Proposition~\ref{prop:HIVDX-id-T=5}]
We prove Proposition~\ref{prop:HIVDX-id-T=5} by verifying Assumption~\ref{asm:id-conditions}(i)--(ii).
Let $\Delta a_t : = a_{t+1} - a_t$ denote the forward first difference of $a_t$.
By the assumption that $\Delta X_{2,it}$ is not identically zero, there is $\pmb x = (\pmb x_t)_{t=1}^T \in \supp(\X_{2,i})$ such that $\Delta {\pmb x} = (\pmb x_{t+1}-\pmb x_t)_{t=1}^{T-1} \ne 0$.
For notational ease, write $\zmat = \Delta {\pmb x} \in \mathbb{R}^{T-1}$.
Conditioning on $X_{2,i} = \pmb x$, we obtain, after first-differencing,
\begin{equation*}
\Delta Y_{it} = b_i \zmat_t + \Delta u_{it},
\end{equation*}
which yields
\begin{equation*}
\frac{\Delta Y_{it}}{\zmat_t} = b_i + \frac{\Delta u_{it}}{\zmat_t}\quad \text{if}\quad \zmat_t \ne 0,
\end{equation*}
and
\begin{equation*}
\Delta Y_{it} = \Delta u_{it}\quad \text{if}\quad \zmat_t = 0.
\end{equation*}
Let $\mathcal{T} := \{t : \zmat_t \ne 0\}$ be the set of indices at which $\zmat_t \ne 0$, with $t(1) < \cdots < t(K)$ being the increasing enumeration of $\mathcal{T}$ and $K = |\mathcal{T}|$.
The annihilator matrix $\Q$ can be chosen so that
\begin{equation*}
(\Q' Y_i)_k = 
\frac{\Delta u_{i{t(k+1)}}}{\zmat_{t(k+1)}} - \frac{\Delta u_{i{t(k)}}}{\zmat_{t(k)}}  \quad \text{for}\ k=1,\ldots, K-1,
\end{equation*}
for the first $K-1$ elements
(an empty block when $K=1$), and
the remaining $|\mathcal T^c|=T-1-K$ elements start at
component index $K$ and end at $T-2$:
\begin{equation*}
((\Q' Y_i)_{K}, \ldots, (\Q' Y_i)_{T-2}) = (\Delta u_{it})_{t \in \mathcal{T}^c}.
\end{equation*}
The second block is empty when $\mathcal T^c=\varnothing$.

Our objective is to determine $\delta = (\sigma^2, \rho)$ from the covariance $\Q' P(\delta)P(\delta)'\Q = \var(\Q' Y_i\,|\,\delta_i = \delta)$ of the normalized second differences $\frac{\Delta u_{i{t(k+1)}}}{\zmat_{t(k+1)}} - \frac{\Delta u_{i{t(k)}}}{\zmat_{t(k)}}$ for $k=1,\ldots, K-1$ and the first differences $\Delta u_{it}$ for $t \in \mathcal{T}^c$.

The argument proceeds by nested case analysis.
We first consider Case~1.

\medskip

\noindent{\text{\textit{\underline{Case 1}: (Two or more zeros in $\zmat$.)}}} Assume $|\mathcal{T}^c| \ge 2$. For any $s > t$ satisfying $\zmat_t = 0$ and $\zmat_s = 0$, we find that
\begin{equation*}
\var\left( \Delta u_{it}\,|\,\delta_i = \delta \right) = \frac{2\sigma^2}{1+\rho}    
\end{equation*}
and
\begin{equation*}
\cov\left( \Delta u_{it},\Delta u_{is}\,|\,\delta_i = \delta\right) = -\frac{\rho^{s-t-1}(1-\rho)\sigma^2}{1+\rho}.
\end{equation*}
Using these equations, we can pin down the values of
\begin{equation}
\label{eq:pf:case-1}
\frac{\sigma^2}{1+\rho}\quad \text{and}\quad \rho^{|s-t|-1}(1-\rho)\quad \forall t,s \in \mathcal{T}^c,\ t \ne s.
\end{equation}
We further divide into three subcases and show that $\delta$ can be uniquely determined in each subcase.

\medskip

\noindent{– \text{\textit{Case 1(a):}}} When $s$ and $t$ are consecutive, i.e., $s = t+1$, it follows from \eqref{eq:pf:case-1} that both $1-\rho$ and $\sigma^2/(1+\rho)$ are uniquely determined. These together identify $\delta = (\sigma^2, \rho)$.

\medskip

\noindent{– \text{\textit{Case 1(b):}}} Suppose there are at least three elements $t_1 < t_2 < t_3$ in $\mathcal{T}^c$ with $d_1 := t_2-t_1 \ge 2$ and $d_2 := t_3 - t_2 \ge 2$ (otherwise, Case~1(a) applies).
Then, the values of $\rho^{d_1-1}(1-\rho)$, $\rho^{d_2-1}(1-\rho)$, and $\rho^{d_1+d_2-1}(1-\rho)$ are identified from \eqref{eq:pf:case-1}.
If all three are zero, this implies $\rho = 0$, and vice versa. Thus, $\rho$ is identified in this case, and $\sigma^2$ is also identified from $\sigma^2/(1+\rho)$.
Thus, $\delta$ is identified.

Otherwise, all three have nonzero values, and one can pin down
\begin{equation*}
   \frac{\rho^{d_1-1}(1-\rho) \cdot \rho^{d_2-1}(1-\rho)}{\rho^{d_1+d_2-1}(1-\rho)} = \rho^{-1}(1-\rho) = \frac{1}{\rho}-1,
\end{equation*}
implying that $\delta = (\sigma^2, \rho)$ is identified in this case as well.

\medskip

\noindent{– \text{\textit{Case 1(c):}}} Suppose that $|\mathcal{T}^c| = 2$ and the elements $t, s \in \mathcal{T}^c$ ($t<s$) are not consecutive.
Note that $|\mathcal{T}| = T-1-|\mathcal{T}^c|=T-3 \ge 2$, and hence there are at least two time indices for which $\zmat_t \ne 0$. 
If $t=1$ and $s=T-1$, then $\zmat_2,\zmat_3\neq0$ because $T\geq5$ and the
only zeros occur at the endpoints. Using the variances and covariance of adjacent
first differences, we obtain
\begin{equation*}
\var\left(\frac{\Delta u_{i3}}{\zmat_3}
-\frac{\Delta u_{i2}}{\zmat_2}\,\middle|\,\delta_i=\delta\right)
=\frac{\sigma^2}{1+\rho}
\left(\frac{2}{\zmat_3^2}+\frac{2}{\zmat_2^2}
+\frac{2(1-\rho)}{\zmat_2\zmat_3}\right).
\end{equation*}
Dividing by $\sigma^2/(1+\rho)$, already identified by \eqref{eq:pf:case-1},
gives an affine function of $\rho$ with nonzero coefficient
$-2/(\zmat_2\zmat_3)$. Hence $\rho$, and subsequently $\sigma^2$,
are uniquely determined. We may therefore assume that the
two zeros are not both endpoints.
It must therefore be either $t-1,t+1 \in \mathcal{T}$ or $s-1,s+1 \in \mathcal{T}$.
Without loss of generality (reversing time if necessary),
assume that $t-1,t+1\in\mathcal T$.
If the value $\rho^{s-t-1}(1-\rho)$ identified in \eqref{eq:pf:case-1} is zero, then
$\rho=0$ because $s-t\geq2$ and $\rho<1$; $\sigma^2$ is then identified
from $\sigma^2/(1+\rho)$. It therefore remains to consider $\rho\neq0$.
We then observe
\begin{equation*}
\cov\left( \frac{\Delta u_{i{t+1}}}{\zmat_{t+1}} - \frac{\Delta u_{i{t-1}}}{\zmat_{t-1}}, \Delta u_{it}\,|\,\delta_i = \delta\right) = -\frac{(1-\rho)\sigma^2}{1+\rho} \left(\frac{1}{\zmat_{t+1}}-\frac{1}{\zmat_{t-1}} \right),
\end{equation*}
and
\begin{equation*}
\cov\left( \frac{\Delta u_{i{t+1}}}{\zmat_{t+1}} - \frac{\Delta u_{i{t-1}}}{\zmat_{t-1}}, \Delta u_{is}\,|\,\delta_i = \delta\right) = -\frac{(1-\rho)\rho^{s-t-2}\sigma^2}{1+\rho} \left(\frac{1}{\zmat_{t+1}}-\frac{\rho^2}{\zmat_{t-1}} \right).
\end{equation*}
If $\zmat_{t+1}\ne \zmat_{t-1}$, the first equation identifies ${(1-\rho)\sigma^2}/(1+\rho)$. Since $\sigma^2/(1+\rho)$ is already identified by \eqref{eq:pf:case-1}, these together uniquely determine $\rho$ and $\sigma^2$, and hence $\delta$.
If $\zmat_{t+1} = \zmat_{t-1} \ne 0$, we use the second equation to identify
\begin{equation*}
    \frac{(1-\rho)\rho^{s-t-2}\sigma^2}{1+\rho} (1-\rho^2) = \frac{\sigma^2}{1+\rho} \rho^{s-t-1}(1-\rho) \frac{1-\rho^2}{\rho}.
\end{equation*}
Together with \eqref{eq:pf:case-1}, this further identifies
\begin{equation*}
\frac{\rho}{1-\rho^2},
\end{equation*}
which is a strictly increasing function of $\rho \in (-1,1)$.
Hence, $\rho$ is identified, and so is $\delta$.
\bigskip

\noindent{\text{\textit{\underline{Case 2}: (Only one zero in $\zmat$.)}}}
Let $t$ be the only index for which $\zmat_t = 0$. As noted earlier in Case~1, we can then identify
\begin{equation*}
\var\left( \Delta u_{it}\,|\,\delta_i = \delta \right) = \frac{2\sigma^2}{1+\rho}.
\end{equation*}
Again consider three subcases as follows.
\medskip

\noindent{– \text{\textit{Case 2(a):}}} Assume that $t=1$ or $t = T-1$, i.e., that the only zero of $\zmat$ occurs either the first or last time period.
Without loss of generality, let $t = 1$.
Since $2,3 \in \mathcal{T}$, we observe that
\begin{equation*}
\var\left( \frac{\Delta u_{i{3}}}{\zmat_{3}} - \frac{\Delta u_{i{2}}}{\zmat_{2}}\,|\,\delta_i = \delta \right) = \frac{\sigma^2}{1+\rho}\left( \frac{2}{\zmat_{3}^2}+\frac{2}{\zmat_{2}^2}+\frac{2(1-\rho)}{\zmat_{3}\zmat_{2}}\right),
\end{equation*}
from which we can identify
\begin{equation*}
    \left( \frac{2}{\zmat_{3}^2}+\frac{2}{\zmat_{2}^2}+\frac{2(1-\rho)}{\zmat_{3}\zmat_{2}}\right).
\end{equation*}
Clearly, this uniquely determines $\rho$, and $\sigma^2$ subsequently follows from $\sigma^2/(1+\rho)$.
\medskip

\noindent{– \text{\textit{Case 2(b):}}} 
Assume that $1 < t < T-1$ and that $\zmat_{t+1} \ne \zmat_{t-1}$, which must be nonzero.
Observe that
\begin{equation*}
\cov\left( \frac{\Delta u_{i{t+1}}}{\zmat_{t+1}} - \frac{\Delta u_{i{t-1}}}{\zmat_{t-1}}, \Delta u_{it}\,|\,\delta_i = \delta \right) = -\frac{(1-\rho)\sigma^2}{1+\rho} \left( \frac{1}{\zmat_{t+1}} - \frac{1}{\zmat_{t-1}}\right).
\end{equation*}
From this equation, we identify $(1-\rho)\sigma^2/(1+\rho)$.
Combined with $\sigma^2/(1+\rho)$, we identify $\delta$.

\bigskip

\noindent{– \text{\textit{Case 2(c):}}} 
Assume that $1 < t < T-1$ and $\zmat_{t+1} = \zmat_{t-1} \ne 0$. 
Since $t$ is the only zero, either $\zmat_{t+2} \ne 0$ or $\zmat_{t-2} \ne 0$ must hold, depending on whether $t \le T-3$ or $t \ge 3$. (For this, we invoke $T\ge 5$.) Without loss of generality, assume the first.
Let $\zmat \ne 0$ denote the value shared by $\zmat_{t+1} = \zmat_{t-1}$. We use
\begin{equation*}
\var\left( \frac{\Delta u_{i{t+1}}}{\zmat_{t+1}} - \frac{\Delta u_{i{t-1}}}{\zmat_{t-1}}\,|\,\delta_i = \delta \right) = \frac{2 \sigma^2}{\zmat^2(1+\rho)}\left( 2 +  \rho (1-\rho) \right),
\end{equation*}
along with $\sigma^2/(1+\rho)$, to identify $2+\rho(1-\rho)$.
Next, we employ
\begin{equation*}
\cov\left( \frac{\Delta u_{i{t+2}}}{\zmat_{t+2}} - \frac{\Delta u_{i{t+1}}}{\zmat_{t+1}}, \Delta u_{it}\,|\,\delta_i = \delta \right) = -\frac{(1-\rho)\sigma^2}{1+\rho}\left( \frac{\rho}{\zmat_{t+2}} - \frac{1}{\zmat_{t+1}}\right),
\end{equation*}
together with $\sigma^2/(1+\rho)$, to identify $\frac{(1-\rho)\rho}{\zmat_{t+2}} - \frac{(1-\rho)}{\zmat_{t+1}}$.
Lastly, $\frac{(1-\rho)\rho}{\zmat_{t+2}} - \frac{(1-\rho)}{\zmat_{t+1}}$, combined with $2+\rho(1-\rho)$ identified earlier, can be used to identify $\rho$. Indeed, $2+\rho(1-\rho)$ determines $\rho$ up to the replacement of $\rho$ by $1-\rho$, and writing $h(\rho) := \frac{(1-\rho)\rho}{\zmat_{t+2}} - \frac{(1-\rho)}{\zmat_{t+1}}$, we have $h(\rho) - h(1-\rho) = (2\rho-1)/\zmat_{t+1}$, which is nonzero unless $\rho = 1-\rho = 1/2$. $\sigma^2$ is subsequently identified from $\sigma^2/(1+\rho)$.

\bigskip

\noindent{\text{\textit{\underline{Case 3}: (No zeros in $\zmat$.)}}} 
We show that $\delta$ can be identified using the covariance structure of
\begin{equation*}
\left( \frac{\Delta u_{i{2}}}{\zmat_{2}} - \frac{\Delta u_{i{1}}}{\zmat_{1}},\quad \frac{\Delta u_{i{3}}}{\zmat_{3}}-\frac{\Delta u_{i{2}}}{\zmat_{2}},\quad \frac{\Delta u_{i{4}}}{\zmat_{4}} - \frac{\Delta u_{i{3}}}{\zmat_{3}}\right).
\end{equation*}
Although the last term requires $T \ge 5$, Cases~3(a) and~(b) remain applicable under $T = 4$, in which the first two terms are sufficient for identification.

We begin by noting that
\begin{equation*}
\var\left( \frac{\Delta u_{i{t+1}}}{\zmat_{t+1}} - \frac{\Delta u_{i{t}}}{\zmat_{t}}\,|\,\delta_i = \delta \right) = \frac{2\sigma^2}{1+\rho}\left( 1 + \frac{\zmat_{t+1}}{\zmat_{t}}+\frac{\zmat_{t}}{\zmat_{t+1}} - \rho \right)\frac{1}{\zmat_{t+1}\zmat_{t}},\quad t = 1,2,3.
\end{equation*}
Consider the following three subcases.

\medskip

\noindent{– \text{\textit{Case 3(a):}}}
Assume that there exist $s < t$ such that $\frac{\zmat_{s+1}}{\zmat_{s}}+\frac{\zmat_{s}}{\zmat_{s+1}} \ne \frac{\zmat_{t+1}}{\zmat_{t}}+\frac{\zmat_{t}}{\zmat_{t+1}}$.
We may assume $s=1$ and $t = 2$. Then, we can use the equation
\begin{equation*}
\frac{\zmat_{2}\zmat_{1}\var\left( \frac{\Delta u_{i{2}}}{\zmat_{2}} - \frac{\Delta u_{i{1}}}{\zmat_{1}}\,|\,\delta_i = \delta \right)}{\zmat_{3}\zmat_{2}\var\left( \frac{\Delta u_{i{3}}}{\zmat_{3}} - \frac{\Delta u_{i{2}}}{\zmat_{2}}\,|\,\delta_i = \delta \right)}
= \frac{1 + \frac{\zmat_{2}}{\zmat_{1}}+\frac{\zmat_{1}}{\zmat_{2}} - \rho}{1 + \frac{\zmat_{3}}{\zmat_{2}}+\frac{\zmat_{2}}{\zmat_{3}} - \rho} = 1 + \frac{(\frac{\zmat_{2}}{\zmat_{1}}+\frac{\zmat_{1}}{\zmat_{2}}) - (\frac{\zmat_{3}}{\zmat_{2}}+\frac{\zmat_{2}}{\zmat_{3}})}{1 + \frac{\zmat_{3}}{\zmat_{2}}+\frac{\zmat_{2}}{\zmat_{3}} - \rho}
\end{equation*}
to pin down $\rho$, which then identifies $\delta$.

\medskip 

If Case 3(a) does not hold, $\frac{\zmat_{t+1}}{\zmat_{t}}+\frac{\zmat_{t}}{\zmat_{t+1}}$ must be constant in $t$, which is equivalent to the existence of some $\kappa \ne 0$ such that
\begin{equation*}
   r_t := \frac{\zmat_{t+1}}{\zmat_{t}} \in \left\{\kappa , \frac{1}{\kappa} \right\} \ \text{for}\ t=1,2,3.
\end{equation*}
This implies that if Case 3(a) does not hold, we can identify
\begin{equation}
\label{eq:pf:id-quantity-1}
\frac{\sigma^2}{1+\rho}\left( 1 + \kappa + \frac{1}{\kappa} - \rho \right)
\end{equation}
from the equation above.
Further using covariances, we identify
\begin{align}
\label{eq:pf:id-quantity-2}
& \cov \left( \frac{\Delta u_{i{t+2}}}{\zmat_{t+2}} - \frac{\Delta u_{i{t+1}}}{\zmat_{t+1}}, \frac{\Delta u_{i{t+1}}}{\zmat_{t+1}} - \frac{\Delta u_{i{t}}}{\zmat_{t}}\,|\,\delta_i = \delta \right) \nonumber \\
= &\ \frac{1}{\zmat_{t+1}^2}  \cov \left( \frac{1}{r_{t+1}}\Delta u_{i{t+2}} -  \Delta u_{i{t+1}}, \Delta u_{i{t+1}} - r_t \Delta u_{i{t}}\,|\,\delta_i = \delta \right) \nonumber \\
= &\ -\frac{1}{\zmat_{t+1}^2} \frac{\sigma^2}{1+\rho} \left( 2 + (1-\rho)\left(\frac{1}{r_{t+1}}+r_{t}-\frac{r_t}{r_{t+1}}\rho\right)\right).
\end{align}
Now proceed to Cases~3(b) and 3(c).

\medskip

\noindent{– \text{\textit{Case 3(b):}}}
Assume that there are consecutive equal terms in $r_t$.
Without loss of generality, assume that $r_1 = r_2 = \kappa$.
Substituting $r_1$ and $r_2$ into \eqref{eq:pf:id-quantity-2} with $t=1$, the ratio of \eqref{eq:pf:id-quantity-2} to \eqref{eq:pf:id-quantity-1} identifies
\begin{equation}
\label{eq:pf:id-quantity-3}
\frac{-\frac{1}{\zmat_{t+1}^2} \frac{\sigma^2}{1+\rho} \left( 2 + (1-\rho)\left(\frac{1}{r_{t+1}}+r_{t}-\frac{r_t}{r_{t+1}}\rho\right)\right)}{\frac{\sigma^2}{1+\rho}\left( 1 + \kappa + \frac{1}{\kappa} - \rho \right)} = C \frac{ (1-\rho)( \kappa + 1/\kappa -\rho)+2}{1 + \kappa + 1/\kappa - \rho},
\end{equation}
where $C \ne 0$ is a known constant.
Define $\alpha := \kappa + \frac{1}{\kappa}$, which satisfies $\alpha \ge 2$ when $\kappa > 0$, and $\alpha \le -2$ when $\kappa < 0$.
From the expression above, we can identify
\begin{equation*}
-\frac{ (1-\rho)( \kappa + 1/\kappa -\rho)+2}{1 + \kappa + 1/\kappa - \rho} = \rho + \frac{\alpha + 2}{\rho - (\alpha + 1)},
\end{equation*}
whose value is denoted by $A$.
If $\alpha \le -2$, the function $\rho \mapsto \rho + {\alpha + 2}/(\rho - (\alpha + 1))$ is strictly increasing in $\rho \in (-1,1)$, which identifies $\rho$, and subsequently $\sigma^2$ from $\sigma^2/(1+\rho)(1+\kappa +1/\kappa - \rho)$.
Otherwise, if $\alpha \ge 2$, the equation
\begin{equation}
\label{eq:pf:id-equation}
\rho + \frac{\alpha + 2}{\rho - (\alpha + 1)} = A
\end{equation}
admits at most two distinct solutions $\rho_1$ and $\rho_2$, which are linked by
\begin{equation*}
\rho_2 = \alpha + 1 + \frac{\alpha+2}{\rho_1-(\alpha+1)},
\end{equation*}
and vice versa symmetrically.
Suppose $\rho_1$ lies in $(-1,1)$. Consequently,
\begin{equation*}
\rho_2 = \alpha + 1 + \frac{\alpha+2}{\rho_1-(\alpha+1)} \ge \alpha + 1 + \frac{\alpha+2}{1-(\alpha+1)} = \alpha - \frac{2}{\alpha} \ge 1,
\end{equation*}
which necessarily falls outside $(-1,1)$.
Hence, \eqref{eq:pf:id-equation} yields at most one solution in $(-1,1)$, thereby identifying $\rho$, and subsequently $\delta$.
\medskip

\noindent{– \text{\textit{Case 3(c):}}}\ \
Consider the case in which no consecutive terms in $(r_1,r_2,r_3)$ are equal, i.e., $(r_1,r_2,r_3)$ are alternating between $\kappa$ and $1/\kappa$ with $\kappa \ne \pm 1$.
Without loss of generality, let $r_1 = r_3 = \kappa$ and $r_2 = 1/\kappa$.
We observe that
\begin{align*}
\zmat_{3}\zmat_{2} \cov \left( \frac{\Delta u_{i{3}}}{\zmat_{3}} - \frac{\Delta u_{i{2}}}{\zmat_{2}}, \frac{\Delta u_{i{2}}}{\zmat_{2}} - \frac{\Delta u_{i{1}}}{\zmat_{1}}\,|\,\delta_i = \delta \right) 
&= -\frac{\sigma^2}{1+\rho} \left( \frac{2}{\kappa} + (1-\rho)\left(2-\kappa\rho\right)\right),\\
\zmat_{4}\zmat_{3}\cov \left( \frac{\Delta u_{i{4}}}{\zmat_{4}} - \frac{\Delta u_{i{3}}}{\zmat_{3}}, \frac{\Delta u_{i{3}}}{\zmat_{3}} - \frac{\Delta u_{i{2}}}{\zmat_{2}}\,|\,\delta_i = \delta \right)
&=  - \frac{\sigma^2}{1+\rho} \left( 2\kappa + (1-\rho)\left(2-\rho/\kappa\right)\right).
\end{align*}
Taking the difference of these equations, we identify
\begin{equation*}
\left( \kappa - \frac{1}{\kappa}\right)\frac{\sigma^2}{1+\rho} \left( 2 + (1-\rho)\rho \right),
\end{equation*}
where $\kappa - 1/\kappa \ne 0$ is known.
Additionally, we know the value of
\begin{equation*}
\zmat_{4}\zmat_{2}\cov \left( \frac{\Delta u_{i{4}}}{\zmat_{4}} - \frac{\Delta u_{i{3}}}{\zmat_{3}}, \frac{\Delta u_{i{2}}}{\zmat_{2}} - \frac{\Delta u_{i{1}}}{\zmat_{1}}\,|\,\delta_i = \delta \right)
= \frac{\sigma^2}{1+\rho} (1-\rho)(\kappa\rho - 1)(\rho-\kappa).
\end{equation*}
Combined with \eqref{eq:pf:id-quantity-1} identified earlier, these two terms are used to identify the following:
\begin{equation*}
-\frac{2 + (1-\rho)\rho}{\rho - (\alpha+1)} = \rho+\alpha  + \frac{(\alpha+2)(\alpha-1)}{\rho - (\alpha+1)}
\end{equation*}
and
\begin{equation*}
-\frac{(1-\rho)(\kappa\rho - 1)(\rho-\kappa)}{\rho-(\alpha+1)} = \kappa \left( \rho^2 + \alpha + 1 + \frac{\alpha(\alpha+2)}{\rho- (\alpha+1)} \right).
\end{equation*}
Since $\alpha$ and $\kappa$ are known, this information is equivalent to knowing $\rho + {(\alpha+2)(\alpha-1)}/({\rho - (\alpha+1)})$ and $\rho^2 + {\alpha(\alpha+2)}/({\rho- (\alpha+1)})$, which can be used to identify
\begin{equation*}
\left[ \rho^2 + \frac{\alpha(\alpha+2)}{\rho- (\alpha+1)}\right] - \frac{\alpha}{\alpha-1} \left[ \rho + \frac{(\alpha+2)(\alpha-1)}{\rho - (\alpha+1)}\right] = \rho^2 - \frac{\alpha}{\alpha-1}\rho.    
\end{equation*}
Label the known values of $\rho + \alpha + {(\alpha+2)(\alpha-1)}/({\rho - (\alpha+1)})$ and $\rho^2 - \alpha\rho/(\alpha-1)$ as $B$ and $C$, respectively.
Viewed as rational equations for $\rho$, this leads to the following system of quadratic equations:
\begin{equation}
\label{eq:pf:system-id-eqs}
\begin{aligned}    
& \rho^2 - (B+1)\rho + B(\alpha+1)-2 = 0,\\
& \rho^2 - \frac{\alpha}{\alpha-1}\rho -C = 0.
\end{aligned}
\end{equation}
We show by way of contradiction that \eqref{eq:pf:system-id-eqs} can be consistent with at most one $\rho$ in $(-1,1)$, thus uniquely determining $\rho$.
For the quadratic system \eqref{eq:pf:system-id-eqs} to admit multiple solutions, note that the two equations must share the same roots, and hence coincide, i.e.,
\begin{equation*}
B = \frac{1}{\alpha-1},\quad C = 1-\frac{2}{\alpha-1}.
\end{equation*}
Plugging $C = 1-\frac{2}{\alpha-1}$ into the second equation, the proof is complete once we verify that
\begin{equation*}
   \rho^2 - \frac{\alpha}{\alpha-1}\rho + \frac{3-\alpha}{\alpha-1} = 0
\end{equation*}
admits at most one solution in $(-1,1)$.
To show this, evaluate this function at $\rho = 1$:
\begin{equation*}
1 -    \frac{\alpha}{\alpha-1} + \frac{3-\alpha}{\alpha-1} = \frac{2-\alpha}{\alpha-1}.
\end{equation*}
Similarly, at $\rho = -1$, we obtain
\begin{equation*}
1 + \frac{\alpha}{\alpha-1}+\frac{3-\alpha}{\alpha-1} = \frac{\alpha+2}{\alpha-1}.
\end{equation*}
Since $\alpha = \kappa + 1/\kappa$ satisfies $|\alpha| > 2$, we have
\begin{equation*}
\frac{2-\alpha}{\alpha-1} \frac{\alpha+2}{\alpha-1} = \frac{4-\alpha^2}{(\alpha-1)^2} < 0,
\end{equation*}
which shows that the function $\rho \mapsto \rho^2 - \frac{\alpha}{\alpha-1}\rho + \frac{3-\alpha}{\alpha-1}$ takes values of opposite signs at $\rho = 1$ and $\rho = -1$. Hence, it must have a unique root in $(-1,1)$, leading to a contradiction.
This concludes Case~3(c).

\end{proof}

\begin{proof}[Proof of Proposition~\ref{prop:EB-HIVDX}]
We verify parts~(i), (ii), and (iii) sequentially.

\noindent \textit{Proof of part~(i):}
The regret consistency of $\hat{\sigma}^{2,\EB}_i$ and $\hat{\rho}^{\EB}_i$ is an immediate consequence of Theorem~\ref{thm:EB-consistency} because $(\sigma_i^2, \rho_i)$ have compact support by Assumption~\ref{asm:DGP}.

\bigskip

\noindent \textit{Proof of part~(ii):}
We show that Assumption~\ref{asm:EB}(ii) holds for $\Tau_i \in \{a_i, b_i\}$ with $p = 2$.
By Markov's inequality, we have
\begin{align*}
\int_{\Theta} \|\hc\|^2 1\{\|\beta\| \ge K\} d\hat{G}_N(\theta) \le \frac{1}{K^\varepsilon}\int_{\Theta} \|\hc\|^{2+\varepsilon} d\hat{G}_N(\theta).
\end{align*}
It follows from Lemma~\ref{lem:prior-mean-bound} that
\begin{equation*}
\int_{\Theta} \|\hc\|^{2} 1\{\|\beta\| \ge K\} d\hat{G}_N(\theta) \le \frac{1}{K^\varepsilon}\frac{A^{2+\varepsilon}}{N} \sum_{i=1}^N \|Y_i\|^{2+\varepsilon},
\end{equation*}
and thus,
\begin{equation*}
\limsup_{N \to \infty} \mathbb{E} \left[ \int_{\Theta} \|\hc\|^2 1\{\|\beta\| \ge K\} d\hat{G}_N(\theta) \right] \le \frac{A^{2+\varepsilon}}{K^\varepsilon} \mathbb{E}[\|Y_i\|^{2+\varepsilon}].
\end{equation*}
One can make the right-hand side arbitrarily close to zero by choosing a sufficiently large $K$.
Since $a_i$ and $b_i$ are dominated in absolute value by $\|\hc_i\|$, this proves Assumption~\ref{asm:EB}(ii) for these parameters.
It follows from Theorem~\ref{thm:EB-consistency} that $\hat{a}^{\EB}_i$ and $\hat{b}^{\EB}_i$ are regret-consistent.

\bigskip 

\noindent \textit{Proof of part~(iii):}
By Cauchy-Schwarz inequality, we observe that
\begin{align*}
& \frac{1}{N} \sum_{i=1}^N (\hat{Y}^{\EB}_{iT+1} - Y_{iT+1}^{\oracle})^2/5 \\
\le &\ \frac{1}{N} \sum_{i=1}^N(\hat{a}_{i}^{\EB}- a_{i}^{\oracle})^2 + \frac{1}{N} \sum_{i=1}^N(\hat{b}_{i}^{\EB}- b_{i}^{\oracle})^2 |X_{2,i T+1}|^2 + \frac{1}{N} \sum_{i=1}^N(\hat{\rho}_{i}^{\EB}- \rho_{i}^{\oracle})^2 |Y_{iT}|^2 \\
&\quad + \frac{1}{N} \sum_{i=1}^N(\widehat{a\rho}_{i}^{\EB}- (a\rho)_{i}^{\oracle})^2 +  \frac{1}{N} \sum_{i=1}^N(\widehat{b\rho}_{i}^{\EB}- (b\rho)_{i}^{\oracle})^2  |X_{2,i,T}|^2 \\
\le &\ \frac{1}{N} \sum_{i=1}^N(\hat{a}_{i}^{\EB}- a_{i}^{\oracle})^2 + \sqrt{\frac{1}{N} \sum_{i=1}^N(\hat{b}_{i}^{\EB}- b_{i}^{\oracle})^4} \sqrt{\frac{1}{N} \sum_{i=1}^N |X_{2,i T+1}|^4}  \\
&\quad + \sqrt{\frac{1}{N} \sum_{i=1}^N(\hat{\rho}_{i}^{\EB}- \rho_{i}^{\oracle})^4}\sqrt{ \frac{1}{N} \sum_{i=1}^N |Y_{iT}|^4 }\\
&\quad + \frac{1}{N} \sum_{i=1}^N(\widehat{a\rho}_{i}^{\EB}- (a\rho)_{i}^{\oracle})^2 + \sqrt{\frac{1}{N} \sum_{i=1}^N(\widehat{b\rho}_{i}^{\EB}- (b\rho)_{i}^{\oracle})^4} \sqrt{\frac{1}{N} \sum_{i=1}^N |X_{2,i,T}|^4},
\end{align*}
where $(a\rho)_i := a_i \rho_i$ and $(b\rho)_i := b_i \rho_i$.
Taking expectations on both sides, it suffices to show that the expectations of the terms on the RHS converge to zero.

Observe that $\Tau_i \in \{a_i, b_i, a_i\rho_i, b_i\rho_i\}$ are all dominated in absolute value by $\|\hc_i\|$ since $|\rho_i|\le 1$.
Combining the moment condition $\mathbb{E}[\|Y_i\|^{4+\varepsilon}] < \infty$ with the argument from Part~(ii), we find that Assumption~\ref{asm:EB}(ii) holds for $\Tau_i \in \{a_i, b_i, a_i\rho_i, b_i\rho_i\}$ with $p = 4$.
By Theorem~\ref{thm:EB-consistency} (with $p=2$), we have $\mathbb{E}[N^{-1} \sum_{i=1}^N(\hat{a}_{i}^{\EB}- a_{i}^{\oracle})^2] \to 0$ and $\mathbb{E}[N^{-1} \sum_{i=1}^N(\widehat{a\rho}_{i}^{\EB}- (a\rho)_{i}^{\oracle})^2]\to 0$ immediately.
This addresses the first and fourth terms on the RHS.
Applying Cauchy-Schwarz inequality, the second term on the RHS is bounded above by
\begin{multline*}
\mathbb{E}\left[ \sqrt{\frac{1}{N} \sum_{i=1}^N(\hat{b}_{i}^{\EB}- b_{i}^{\oracle})^4} \sqrt{\frac{1}{N} \sum_{i=1}^N |X_{2,i T+1}|^4}\right] \\
\le \sqrt{\mathbb{E}\left[ \frac{1}{N} \sum_{i=1}^N(\hat{b}_{i}^{\EB}- b_{i}^{\oracle})^4\right]}\sqrt{\mathbb{E}\left[ \frac{1}{N} \sum_{i=1}^N |X_{2,i T+1}|^4\right]}.
\end{multline*}
Theorem~\ref{thm:EB-consistency} implies that $\mathbb{E}[ N^{-1} \sum_{i=1}^N(\hat{b}_{i}^{\EB}- b_{i}^{\oracle})^4] \to 0$, while $\mathbb{E}[ N^{-1}\sum_{i=1}^N |X_{2,i T+1}|^4] = \mathbb{E}[|X_{2,i T+1}|^4]$ is finite by assumption. Thus, the second term converges to zero.
The third and fifth terms can be handled in an analogous way. We omit the details to avoid repetition.
Each term on the right-hand side has expectation converging to zero. Hence
$\mathbb{E}[N^{-1}\sum_{i=1}^N (\hat{Y}^{\EB}_{iT+1} - Y_{iT+1}^{\oracle})^2 ]\to 0$ as desired, which establishes the regret consistency of $\hat{Y}_{iT+1}^{\EB}$.

\end{proof}

\subsection{Identification Lemma (Theorem~\ref{thm:identification})}

We establish an identification result for finite Gaussian mixtures (Lemma~\ref{lem:finite-mixture}), which is used in the proof of Theorem~\ref{thm:identification}.
Let $n \in \mathbb{N}$, $d \in \mathbb{N}$, and let $\Theta \subseteq \mathbb{R}^n$ denote a compact parameter space for $\theta$.
Following the notation in \citet{bruni1985identifiability}, let $\lambda_1 : \Theta \to \mathbb{R}^d$ and $\lambda_2 : \Theta \to S_{+,d}$ be continuous functions, where $S_{+,d}$ denotes the set of $d\times d$ positive definite matrices.
The pair $(\lambda_1(\theta), \lambda_2(\theta))$ represents the conditional mean and covariance of $Y_i \,|\, \theta_i \sim \mathcal{N}(\lambda_1(\theta_i), \lambda_2(\theta_i))$, where $\theta_i \in \Theta$ represents a random parameter.

We write $\mathcal{M}(\Theta)$ for the set of complex Borel measures on $\Theta$, and $\mathcal{P}(\Theta)$ for its subset consisting of probability measures.
Finally, for such $(\lambda_1,\lambda_2)$ and $\mu \in \mathcal{M}(\Theta)$, define $\mathcal{T}[\lambda_1,\lambda_2,\mu] \in C(\mathbb{R}^d,\mathbb{C})$ by
\begin{equation*}
\mathcal{T}[\lambda_1,\lambda_2, \mu](y) := \int \frac{1}{(2\pi)^{d/2} |\lambda_2(\theta)|^{1/2}}\exp\left( -\frac{1}{2}\|\lambda_2(\theta)^{-1/2} (y - \lambda_1(\theta))\|^2\right) d\mu(\theta),
\end{equation*}
where $C(\mathbb{R}^d,\mathbb{C})$ denotes the class of all complex-valued continuous functions on $\mathbb{R}^d$.
Note that, for $\mu \in \mathcal{P}(\Theta)$, $\mathcal{T}[\lambda_1,\lambda_2,\mu](x)$ represents the density of $\int_{\Theta} \mathcal{N}(\lambda_1(\theta), \lambda_2(\theta)) d\mu(\theta)$, i.e., the marginal density of
\begin{equation*}
Y_i = \lambda_1(\theta_i) + \lambda_2(\theta_i)^{1/2} e_i
\end{equation*}
under $\theta_i \sim \mu$ and $e_i \,|\, \theta_i \sim \mathcal{N}(0, I_d)$.
The following lemma shows that a finitely supported mixing distribution is identified from $\mathcal{T}$.
\begin{lem}
[Identification of Finite Gaussian Mixtures]
\label{lem:finite-mixture}
Let $(\lambda_1,\lambda_2)$ be as above, with $\theta \mapsto (\lambda_1(\theta),\lambda_2(\theta))$ injective on $\Theta$, and let $\mu, \tilde{\mu} \in \mathcal{P}(\Theta)$ satisfy
\begin{equation*}
\mathcal{T}[\lambda_1,\lambda_2, \mu](y) = \mathcal{T}[\lambda_1,\lambda_2, \tilde{\mu}](y) \quad \text{for all } y \in \mathbb{R}^d.
\end{equation*}
Then, $\mu = \tilde{\mu}$ if either of the following holds:
\begin{enumerate}
\item[(a)] both $\mu$ and $\tilde{\mu}$ have finite support;
\item[(b)] $\lambda_1 \equiv 0$ and $\mu$ has finite support.
\end{enumerate}
\end{lem}

\begin{proof}[Proof of Lemma~\ref{lem:finite-mixture}]

\noindent \textit{Case~(a)}.
Let $\theta_1,\ldots,\theta_J$ denote the distinct points in the union of the supports of $\mu$ and $\tilde{\mu}$, and write $c_j := \mu(\{\theta_j\}) - \tilde{\mu}(\{\theta_j\})$ and $(m_j, V_j) := (\lambda_1(\theta_j), \lambda_2(\theta_j))$ for $j = 1,\ldots,J$.
The pairs $(m_j, V_j)$ are distinct by injectivity, and the hypothesis of the lemma can be written as the linear relation
\begin{equation*}
\sum_{j=1}^J c_j\, \phi(y; m_j, V_j) = 0 \quad \text{for all } y \in \mathbb{R}^d,
\end{equation*}
where $\phi(\cdot\,; m, V)$ denotes the density of $\mathcal{N}(m, V)$; integrating it gives the same relation for the corresponding distribution functions.
By \citet[Proposition 2]{YakowitzSpragins:1968}, the family of multivariate Gaussian distributions generates identifiable finite mixtures, which is equivalent to its linear independence over the real numbers by the theorem in their Section 3.
Hence, the above relation is trivial, that is, $c_j = 0$ for all $j$, and thus $\mu = \tilde{\mu}$.

\medskip

\noindent \textit{Case~(b)}.
Let $\theta_1,\ldots,\theta_J$ denote the support points of $\mu$.
Let $u \in \mathbb{R}^d$ be arbitrary, and let $q_u(\theta) := u'\lambda_2(\theta)u \ge 0$.
Since $\lambda_1 \equiv 0$, the characteristic function of $u'Y$, where $Y$ has density $\mathcal{T}[0,\lambda_2,\mu]$, is $r \mapsto \int \exp(-r^2 q_u(\theta)/2)\, d\mu(\theta)$, and similarly for $\tilde{\mu}$.
By hypothesis, $Y$ has the same density under $\mu$ and $\tilde{\mu}$, so that
\begin{equation*}
\int \exp\left( -\tfrac{1}{2}\, r^2 q_u(\theta) \right) d\mu(\theta) = \int \exp\left( -\tfrac{1}{2}\, r^2 q_u(\theta) \right) d\tilde{\mu}(\theta) \quad \text{for all } r \in \mathbb{R}.
\end{equation*}
The two sides are the Laplace transforms, evaluated at $r^2/2$, of the distributions of the nonnegative random variable $q_u(\theta)$ under $\theta \sim \mu$ and $\theta \sim \tilde{\mu}$, respectively.
Since the Laplace transform determines the distribution of a nonnegative random variable, the two distributions coincide; in particular, both are supported on the finite set $\{q_u(\theta_1),\ldots,q_u(\theta_J)\}$.

Now suppose that some $\tilde{\theta} \in \supp(\tilde{\mu})$ does not belong to $\{\theta_1,\ldots,\theta_J\}$.
Because $q_u$ is continuous (as $\lambda_2$ is), $q_u(\tilde{\theta})$ lies in the support of the distribution of $q_u(\theta)$ under $\tilde{\mu}$, so that $q_u(\tilde{\theta}) = q_u(\theta_j)$ for some $j$.
As $u$ is arbitrary, it follows that
\begin{equation*}
\mathbb{R}^d = \bigcup_{j=1}^J \left\{ u \in \mathbb{R}^d : u'\big(\lambda_2(\tilde{\theta}) - \lambda_2(\theta_j)\big)u = 0 \right\}.
\end{equation*}
However, since $\lambda_1 \equiv 0$, injectivity implies that each $\lambda_2(\tilde{\theta}) - \lambda_2(\theta_j)$ is a nonzero symmetric matrix, so each set on the right-hand side is the zero set of a nonzero quadratic form and has Lebesgue measure zero.
This is a contradiction.
Therefore, $\tilde{\mu}$ is supported on $\{\theta_1,\ldots,\theta_J\}$, and Case~(a) yields $\mu = \tilde{\mu}$.

\end{proof}

\subsection{Consistency Lemmas (Theorem~\ref{thm:consistency})}
Let $\psub$ denote the class of all subprobability measures on $\Theta$.
We extend the metric $d$ to $\psub$ by defining
\begin{equation}
\label{eq:def-metric}
    d(\mu_0, \mu_1) := \sum_{r=1}^\infty \frac{1}{2^r} \left| \int_{\Theta} h_r(\theta)d\mu_0(\theta) -  \int_{\Theta} h_r(\theta)d\mu_1(\theta) \right|,
\end{equation}
where $(h_r)_{r=1}^\infty$ is the same as in Section~\ref{sec:consistency}.
Let $\bar{\mathcal{G}}$ denote the closure of $\mathcal{G}$ relative to $\psub$ with respect to the metric $d$.
Since $d$ induces the vague topology on $\psub$, and $\psub$ is compact under this topology, $(\bar{\mathcal{G}}, d)$ is itself a compact metric space (\citeay{kallenbergFoundationsModernProbability2021}).
We extend $f_{G}(Y, \X)$, originally defined for $(G,Y,\X) \in \mathcal{G} \times \mathbb{R}^T \times \SX$, to $(\mu,Y,\X) \in \bar{\mathcal{G}}\times \mathbb{R}^T \times \SX$ by
\begin{equation}
\label{eq:extend-likelihood}
f_\mu(Y,\X) := \int_{\Theta} \ell(Y \,|\, \X , \theta) d\mu(\theta).    
\end{equation}
The next lemma establishes that the extended likelihood is continuous in $\mu \in \bar{\mathcal{G}}$.

\begin{lem}
\label{lem:likelihood-continuity}
Let Assumptions~\ref{asm:DGP} and \ref{asm:consistency} hold.
Then, the mapping $\mu \mapsto f_\mu(Y,\X) : \bar{\mathcal{G}} \to \mathbb{R}$ is continuous for all $Y \in \mathbb{R}^{T}$ and all $\X \in \SX$ with $\rank(\X) = \dbeta$.
\end{lem}

\begin{proof}[Proof of Lemma~\ref{lem:likelihood-continuity}]
Note that
\begin{equation*}
\ell(Y\,|\,\X,\theta) = \frac{1}{(2\pi)^{T/2}|P(\delta)|} \exp \left( -\frac{1}{2} \|P(\delta)^{-1}(Y - \X \hc)\|^2\right) \le C_1 \exp \left( -C_2 \|Y - \X \hc\|^2\right)
\end{equation*}
for some constants $C_1, C_2 > 0$ by Assumption~\ref{asm:DGP}(ii).
For $\X$ with $\rank(\X) = \dbeta$, which holds with probability one by Assumption~\ref{asm:consistency}(ii), it follows that $\|Y - \X \hc\|^2\to\infty$ as $\|\hc\|\to\infty$, and hence $\lim_{\|\theta\|\to\infty} |\ell(Y \,|\, \X, \theta)| = 0$.
Let $\varphi_{Y,\X}(\theta) := \ell(Y\,|\,\X,\theta)$.
This implies that $\varphi_{Y,\X} \in C_0(\Theta)$ for all such pairs $(Y, \X)$, where $C_0(\Theta)$ denotes the class of continuous functions $h : \Theta \to \mathbb{R}$ such that $\lim_{\|\theta\|\to \infty} |h(\theta)| = 0$.
It is well-known that $C_0(\Theta)$ is the completed space of $C_c(\Theta)$ with respect to the uniform norm $\|h\|_\infty = \sup_{\theta\in\Theta}|h(\theta)|$ (see, e.g., Section 21.4 of \citeay{RoydenFitzpatrick2010}).

Now, consider a sequence $\mu_n \in \bar{\mathcal{G}}$ such that $\mu_n$ vaguely converges to $\mu \in \bar{\mathcal{G}}$, i.e., $\lim_{n \to \infty} d(\mu_n, \mu) = 0$.
Fix such a pair $(Y,\X)$.
Since $(h_r)_{r=1}^\infty$ in \eqref{eq:def-metric} spans a dense subset of $C_c(\Theta)$, and hence of $C_0(\Theta)$, for any $\varepsilon > 0$, there exists a real sequence $(a_r)_{r=1}^{\bar{r}}$ with $\bar{r}<\infty$ such that $\left\|\varphi_{Y,\X} - \sum_{r=1}^{\bar{r}} a_r h_r\right\|_\infty < \varepsilon$.
This implies
\begin{align*}
|f_{\mu_n}(Y,\X)-f_\mu(Y,\X)| & = \left| \int_{\Theta}\varphi_{Y,X}(\theta)d\mu_n(\theta) - \int_{\Theta}\varphi_{Y,X}(\theta)d\mu(\theta) \right| \\
&\le \left| \int_{\Theta} \sum_{r=1}^{\bar{r}} a_r h_r(\theta) d\mu_n(\theta) - \int_{\Theta} \sum_{r=1}^{\bar{r}} a_r h_r(\theta) d\mu(\theta) \right|  + 2\varepsilon \\
&\le \sup_{1\le r \le \bar{r}} |a_r 2^r| \cdot d(\mu_n, \mu) + 2\varepsilon
\end{align*}
by construction of $d$, and hence
\begin{equation*}
\limsup_{n\to\infty} |f_{\mu_n}(Y,\X)-f_\mu(Y,\X)| \le 2\varepsilon.
\end{equation*}
Since $\varepsilon>0$ can be made arbitrarily small, it follows that $|f_{\mu_n}(Y,\X)-f_\mu(Y,\X)| \to 0$, completing the proof.

\end{proof}

\begin{lem}
\label{lem:key-to-consistency}

Let Assumptions~\ref{asm:DGP}, \ref{asm:id-conditions}, and \ref{asm:consistency} hold.
Then, the following hold.
\begin{enumerate}

\item [(i)]
$
\mathbb{E}_{G_*}[ |\log f_{G_*}(Y_i, \X_i)|] < \infty.
$

\item [(ii)] 
For any $\varepsilon > 0$ and $\mu \in \bar{\mathcal{G}}$,
\begin{equation*}
\bar{f}_{\mu,\varepsilon}(Y_i,\X_i) := \sup \{ f_{\tilde{\mu}}(Y_i,\X_i) : d(\tilde{\mu},\mu) \le \varepsilon \}
\end{equation*}
is a uniformly bounded random variable.
Moreover,
\begin{equation*}
\lim_{\varepsilon \to 0}\mathbb{E}[\log \bar{f}_{\mu,\varepsilon}(Y_i,\X_i)] = \mathbb{E}[\log f_{\mu}(Y_i,\X_i)].
\end{equation*}

\item [(iii)]
For any $\mu\in\bar{\mathcal{G}}$ with $\mu \ne G_*$, $f_\mu(Y_i, \X_i) \ne f_{G_*}(Y_i, \X_i)$ with positive probability.

\end{enumerate}
\end{lem}
\begin{proof}[Proof of Lemma~\ref{lem:key-to-consistency}]

Part~(i): Since $\underbar{c}I_T \le P(\delta)P(\delta)'\le \bar{c} I_T$ for all $\delta \in \mathcal{K}_\delta$ by Assumption~\ref{asm:DGP}(ii), we have
\begin{equation*}
C_1 \int \exp(-C_2\|Y_i - \X_i \hc \|^2) dG_*(\theta) \le f_{G_*}(Y_i,\X_i) \le C_3
\end{equation*}
for some absolute constants $C_1,C_2,C_3>0$.
Taking logarithms on both sides and applying Jensen's inequality, it follows that
\begin{equation*}
   \log(C_1) - C_2 \int \|Y_i - \X_i \hc\|^2 dG_*(\theta) \le \log f_{G_*}(Y_i,\X_i) \le \log(C_3),
\end{equation*}
which, in turn, implies
\begin{align*}
        |\log f_{G_*}(Y_i,\X_i)| &\le |\log(C_1)| + |\log(C_3)| + C_2 \int \|Y_i - \X_i \hc\|^2 dG_*(\theta).
\end{align*}
Using $\|Y_i - \X_i \hc\|^2 \le 2\|Y_i\|^2 + 2 \|X_i \hc\|^2$, it follows that
\begin{equation*}
\mathbb{E}[|\log f_{G_*}(Y_i,\X_i)|] \le |\log(C_1)| + |\log(C_3)| + 2C_2 \mathbb{E}[\|Y_i\|^2] + 2C_2\mathbb{E}_{G_*}[\|X_i \hc_i\|^2] < \infty,
\end{equation*}
since $\mathbb{E}_{G_*}[\|X_i \hc_i\|^2] \le \mathbb{E}_{G_*}[\|X_i \hc_i\|^2] + \mathbb{E}_{G_*}[\|P_i e_i\|^2] = \mathbb{E}[\|Y_i\|^2] < \infty$ by Assumption~\ref{asm:consistency}(i).

\bigskip

\noindent  Part~(ii):
Let $\varepsilon>0$ and $\mu \in \bar{\mathcal{G}}$ be given.
Since $(\bar{\mathcal{G}}, d)$ is a compact metric space, $\{\tilde{\mu} \in \bar{\mathcal{G}} : d(\tilde{\mu}, \mu) \le \varepsilon\}$ is closed and hence compact, and therefore admits a countable dense subset $\{ \tilde{\mu}_j\}_{j \ge 1}$.
We thus have
\begin{equation*}
\bar{f}_{\mu,\varepsilon}(Y_i,\X_i) = \sup_{j \in \mathbb{N}}  f_{\tilde{\mu}_j}(Y_i,\X_i)
\end{equation*}
by the continuity of $\mu \mapsto f_\mu$ established in Lemma~\ref{lem:likelihood-continuity}.
Since each $f_{\tilde{\mu}_j}(Y_i,\X_i)$ is a random variable, we find that $\bar{f}_{\mu,\varepsilon}(Y_i,\X_i)$ is as well.
The uniform boundedness of $\bar{f}_{\mu,\varepsilon}(Y_i,\X_i)$ follows from that of $f_\mu(Y, \X)$ across all $(Y,\X)$ and $\mu \in \psub$.
Finally, by Lemma~\ref{lem:likelihood-continuity}, $\bar{f}_{\mu,\varepsilon}(Y_i,\X_i)$ decreases to $f_\mu(Y_i,\X_i)$ almost surely as $\varepsilon \downarrow 0$.
Since $\log \bar{f}_{\mu,\varepsilon}(Y_i,\X_i)$ is bounded above uniformly in $\varepsilon$, the monotone convergence theorem yields $\lim_{\varepsilon \to 0}\mathbb{E}[\log \bar{f}_{\mu,\varepsilon}(Y_i,\X_i)] = \mathbb{E}[\log f_{\mu}(Y_i,\X_i)]$.

\bigskip

\noindent  Part~(iii): 
The case $\mu\in \mathcal{G}$ with $\mu \ne G_*$ is established in Theorem~\ref{thm:identification}.
Consider instead the case $\mu \in \bar{\mathcal{G}}\setminus \mathcal{G}$, so that $\mu(\Theta) < 1$ because $\mu$ is a subprobability measure on $\Theta$ that does not belong to $\mathcal{G} = \mathcal{P}(\Theta)$.
Suppose, for contradiction, $f_\mu(Y_i,\X_i) = f_{G_*}(Y_i, \X_i)$ a.s.
This implies
\begin{equation*}
\int_{\mathbb{R}^T} f_\mu(y, \X_i) dy = \int_{\mathbb{R}^T} \frac{f_\mu(y, \X_i)}{f_{G_*}(y,\X_i)} f_{G_*}(y,\X_i)dy = \mathbb{E}\left[ \left. \frac{f_\mu(Y_i,\X_i)}{f_{G_*}(Y_i, \X_i)} \right| \X_i \right]  = 1.
\end{equation*}
This leads to a contradiction, since by Fubini's theorem,
\begin{equation*}
    \int_{\mathbb{R}^T} f_\mu(y, \X_i) dy = \int_{\Theta} \left( \int_{\mathbb{R}^T} \ell(y \,|\, \X_i, \theta) dy\right) d\mu(\theta) = \mu(\Theta) < 1.
\end{equation*}

\end{proof}

\subsection{Regret Lemmas (Theorem~\ref{thm:EB-consistency}, Proposition~\ref{prop:EB-HIVDX})}

\begin{lem}
\label{lem:EB-unif-conv}
Let $\taupm(Y,\X) = \mathbb{E}_{G_*}[\Tau_i \,|\, (Y_i,\X_i) = (Y, \X)]$ and $\taueb(Y,\X) = \mathbb{E}_{\hat{G}}[\Tau_i \,|\, (Y_i,\X_i) = (Y, \X)]$ denote the oracle and EB decision rules for $\Tau_i$, respectively.
Let $\mathcal{R} = \{(Y,X) \in \mathbb{R}^T \times \SX : \|Y\| \le M,\, M^{-1} \I_{\dbeta} \le X'X \le M \I_{\dbeta}\}$ for a positive constant $M>1$.
Let Assumptions~\ref{asm:DGP}, \ref{asm:id-conditions}, \ref{asm:consistency}, and \ref{asm:EB} hold. Then, $\sup_{(Y,\X) \in \mathcal{R}}|\taueb(Y,X) - \taupm(Y,X)| \asto 0$.
\end{lem}

\begin{proof}[Proof of Lemma~\ref{lem:EB-unif-conv}]
We first observe that
\begin{equation*}
\taueb(Y,\X) = \frac{\int_{\Theta} \Tau(\theta) \ell(Y\,|\,\X,\theta) d\hat{G}(\theta)}{\int_{\Theta} \ell(Y\,|\,\X,\theta) d\hat{G}(\theta)} = \frac{\int_{\Theta} \Tau(\theta) \varphi_{Y,\X}(\theta) d\hat{G}(\theta)}{\int_{\Theta} \varphi_{Y,\X}(\theta) d\hat{G}(\theta)},
\end{equation*}
where $\varphi_{Y,\X}(\theta) = \ell(Y\,|\,\X,\theta)$ is defined in the proof of Lemma~\ref{lem:likelihood-continuity}.
In that proof, we have established that $\varphi_{Y,\X} \in C_0(\Theta)$ for each $Y \in \mathbb{R}^T$ and $\X \in \SX$ with $\rank(\X) = \dbeta$, and hence that
\begin{equation*}
\int_{\Theta} \varphi_{Y,\X}(\theta) d\hat{G}(\theta) \asto \int_{\Theta} \varphi_{Y,\X}(\theta) dG_*(\theta)
\end{equation*}
as $d(\hat{G}, G_*) \asto 0$.
To further establish that the above convergence occurs uniformly in $(Y,\X) \in \mathcal{R}$, we verify that the family $\{\varphi_{Y,\X} : (Y, \X) \in \mathcal{R}\}$ forms a relatively compact subset of $C_0(\Theta)$ with respect to the sup-norm.

Since $\X'(P(\delta)P(\delta)')^{-1}\X \ge (M\bar{c})^{-1} \I_{\dbeta}$ for all $\delta \in \mathcal{K}_\delta$, by Assumption~\ref{asm:DGP} and the construction of $\mathcal{R}$, there exist some positive constants $c_1$ and $c_2$ such that
\begin{equation*}
\sup_{(Y,\X) \in \mathcal{R}} |\varphi_{Y,\X}(\theta)| \le c_1 \exp \left( - c_2 \|\hc\|^2 \right).
\end{equation*}
This implies that there exists $D>0$ such that
\begin{equation}
\label{eq:vanish-at-infinity}
\sup_{\theta:\|\hc\| \ge D} \sup_{(Y,\X) \in \mathcal{R}} |\varphi_{Y,\X}(\theta)| < \epsilon/2.
\end{equation}
Consider the restriction of $\varphi_{Y,\X}$ to $B_D \times \mathcal{K}_\delta$, where $B_D := \{\hc:\|\hc\|\le D\}$.
Arzela-Ascoli theorem implies that the class $\{\left.\varphi_{Y,\X}\right|_{B_D\times \mathcal{K}_\delta} : (Y,\X) \in \mathcal{R}\}$ is relatively compact in $C(B_D \times \mathcal{K}_\delta)$ provided it is uniformly bounded and equi-continuous (\citeay{RoydenFitzpatrick2010}).
The uniform boundedness is obvious from the previous argument.
The uniform equicontinuity follows from the fact that
\begin{equation*}
\sup_{(Y,\X) \in \mathcal{R}}\left\|  \nabla_{\theta}\varphi_{Y,\X}(\theta) \right\| \le c_1' \exp \left( - c_2 \|\hc\|^2 \right) (\|\hc\|^2 + 1),
\end{equation*}
for some $c_1'>0$ and the same $c_2$ as before, where $\nabla_{\theta} = \partial/\partial \theta$.
Taking the supremum over $\theta$, this implies that $\sup_{ \theta \in B_D \times \mathcal{K}_{\delta}} \left( \sup_{(Y,\X) \in \mathcal{R}}\left\|  \nabla_{\theta}\varphi_{Y,\X}(\theta) \right\| \right)< \infty$.
Consequently, Arzela-Ascoli theorem implies that there exist $(Y_1,\X_1),\ldots, (Y_K, \X_K) \in \mathcal{R}$, $K < \infty$ such that, for all $(Y,\X) \in \mathcal{R}$,
\begin{equation*}
\min_{k=1,\ldots, K}\sup_{\theta \in B_D \times \mathcal{K}_\delta} |\varphi_{Y,X}(\theta) - \varphi_{Y_k,\X_k}(\theta)|    < \epsilon
\end{equation*}
from which, combined with \eqref{eq:vanish-at-infinity}, it follows
\begin{equation*}
\min_{k=1,\ldots, K}\sup_{\theta \in \Theta} |\varphi_{Y,X}(\theta) - \varphi_{Y_k,\X_k}(\theta)|    < \epsilon.
\end{equation*}
Therefore, $\{\varphi_{Y,\X}\}_{(Y,\X) \in \mathcal{R}}$ is relatively compact in $C_0(\Theta)$.
The same line of arguments establish that $\{\Tau \cdot \varphi_{Y,\X}\}_{(Y,\X) \in \mathcal{R}}$ is also relatively compact, whose details are omitted for brevity.

Since a relatively compact set admits a finite $\epsilon$-net for every $\epsilon>0$, by the standard approximation argument, we find that
\begin{equation*}
\sup_{(Y,\X) \in \mathcal{R}}\left|  \int_{\Theta} \varphi_{Y,\X}(\theta) d\hat{G}(\theta) -     \int_{\Theta} \varphi_{Y,\X}(\theta) dG_*(\theta)\right| \asto 0
\end{equation*}
and
\begin{equation*}
\sup_{(Y,\X) \in \mathcal{R}}\left|  \int_{\Theta} \Tau(\theta)\varphi_{Y,\X}(\theta) d\hat{G}(\theta) -     \int_{\Theta} \Tau(\theta)\varphi_{Y,\X}(\theta) dG_*(\theta)\right|   \asto 0. 
\end{equation*}
By the continuity of $f_{G_*}$ and the compactness of $\mathcal{R}$, we observe that
\begin{equation*}
\inf_{(Y,\X) \in \mathcal{R}}\int_{\Theta} \varphi_{Y,\X}(\theta) dG_*(\theta) = \inf_{(Y,\X) \in \mathcal{R}} f_{G_*}(Y,\X) \ge c > 0
\end{equation*}
for some $c>0$, and hence 
\begin{equation*}
\liminf_{N\to\infty}\inf_{(Y,\X) \in \mathcal{R}}\int_{\Theta} \varphi_{Y,\X}(\theta) d\hat{G}(\theta)  = \liminf_{N\to\infty}\inf_{(Y,\X) \in \mathcal{R}} f_{\hat{G}}(Y,\X) \ge c/2 > 0 \quad \text{a.s.}
\end{equation*}
Moreover, we have
\begin{equation*}
\sup_{(Y,\X) \in \mathcal{R}} \left(  \left| \int_{\Theta} \Tau(\theta) \varphi_{Y,\X}(\theta) dG_*(\theta) \right|  + \left| f_{G_*}(Y,\X) \right| \right) < \infty.
\end{equation*}
This implies that
\begin{align*}
\sup_{(Y,\X) \in \mathcal{R}}\left| \taueb(Y,\X) -\taupm(Y,\X) \right| &=  \sup_{(Y,\X) \in \mathcal{R}}\left|\frac{\int_{\Theta} \Tau(\theta) \varphi_{Y,\X}(\theta) d\hat{G}(\theta)}{f_{\hat{G}}(Y,\X)}-\frac{\int_{\Theta} \Tau(\theta) \varphi_{Y,\X}(\theta) dG_*(\theta)}{f_{G_*}(Y,\X)}\right| \\
&\le \sup_{(Y,\X) \in \mathcal{R}} \left| \frac{\int_{\Theta} \Tau(\theta) \varphi_{Y,\X}(\theta) d\hat{G}(\theta) -\int_{\Theta} \Tau(\theta) \varphi_{Y,\X}(\theta) dG_*(\theta)}{f_{\hat{G}}(Y,\X)} \right|\\
& \quad + \sup_{(Y,\X) \in \mathcal{R}} \left| \frac{\int_{\Theta} \Tau(\theta) \varphi_{Y,\X}(\theta) dG_*(\theta)}{f_{G_*}(Y,\X)}\frac{f_{\hat{G}}(Y,\X)-f_{G_*}(Y,\X)}{f_{\hat{G}}(Y,\X)} \right|\\
& \asto 0.
\end{align*}

\end{proof}

\begin{lem}
\label{lem:prior-mean-bound}
Assume the HIVDX model in \eqref{def:model:hivxd}. Let Assumption~\ref{asm:DGP} hold.
Further suppose that $|\bar{X}_{2,i}| \le M$ and $\sum_{t=1}^T (X_{2,it} - \bar{X}_{2,i})^2 \ge c$ for some constants $M, c>0$ and all $i=1,\ldots, N$.
Then, there exists an absolute constant $A < \infty$ depending only on $\bar{c}$, $\underbar{c}$, $M$, $c$, and $T$, such that
\begin{equation*}
\left( \int_{\Theta} \|\hc\|^p d \hat{G}_N(\theta) \right)^{1/p}\le A \left( \frac{1}{N} \sum_{i=1}^N \| Y_i\|^p\right)^{1/p} \quad \text{for all }p \ge 2,
\end{equation*}
where $\hat{G}_N$ denotes the \NPM\ obtained from the sample $(Y_i, \X_i)_{i=1}^N$.
\end{lem}

\begin{proof}[Proof of Lemma~\ref{lem:prior-mean-bound}]

By the first-order conditions for the \NPM, we have
\begin{equation*}
\frac{1}{N} \sum_{i=1}^N \frac{\ell(Y_i\,|\,\X_i,\theta)}{f_{\hat{G}}(Y_i,\X_i)} \le 1 \quad \text{for all } \theta \in \Theta,
\end{equation*}
and $\hat{G}\large(\{ \theta \in \Theta: N^{-1}\sum_{i=1}^N {\ell(Y_i\,|\,\X_i,\theta)}/f_{\hat{G}}(Y_i,\X_i) = 1\}\large)=1$.
This implies
\begin{equation*}
\frac{1}{N} \sum_{i=1}^N \frac{\ell(Y_i\,|\,\X_i,\theta)}{f_{\hat{G}}(Y_i,\X_i)} \frac{\partial \log \ell(Y_i\,|\,\X_i, \theta)}{\partial \hc}  = 0
\end{equation*}
for $\hat{G}$-a.e. $\theta$, and hence
\begin{equation*}
\frac{1}{N} \sum_{i=1}^N \frac{\ell(Y_i\,|\,\X_i,\theta)}{f_{\hat{G}}(Y_i,\X_i)} \X_i' (P(\delta)P(\delta)')^{-1}(Y_i - \X_i \hc) = 0.
\end{equation*}
Rearranging the equation above, we obtain
\begin{equation*}
\hc = \left( \frac{1}{N} \sum_{i=1}^N \frac{\ell(Y_i\,|\,\X_i,\theta)}{f_{\hat{G}}(Y_i,\X_i)} \X_i' (P(\delta)P(\delta)')^{-1} \X_i\right)^{-1}  \left( \frac{1}{N} \sum_{i=1}^N \frac{\ell(Y_i\,|\,\X_i,\theta)}{f_{\hat{G}}(Y_i,\X_i)} \X_i' (P(\delta)P(\delta)')^{-1} Y_i\right).  
\end{equation*}
for $\hat{G}$-a.e. $\theta$.
Regarded as a weighted regression of $P(\delta)^{-1}Y_i$ on $P(\delta)^{-1} X_i$, this yields the following inequality:
\begin{equation*}
    \hc' \left(  \frac{1}{N} \sum_{i=1}^N \frac{\ell(Y_i\,|\,\X_i,\theta)}{f_{\hat{G}}(Y_i,\X_i)} \X_i' (P(\delta)P(\delta)')^{-1} \X_i \right) \hc \le \frac{1}{N} \sum_{i=1}^N \frac{\ell(Y_i\,|\,\X_i,\theta)}{f_{\hat{G}}(Y_i,\X_i)} \|P(\delta)^{-1}Y_i\|^2.
\end{equation*}
By Assumption~\ref{asm:DGP}, we have $\bar{c}^{-1} \I_{T} \le (P(\delta)P(\delta)')^{-1} \le \underbar{c}^{-1} \I_{T}$.
Let $\mu_i = (1/T)\sum_{t=1}^T X_{2,it}$ and $s_i^2 = (1/T)\sum_{t=1}^T (X_{2,it} - \mu_i)^2$.
Then, the smallest eigenvalue of $\X_i'\X_i$ satisfies
\begin{equation*}
\lambda_{\min}(\X_i'\X_i) \ge \frac{ T s_i^2 }{1 + \mu_i^2 + s_i^2} \ge \underbar{\lambda} >0 \quad \text{for all }i=1,,\ldots,N,
\end{equation*}
where $\underbar{\lambda} = \underbar{\lambda}(M,c) = {c}/({1 + M^2 + c/T})$ is a positive constant.
This implies that
\begin{equation*}
\|\hc\|^2 = \|\hc\|^2 \frac{1}{N} \sum_{i=1}^N \frac{\ell(Y_i\,|\,\X_i,\theta)}{f_{\hat{G}}(Y_i,\X_i)} \le \frac{\bar{c}}{\underbar{c}\underbar{\lambda}} \frac{1}{N} \sum_{i=1}^N \frac{\ell(Y_i\,|\,\X_i,\theta)}{f_{\hat{G}}(Y_i,\X_i)} \|Y_i\|^2 = A^2 \frac{1}{N} \sum_{i=1}^N \frac{\ell(Y_i\,|\,\X_i,\theta)}{f_{\hat{G}}(Y_i,\X_i)} \|Y_i\|^2
\end{equation*}
where $A^2 = {\bar{c}}/(\underbar{c}\underbar{\lambda}(M,c))$.
By the Jensen's inequality, it follows that, for $p \ge 2$,
\begin{equation*}
\|\hc\|^{p} \le A^p \left( \frac{1}{N}\sum_{i=1}^N \frac{\ell(Y_i\,|\,\X_i,\theta)}{f_{\hat{G}}(Y_i,\X_i)} \|Y_i\|^2\right)^{p/2} \le A^p \frac{1}{N} \sum_{i=1}^N \frac{\ell(Y_i\,|\,\X_i,\theta)}{f_{\hat{G}}(Y_i,\X_i)} \|Y_i\|^{p}.
\end{equation*}
Integrating both sides with respect to $\hat{G}$ yields
\begin{equation*}
\int_{\Theta} \|\hc\|^p d\hat{G}(\theta) \le A^p \frac{1}{N} \sum_{i=1}^N \|Y_i\|^p.    
\end{equation*}
\end{proof}

\section{Computation Details}
\label{sec:computation}
Implementation of the HC model is summarized in Algorithm~\ref{alg:iteration}.

\begin{algorithm}[h]
\caption{\WFR\ algorithm for the HC model}
\label{alg:iteration}

\KwIn{
Dataset $(Y_i, \X_i)_{i=1}^N$; initial grid and weights $(\theta_0^j, \omega_0^j)_{j=1}^m$; step size $\eta \in (0,1]$; maximum iterations $\overline{n}$; tolerance parameter \(\tol\).
}
Set $n \leftarrow 0$.

\While{$n < \overline{n} \;\; \mathrm{and} \; \max_{j} \{ \frac{1}{N} \sum_{i=1}^N \pi_{ij}(\pmb{\theta}_n, \pmb{\omega}_n)/\omega_n^j - 1 \} > \tol$}{%

  \tcc{FR-step: reweighting}
  \For{$j = 1,\ldots,m$}{%
    Compute posterior weights
    $
      \pi_{ij}(\pmb\theta_n,\pmb\omega_n)
      = \frac{\ell(Y_i \mid \X_i,\theta_n^j)\,\omega_n^j}
             {\sum_{k=1}^m \ell(Y_i \mid \X_i,\theta_n^k)\,\omega_n^k},\quad i=1,\ldots,N.
    $
    Update
    $
      \omega_{n+1}^j
      = (1-\eta)\,\omega_n^j + \frac{\eta}{N}\sum_{i=1}^N \pi_{ij}(\pmb\theta_n,\pmb\omega_n).
    $
  }

  \tcc{W-step: transport}
  \For{$j = 1,\ldots,m$}{%
    Compute updated posteriors (with new weights)
    $
      \pi_{ij}(\pmb\theta_n,\pmb\omega_{n+1})
      = \frac{\ell(Y_i \mid \X_i,\theta_n^j)\,\omega_{n+1}^j}
             {\sum_{k=1}^m \ell(Y_i \mid \X_i,\theta_n^k)\,\omega_{n+1}^k},\quad i=1,\ldots,N.
    $
    
    Take a gradient step
    $
      \tilde{\theta}_{n+1}^j
      = \theta_n^j + \frac{\eta}{N}\sum_{i=1}^N
      \pi_{ij}(\pmb\theta_n,\pmb\omega_{n+1})\,
      \frac{\partial}{\partial \theta_n^j}\log \ell(Y_i \,|\, \X_i,\theta_n^j).
    $
    
    Project to the feasible set
    $
      \theta_{n+1}^j \leftarrow P_{\Theta}(\tilde{\theta}_{n+1}^j).
    $
  }

  Set $n \leftarrow n+1$.
}

\KwOut{Final grid and weights $(\theta_{\hat{n}}^j, \omega_{\hat{n}}^j)_{j=1}^m$, where $\hat{n}$ denotes the stopping index (the first $n$ satisfying the convergence criterion, or $\overline{n}$ if it is not met).}

\end{algorithm}

\subsection{The EM Algorithm}
\label{subsec:alg:compare}

The EM algorithm is often used to find a (local) solution to mixture problems of the form \eqref{eq:numerical-objective}.  
Let $D_{ij} = \ind{\theta_i = \theta^j}$ denote an unobserved indicator of unit $i$’s membership in component $j$.  
If the true $D_{ij}$ were known, the ``complete-data'' maximum likelihood problem would be
\begin{equation*}
\argmax_{\pmb \theta \in \Theta^m}\frac{1}{N} \sum_{i=1}^N \sum_{j=1}^m D_{ij}\log \ell(Y_i \,|\, \X_i, \theta^j) 
= \argmax_{\pmb \theta \in \Theta^m}\sum_{j=1}^m \sum_{i : D_{ij} = 1} \log \ell(Y_i \,|\, \X_i, \theta^j),
\end{equation*}
which can be optimized directly with respect to $\pmb \theta$.
However, this is infeasible because the true $D_{ij}$ are unobserved.
The EM algorithm circumvents this difficulty by replacing $D_{ij}$ with its conditional expectation,
\begin{equation*}
\mathbb{E}_{\pmb \theta, \pmb \omega}[D_{ij} \,|\, Y_i, \X_i] 
= \frac{\ell(Y_i \,|\, \X_i, \theta^{j}) \, \omega^{j}}{\sum_{k=1}^m \ell(Y_i \,|\, \X_i, \theta^{k}) \, \omega^{k}} = \pi_{ij}(\pmb \theta, \pmb \omega),
\end{equation*}
evaluated at the current iterate $(\pmb \theta, \pmb \omega)$.  
This step is referred to as the E-step.  
Let $(\pmb \theta_n^{\mathrm{EM}}, \pmb \omega_n^{\mathrm{EM}})$ denote the $n$th iterate of the support points and their weights.  
Given these values, the expected log-likelihood can be written as
\begin{equation*}
Q_N(\pmb \theta \,|\, \pmb \theta_n^{\mathrm{EM}}, \pmb \omega_n^{\mathrm{EM}}) 
:= \frac{1}{N} \sum_{j=1}^m \sum_{i=1}^N 
\pi_{ij}(\pmb \theta_n^{\mathrm{EM}}, \pmb \omega_n^{\mathrm{EM}}) 
\log \ell(Y_i \,|\, \X_i, \theta^j).
\end{equation*}
In the subsequent M-step, $(\pmb \theta_n^{\mathrm{EM}}, \pmb \omega_n^{\mathrm{EM}})$ are updated by solving
\begin{align*}
\theta_{n+1}^{j,\mathrm{EM}} 
&= \argmax_{\theta \in \Theta} \frac{1}{N} \sum_{i=1}^N 
\pi_{ij}(\pmb \theta_n^{\mathrm{EM}}, \pmb \omega_n^{\mathrm{EM}}) 
\log \ell(Y_i \,|\, \X_i, \theta), \\
\omega_{n+1}^{j,\mathrm{EM}} 
&= \frac{1}{N} \sum_{i=1}^N \pi_{ij}(\pmb \theta_n^{\mathrm{EM}}, \pmb \omega_n^{\mathrm{EM}}), 
\qquad j=1,\ldots,m.
\end{align*}
The E- and M-steps are iterated until convergence, $(\pmb \theta_n^{\mathrm{EM}}, \pmb \omega_n^{\mathrm{EM}}) \to (\hat{\pmb \theta}^{\mathrm{EM}}, \hat{\pmb \omega}^{\mathrm{EM}})$, yielding the EM estimator $\hat G^{\mathrm{EM}} = G^{\hat{\pmb \theta}^{\mathrm{EM}}, \hat{\pmb \omega}^{\mathrm{EM}}}$.

\section{\texorpdfstring{Failure of Assumption~\ref{asm:id-conditions}(ii) with}{Failure of Assumption 2(ii) with} $T=4$}\label{app:T4-failure}

Proposition~\ref{prop:HIVDX-id-T=4} does not cover the case in which $T = 4$ and the sequence of first differences $(\Delta X_{2,it})_{t=2}^4$ takes the symmetric form
\begin{equation*}
(\Delta X_{2,i2}, \Delta X_{2,i3}, \Delta X_{2,i4}) = (x, \tilde{x}, x),
\end{equation*}
with $|\tilde{x}| \ne |x|$.
In this case, Assumption~\ref{asm:id-conditions}(ii) can fail.
To see this, consider the case where $X_{2,it} = \mathbf{1}\{t \ge 3\}$ for all $i$:
\begin{align}\label{model:event-study}
Y_{it} = a_i + b_i \mathbf{1}\{t \ge 3\} + u_{it}, \quad t=1,\ldots, 4,
\end{align}
with $u_{it}$ following the structure in \eqref{def:model:hivxd:2}.
Here, the difference sequence is $(0, 1, 0)$.
To eliminate the random slope $b_i$ (which appears only in the difference $\Delta Y_{i3}$), we construct $\Q$ to select only the second and fourth differences:
\begin{equation*}
\Q = \begin{pmatrix}
    -1 & 0 \\
    1 & 0 \\
    0 & -1 \\
    0 & 1
\end{pmatrix}.
\end{equation*}
This yields $\Q' Y_i = (\Delta u_{i2}, \Delta u_{i4})' = (u_{i2}-u_{i1}, u_{i4}-u_{i3})'$.
The corresponding covariance matrix $\mathcal{V} := \Q'P(\delta)P(\delta)'\Q = \var(\Q' Y_i \mid \delta_i = \delta)$ is $2\times 2$ and imposes the following two restrictions on $\delta=(\sigma^2, \rho)$:
\begin{align*}
\mathcal{V}_{11} & =\mathcal{V}_{22} = \frac{2\sigma^2}{1+\rho},\\
\mathcal{V}_{12} &= \mathcal{V}_{21} = -\frac{\rho(1-\rho)\sigma^2}{1+\rho}.
\end{align*}
These relations are insufficient to uniquely determine $\delta$. For instance, the distinct parameter pairs $(1.04, 0.2)$ and $(1.56, 0.8)$ yield identical covariance matrices $\mathcal{V}$.

The intuition for this failure is straightforward.
Because the differenced regressor sequence $(0, 1, 0)$ is symmetric and the error process $u_{it}$ is stationary, the joint distribution of the utilized differences $(\Delta Y_{i2}, \Delta Y_{i4})$ is symmetric.
Specifically, $\mathcal{V}_{11} = \mathcal{V}_{22}$ provides the scale of the variance, while the correlation depends on the ratio $\mathcal{V}_{12}/\mathcal{V}_{11} = -\rho(1-\rho)/2$.
Since the function $f(\rho) = \rho(1-\rho)$ is not injective on $(-1, 1)$, solving for $\rho$ involves a quadratic equation with potentially two distinct valid roots, so that Assumption~\ref{asm:id-conditions}(ii) fails and Theorem~\ref{thm:identification} does not apply.

This failure highlights that sufficient time variation in the covariates is a prerequisite for our identification argument in short panels.
However, as established in Proposition~\ref{prop:HIVDX-id-T=5}, the additional variation afforded by a fifth time period breaks this symmetry and ensures that $\delta$ is uniquely determined.

{
\begin{figure}[!htbp]
	\caption{Pairwise scatter plots of $\hat \theta_i^{\mathrm{EB}}$}
	\label{fig:PSID:post:dist}
	{
    \begin{center}

    \begin{subfigure}[b]{0.45\textwidth}
        \centering
	\includegraphics[width=0.99\linewidth,keepaspectratio]{ab_eb.jpg} 
    \subcaption{Pair $(\hat a_i^{\mathrm{EB}}, \hat b_i^{\mathrm{EB}})$.}
    \end{subfigure}
    \begin{subfigure}[b]{0.45\textwidth}
        \centering
        \includegraphics[width=0.99\linewidth,keepaspectratio]{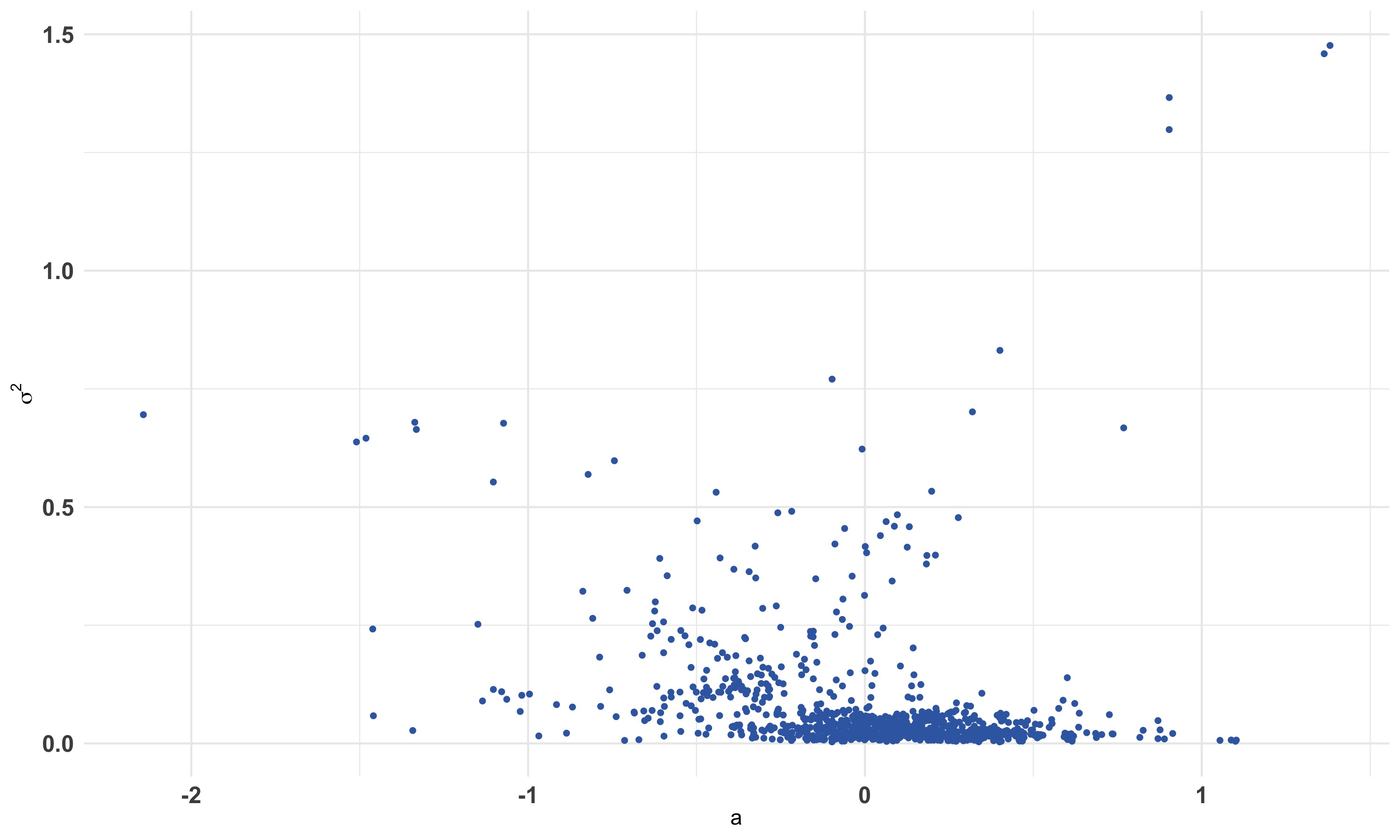}
        \subcaption{Pair $(\hat a_i^{\mathrm{EB}},\hat \sigma_i^{2,\mathrm{EB}})$.}
    \end{subfigure}

    \begin{subfigure}[b]{0.45\textwidth}
        \centering
        \includegraphics[width=0.99\linewidth,keepaspectratio]{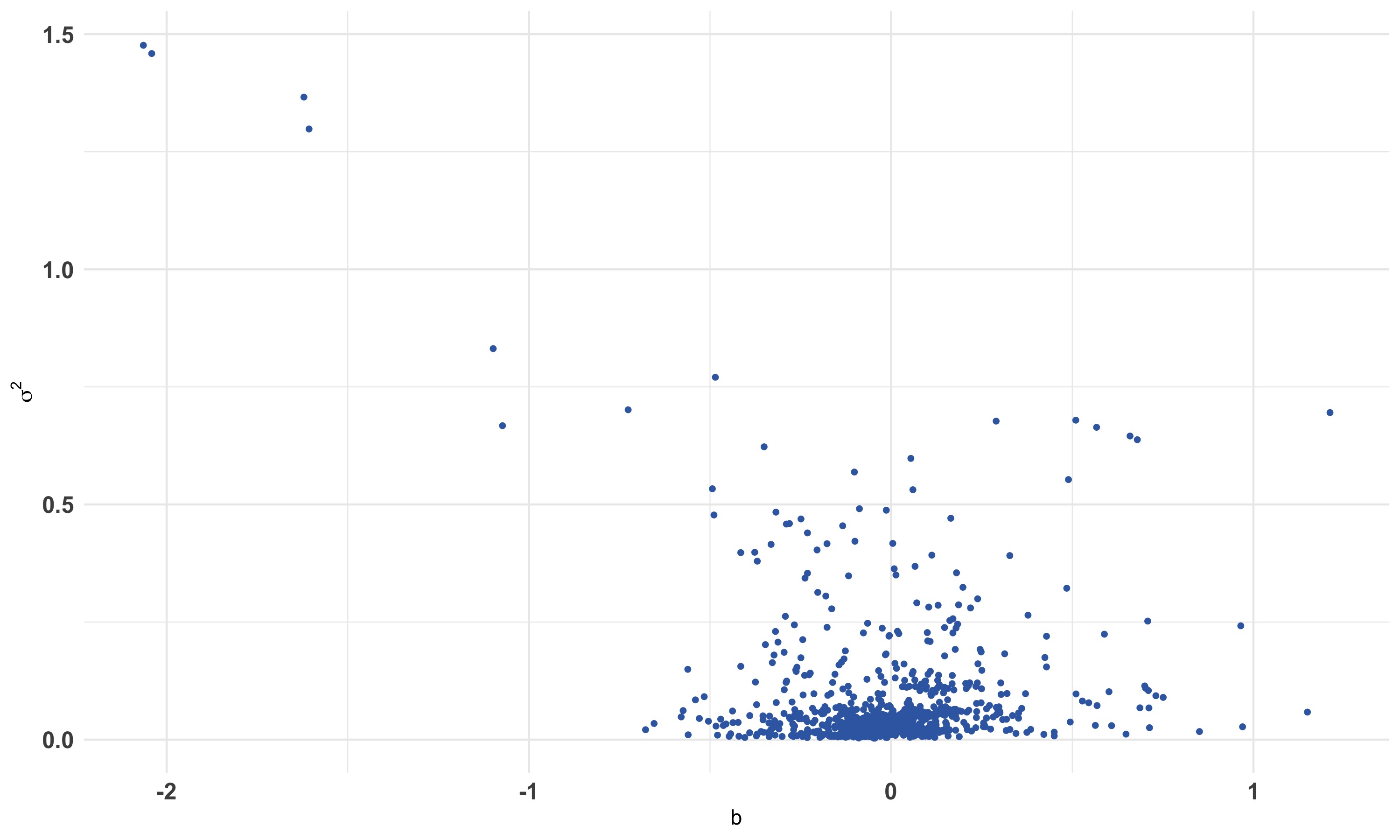}
        \subcaption{Pair $(\hat b_i^{\mathrm{EB}},\hat \sigma_i^{2,\mathrm{EB}})$.}
    \end{subfigure}
    \begin{subfigure}[b]{0.45\textwidth}
        \centering
        \includegraphics[width=0.99\linewidth,keepaspectratio]{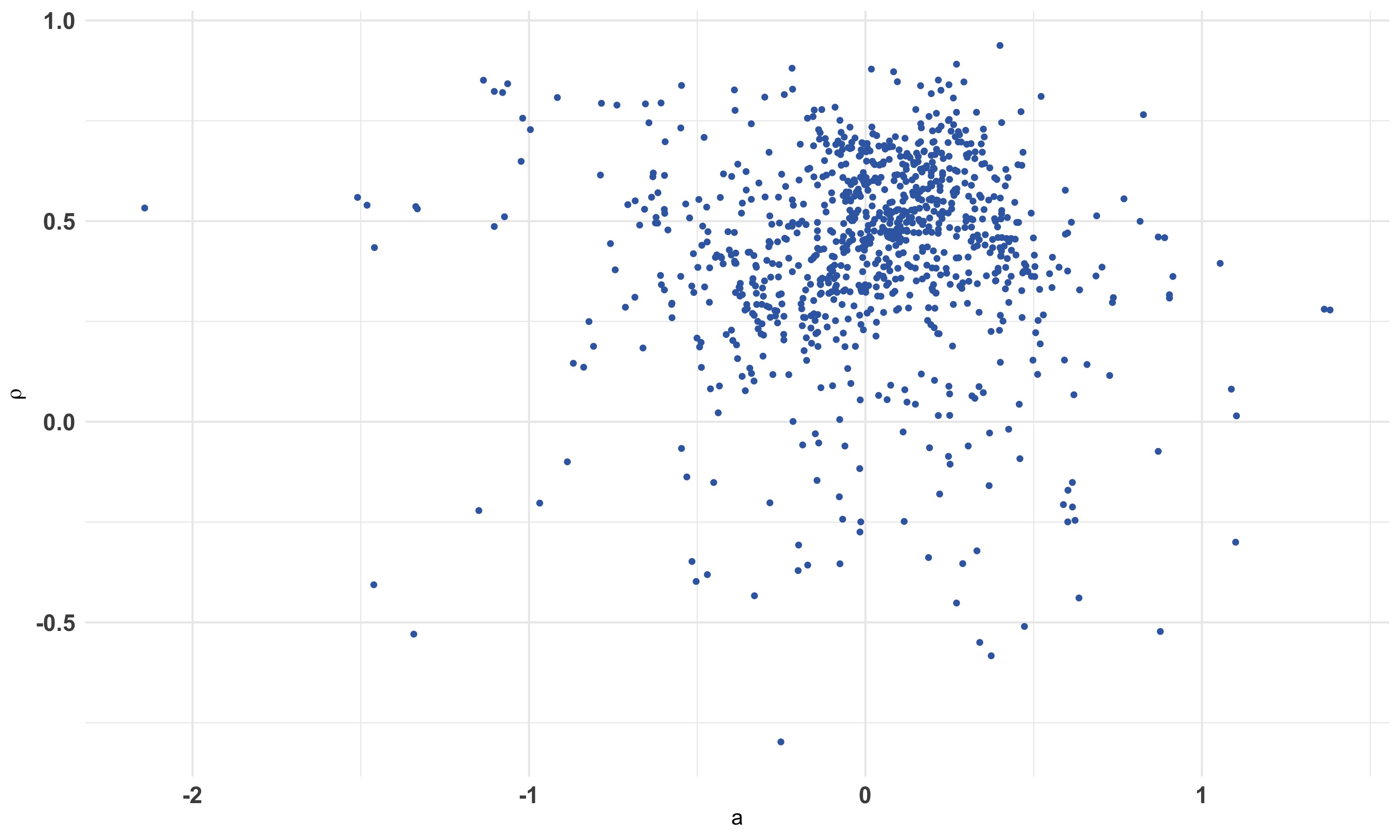}
        \subcaption{Pair $(\hat a_i^{\mathrm{EB}},\hat \rho_i^{\mathrm{EB}})$.}
    \end{subfigure}

    \begin{subfigure}[b]{0.45\textwidth}
        \centering
        \includegraphics[width=0.99\linewidth,keepaspectratio]{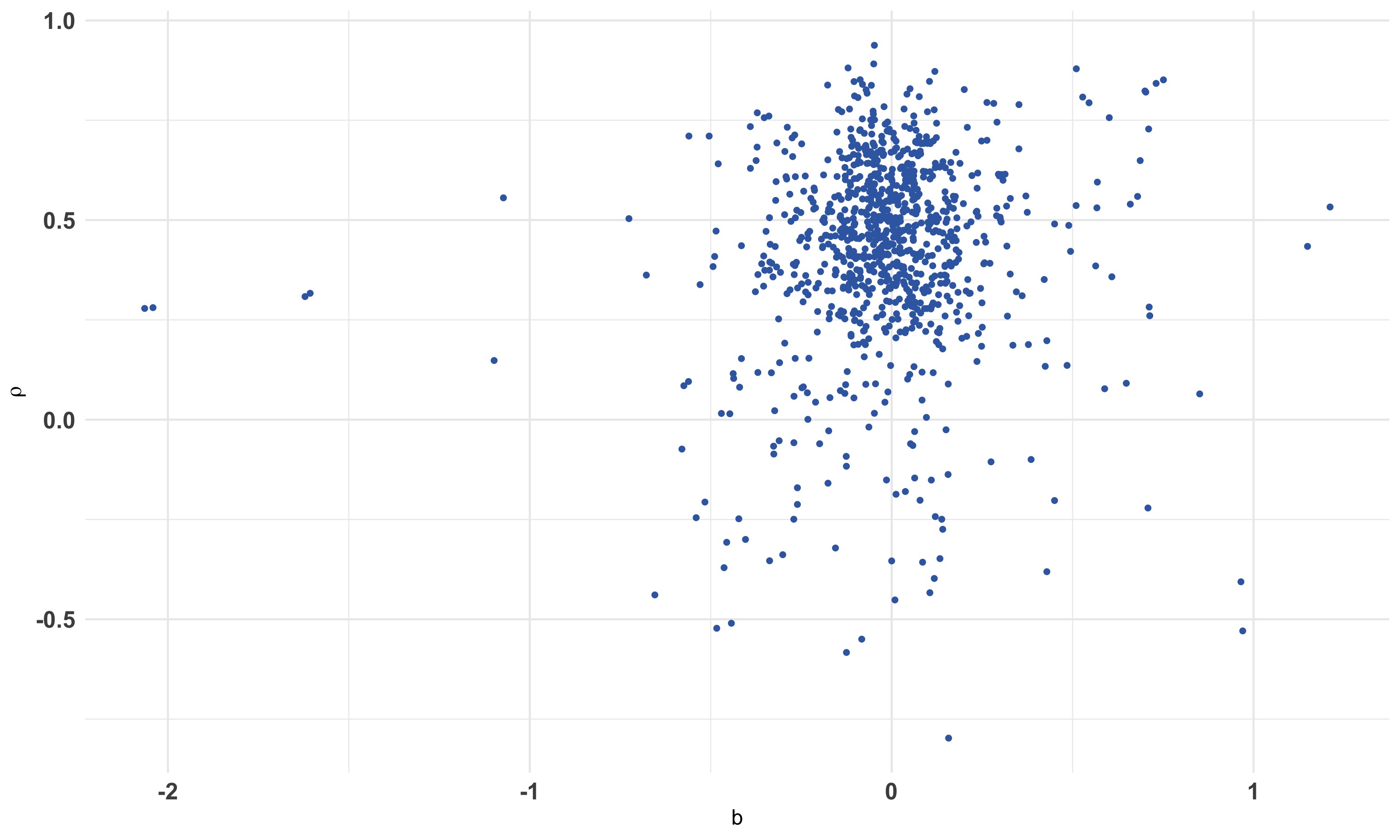}
        \subcaption{Pair $(\hat b_i^{\mathrm{EB}},\hat \rho_i^{\mathrm{EB}})$.}
    \end{subfigure}
    \begin{subfigure}[b]{0.45\textwidth}
        \centering
        \includegraphics[width=0.99\linewidth,keepaspectratio]{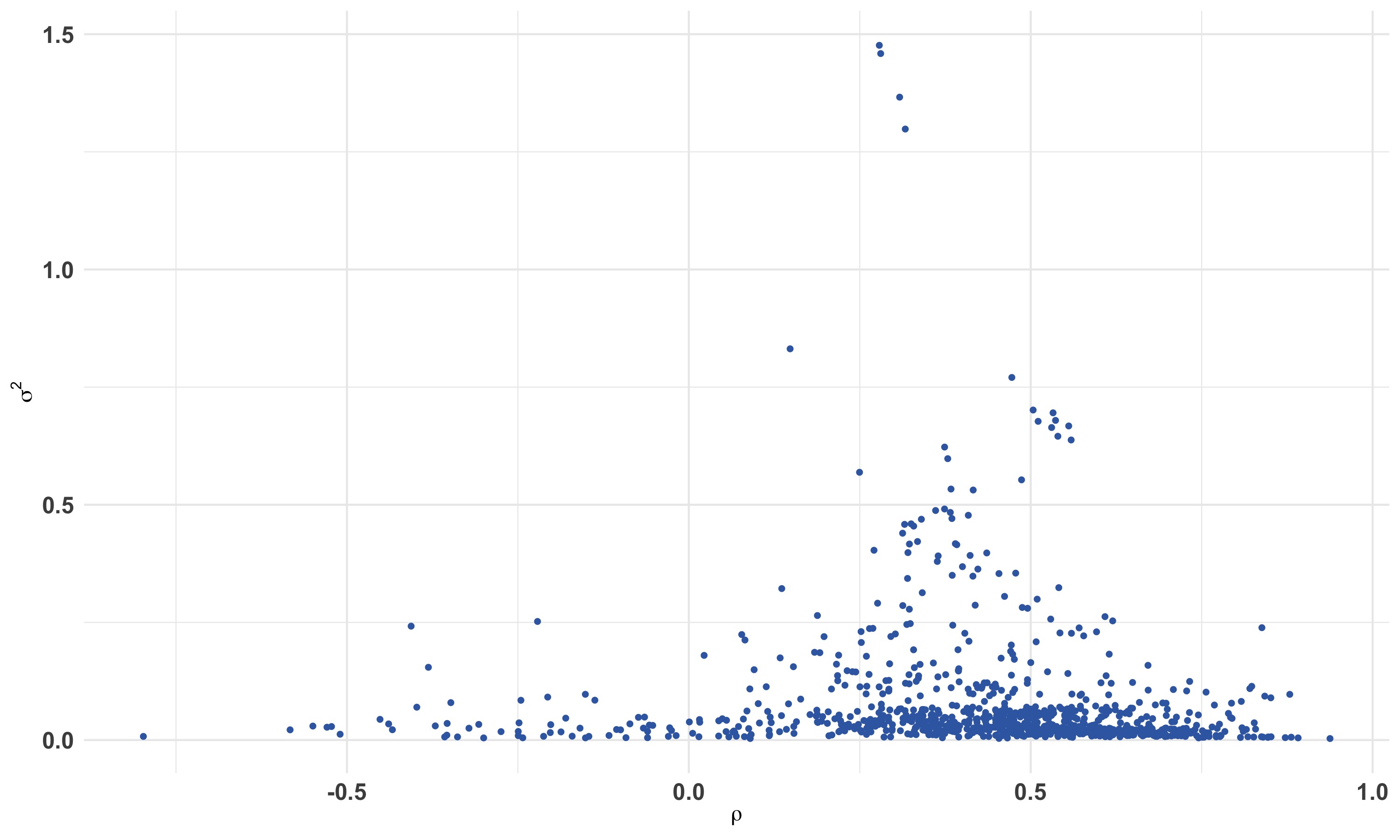}
        \subcaption{Pair $(\hat \rho_i^{\mathrm{EB}},\hat \sigma_i^{2,\mathrm{EB}})$.}
    \end{subfigure}
    \end{center}

    {\begin{center}
        \parbox{0.95\textwidth}{\footnotesize Note: Each panel shows the scatter plot of a pair of components of $(\hat{\theta}_i^{\mathrm{EB}})_{i=1}^N$.
        }
    \end{center}
    }
    }
\end{figure}
}

\section{Additional Empirical Results}\label{appendix:pairwise:EB}

Figure~\ref{fig:PSID:post:dist}  shows the pairwise scatter plots of the EB estimates for \(\theta_i\) computed from the sample of $N=938$ individuals, supporting the heterogeneity pattern reported in Table~\ref{tab:PSID:prior:moments}.\footnote{To be precise, by the law of total variance, $\var(\theta_i) = \var(\hat{\theta}_i^{\mathrm{EB}}) + \mathbb{E}[\var(\theta_i | (Y_{it},X_{2,it})_{t=1}^{T_i})]$, where the second term represents the expected posterior variance, that is, the mean squared error of the EB estimator.} Panel~(a) illustrates that individual intercepts ($\hat{a}_i^{\mathrm{EB}}$) and slopes ($\hat{b}_i^{\mathrm{EB}}$) are negatively correlated. Panels~(b) and (c) also reveal negative associations between $\hat a_i^{\mathrm{EB}}$ and $\hat \sigma_i^{2, \mathrm{EB}}$, and between $\hat b_i^{\mathrm{EB}}$ and $\hat \sigma_i^{2, \mathrm{EB}}$, respectively.  In the remaining panels, individual persistence levels appear only weakly correlated with the other individual-specific parameters. 

\end{document}

%% file: commands.tex
\newcommand{\var}{\mathbb{V}}
\newcommand{\cov}{\operatorname{Cov}}
\newcommand{\argmax}{\operatornamewithlimits{arg\hspace{0.1em} max}}
\newcommand{\argmin}{\operatornamewithlimits{arg\hspace{0.1em} min}}

\newcommand{\rank}{\operatorname{rank}}

\newcommand{\I}{{I}}
\newcommand{\ind}[1]{{\mathbf{1}\large[{#1}\large]}}
\newcommand{\iid}{\overset{\mathrm{i.i.d.}}\sim}
\newcommand{\tol}{\mathtt{tol}}
\newcommand{\xmat}{{\pmb{X}}}

\newcommand{\zmat}{\pmb{z}}

\newcommand{\X}{X}

\newcommand{\hc}{\beta}

\newcommand{\Q}{\pmb{M}}

\newcommand{\ii}{\mathrm{i}}
\newcommand{\supp}{\operatorname{supp}}

\newcommand{\thetamle}{\hat {\theta}^{\operatorname{MLE}}}

\newcommand{\betamle}{\hat {\beta}^{\operatorname{MLE}}}

\newcommand{\betafeb}{\hat {\beta}^{\operatorname{F-EB}}}

\makeatletter
\renewcommand{\underbar}[1]{\underaccent{\bar}{#1}}
\newcommand{\NPM}{\mbox{NP\hspace{0.05em}MLE}}

\newcommand{\WFR}{{Wasserstein-Fisher-Rao}}
\newcommand{\wfr}{\mathrm{WFR}}

\newcommand{\dbeta}{{d_{\beta}}}
\newcommand{\ddelta}{{d_{\delta}}}

\newcommand{\dwhole}{{d_{\theta}}}
\newcommand{\dgamma}{{d_{\gamma}}}

\newcommand{\asto}{\overset{\operatorname{a.s.}}{\longrightarrow}}

\newcommand{\K}{\mathcal{K}}
\newcommand{\psub}{\mathcal{P}_{\operatorname{sub}}}

\newcommand{\SX}{\supp(\X_i)}

\newcommand{\oo}{o}

\newcommand{\Tau}{\tau}

\newcommand{\oracle}{*}
\newcommand{\thetapm}{{\theta}^{\oracle}}
\newcommand{\betapm}{{\beta}^{\oracle}}
\newcommand{\taupm}{{\Tau}^{\oracle}}

\newcommand{\EB}{{\mathrm{EB}}}
\newcommand{\taueb}{\hat{\Tau}^{\mathrm{EB}}}
\newcommand{\sigmaeb}{\hat{\sigma}^{2,\mathrm{EB}}}
\newcommand{\rhoeb}{\hat{\rho}^{\mathrm{EB}}}

%% file: Table_HIVDX.tex
\begin{tabular}{C{0.125\textwidth}
S[table-format=1.4, table-column-width=0.125\textwidth]
S[table-format=1.4, table-column-width=0.125\textwidth]
S[table-format=1.4, table-column-width=0.125\textwidth]
S[table-format=1.4, table-column-width=0.125\textwidth]
}
\toprule
\multirow[c]{2.5}{*}{$\mathbb{V}_{\hat G}(\theta_i)$} & \multicolumn{4}{C{0.5\textwidth}}{$\theta_i$} \\ 
\cmidrule(lr){2-5}
& \multicolumn{1}{c}{$a_i$} & \multicolumn{1}{c}{$b_i$} & \multicolumn{1}{c}{$\sigma_i^2$} & \multicolumn{1}{c}{$\rho_i$} \\ 
\midrule
$a_i$       & 0.2052  & -0.1003 & -0.0105 & -0.0039 \\[0.8em]
$b_i$       & -0.1003 & 0.1007  & -0.0104 & 0.0074  \\[0.8em]
$\sigma_i^2$& -0.0105 & -0.0104 & 0.0229  & -0.0013 \\[0.8em]
$\rho_i$    & -0.0039 & 0.0074  & -0.0013 & 0.1052  \\[0.3em]
\cmidrule(lr){1-5}
$\mathbb{E}_{\hat G}[\theta_i]$ & 0.0134 & -0.0166 & 0.0787 & 0.4219 \\[0.25em]
\bottomrule
\end{tabular}